\UseRawInputEncoding 
\documentclass[a4paper,twocolumn,8pt]{article}

\pdfoutput=1

\usepackage{xcolor}
\definecolor{qcblue}{rgb}{0.098, 0.282, 0.502} 
\definecolor{qcgreen}{rgb}{0.278, 0.604, 0.161}  

\usepackage[colorlinks=true, citecolor=qcgreen, linkcolor=orange, urlcolor=qcblue, pagebackref]{hyperref}
\usepackage[left=1.6cm,right=1.6cm,top=2cm,bottom=2cm, columnsep=0.7cm]{geometry} 
\usepackage[english]{babel}
\usepackage[T1]{fontenc}
\usepackage{amsmath}
\usepackage{array}
\usepackage{float}
\usepackage{amsthm}
\usepackage{physics}
\usepackage{tikz}
\usepackage{booktabs}
\usepackage{lipsum}
\usepackage{tabularx}
\usepackage{amsmath}
\usepackage{amssymb}
\usepackage{algorithm}
\usepackage{algpseudocode}
\usepackage{multirow}
\usepackage{float}
\usepackage{soul}
\usepackage{dsfont}
\usepackage{xfrac}
\usepackage{backref}
\usepackage{subcaption}

\usepackage{titlesec}
\usepackage{svg}
\usepackage[numbers, sort&compress]{natbib}

\newtheorem{proposition}{Proposition}
\newtheorem{lemma}{Lemma}

\newtheorem*{propositionunnum}{Proposition}

\newtheorem*{lemmaunnum}{Lemma}

\newtheorem{definition}{Definition}

\newtheorem{corollary}{Corollary}
\newcommand{\corrref}[1]{Corollary~\ref{#1}}
\newcommand{\defref}[1]{Definition~\ref{#1}}
\newcommand{\secref}[1]{\nameref{#1}}

\newcommand{\lemref}[1]{Lemma~\ref{#1}}
\newcommand{\propref}[1]{Proposition~\ref{#1}}

\newcommand{\tabref}[1]{Table~\ref{#1}}
\renewcommand{\eqref}[1]{equation~(\ref{#1})}
\newcommand{\figref}[1]{Fig.~\ref{#1}}
\newcommand{\appref}[1]{Supp.~Mat.~\ref{#1}}

\titleformat{\section}[block]{\normalfont\Large\bfseries}{}{0em}{}
\titleformat{\subsection}[block]{\normalfont\large\bfseries}{}{0em}{}
\titleformat{\subsubsection}[block]{\normalfont\normalsize\bfseries}{}{0em}{}

\def\orcid#1{\kern -0.4em\href{https://orcid.org/#1}{\includegraphics[keepaspectratio,width=0.7em]{orcid_logo.pdf}}} 

\usepackage{tocbibind}
\usepackage{etoolbox} 
\usepackage{tocloft} 
\usepackage{minitoc} 
\usepackage{authblk}

\usepackage{lineno}

\title{Training-efficient density quantum machine learning}

\author[1, *]{Brian Coyle}
\author[1, 2]{Snehal Raj}
\author[1, 3]{Natansh Mathur}
\author[1]{El Amine Cherrat}
\author[1, 4]{Nishant Jain}
\author[1]{Skander Kazdaghli}
\author[1, 3]{Iordanis Kerenidis}
\affil[1]{QC Ware, Palo Alto, USA and Paris France.}
\affil[2]{LIP6, CNRS, Sorbonne Universit\'e, 4 Place Jussieu, 75005 Paris, France}
\affil[3]{IRIF, CNRS - University of Paris, France.}
\affil[4]{Indian Institute of Technology, Roorkee, India.}
\affil[*]{Corresponding author, brian.coyle@qcware.com}

\begin{document}
\doparttoc 
\faketableofcontents 

\date{}

\twocolumn[
  \begin{@twocolumnfalse}
    \maketitle
\begin{abstract}
\textbf{
Quantum machine learning (QML) requires powerful, flexible and efficiently trainable models to be successful in solving challenging problems. We introduce density quantum neural networks, a model family that prepares mixtures of trainable unitaries, with a distributional constraint over coefficients. This framework balances expressivity and efficient trainability, especially on quantum hardware. For expressivity, the Hastings-Campbell Mixing lemma converts benefits from linear combination of unitaries into density models with similar performance guarantees but shallower circuits. For trainability, commuting-generator circuits enable density model construction with efficiently extractable gradients. The framework connects to various facets of QML including post-variational and measurement-based learning. In classical settings, density models naturally integrate the mixture of experts formalism, and offer natural overfitting mitigation. The framework is versatile - we uplift several quantum models into density versions to improve model performance, or trainability, or both. These include Hamming weight-preserving and equivariant models, among others. Extensive numerical experiments validate our findings.
}
\end{abstract}
  \end{@twocolumnfalse}
]

\section{Introduction}

\begin{figure*}[!ht]
\centering
    \includegraphics[width=\linewidth]{Fig1_LCU_Density_Det_Density_Rand.pdf}
    \caption{
    \textsf{
    \textbf{Density quantum neural networks.} \\
    a) Linear combination of unitaries quantum neural networks (LCU QNNs) preparing the state $\sum_k\alpha_k U_k(\boldsymbol{\theta}_k)\ket{\boldsymbol{x}}$ via postselection on an ancilla register $\mathcal{A}$ which prepares the distribution $\boldsymbol{\alpha}$. b) shows corresponding density quantum neural network, implemented deterministically to prepare the state $\rho(\boldsymbol{\theta}, \boldsymbol{\alpha}, \boldsymbol{x})$. Finally, the instantiation of the density QNN state via randomisation is shown in (c) where sub-unitary, $U_k(\boldsymbol{\theta}_k)$ is only prepared with the probability $\alpha_k$ without the need for the multi controlled deep circuits and ancilla qubits. The deterministic density QNN, (b) is required if one wishes to make a true comparison of these networks to the dropout mechanism. From the Mixing lemma, the randomised version, (c), can distill the performance benefits of the more powerful LCU QNN, (a), into very short depth circuits. The probability loaders, $\mathsf{Load}\left(\sqrt{\boldsymbol{\alpha}}\right)$ are assumed to be unary data loaders which act on $K$ qubits within the register, $\mathcal{A}$ and have depth $\log(K)$~\cite{johri_nearest_2021}. One could also use binary \texttt{Prepare} and \texttt{Select} circuits acting on $\log(K)$ qubits as is more standard in LCU literature. The resulting functions from each network, $f(\boldsymbol{\theta}, \boldsymbol{\alpha}, \boldsymbol{x})$ result from the measurement of an observable, $\mathcal{O}$, via $f(\boldsymbol{\theta}, \boldsymbol{\alpha}, \boldsymbol{x}) = \Tr(\mathcal{O}\sigma(\boldsymbol{\theta}, \boldsymbol{\alpha}, \boldsymbol{x}))$, where $\sigma$ is the output state from each circuit. $V(\boldsymbol{x})$ is the $n$-qubit data loader acting on register, $\mathcal{B}$.
    }
    }
    \label{fig:density_qnn_mainfig}
\end{figure*}


Modern deep learning owes much of its success to the existence to three factors. The first is increasing resource availability, meaning data and compute power. Secondly, families of powerful and expressive models, such as multilayer perceptrons (MLP)~\cite{rosenblatt_perceptron_1958} convolutional and recurrent neural networks (CNN/RNNs)~\cite{lecun_deep_2015}, and multi-head attention mechanisms ~\cite{bahdanau_neural_2016,vaswani_attention_2017}. With these model primitives, one may build composites such as transformer, diffusion and mixer models which ultimately lead to specific state of the art models such as AlphaZero~\cite{silver_mastering_2017}, GPT-3~\cite{brown_language_2020} or Dalle~\cite{ramesh_zero-shot_2021}. Particularly for these vast models containing billions of parameters, and the enormous quantity of data required, an efficient training algorithm is essential. A cornerstone of many such protocols is the \emph{backpropagation} algorithm~\cite{rumelhart_learning_1986}, and its variants, which enables the computation of gradients throughout the entirety of the network, with minimal overhead beyond the network evaluation itself. 


It is reasonable to assume a similar trajectory for quantum machine learning. While rapid progress in quantum error correction~\cite{sivak_real-time_2023, acharya_quantum_2024, silva_demonstration_2024} is increasing the number and quality of effective qubits (compute power), quantum processing units (QPUs) will still likely remain significantly depth-limited for the foreseeable future. On the model side, \emph{quantum} neural networks (QNNs) are typically (but not exclusively) constructed from parametrised quantum circuits (PQCs)~\cite{benedetti_parameterized_2019, bharti_noisy_2022, cerezo_variational_2021, cerezo_challenges_2022}. However, in many cases these lack task-specific features and unfortunately do not \emph{generally} possess an efficient scaling for training in line with classical backpropagation~\cite{abbas_quantum_2023}. Therefore, it is essentially to develop quantum models which can be associated to specific interpretations (for example the convolutional operation on images), and which are efficient to train. The popular \emph{parameter-shift} rule for QNNs~\cite{mitarai_quantum_2018, crooks_gradients_2019, vidal_calculus_2018, schuld_evaluating_2019, sweke_stochastic_2020, kyriienko_generalized_2021}, extracts analytic gradients (i.e., not relying on approximate finite differences), but even the simplest instance of the rule, requires $\mathcal{O}(N)$ gradient circuits to be evaluated for $N$ parameters. To put this in perspective, it was estimated in Ref.~\cite{abbas_quantum_2023} that, if one is allowed only a single day of computation, the parameter-shift rule can only compute gradients on $n\sim 100$ qubit trainable circuits with \emph{only} $\sim 9000$ parameters, assuming reasonable quantum clock speeds. This also does not account for various other problem specific scalings, such as the data which needs to be iterated over for \emph{each} training iteration, or other obstacles such as barren plateaus~\cite{mcclean_barren_2018}. Scaling current quantum training approaches towards the size of billion or trillion-parameter deep neural networks, which have been so successful in the modern era, clearly will not be feasible with such methods. Additionally, since we are arguably in the boundary between the NISQ and ISQ eras, models should use circuits which are as compact, yet expressive as possible. On the other hand, they should also be complex enough to avoid classical simulation, surrogation or dequantisation~\cite{landman_classically_2023, rudolph_classical_2023, bermejo_quantum_2024} but not so complex to admit barren plateaus~\cite{cerezo_does_2024}. Clearly, satisfying all of these constraints is challenging task.

To partially address some of these challenges, in this work we introduce a framework of models dubbed \emph{density} quantum neural networks. Our primary aim is to showcase how these models add another dimension to the landscape of quantum learning models, giving practitioners a new toolkit to experiment with when tackling the above questions. Through the text, we demonstrate how one may construct density QNNs may be constructed which are more trainable, or more expressive than their pure-state counterparts. Our results are laid out as follows. First, we introduce the general form of the density framework, before discussing comparisons and relationships to other QML model families/frameworks in the literature. Then, we propose two methods of preparing such models on a quantum computer. Next, we prove our primary theoretical results - firstly relating to the gradient query complexity of such models, and secondly discussing the connection to non-unitary quantum machine learning via the Mixing lemma from randomised compiling. We then discuss two proposed connections between density networks and mechanisms in the classical machine learning literature. First, there has been suggestion in the literature that the density networks as we propose them may be a quantum-native analogue of the \emph{dropout} mechanism. We propose separate training and inference phases for density QNNs to bring this comparison closer to reality, but find it still lacking as a valid comparison. Secondly, we demonstrate a strong realisation of density networks within the \emph{mixture of experts} (MoE) framework from classical machine learning - density QNNs can be viewed as a `quantum mixture of experts'. Finally, we provide numerical results to demonstrate the flexibility of the model to improve performance, or improve trainability (or both). We test several QNN architectures on synthetic translation-invariant data, and Hamming weight preserving architectures on the MNIST image classification task. Finally, we show numerically how, in some capacity, density QNNs may prevent data overfitting using data reuploading as an example, despite not functioning as a true dropout mechanism.

\section{Results} \label{sec:results}

\subsection{Density quantum neural networks} \label{sec:density_qnn}
To begin, we explicitly define the framework (see~Supp.~Mat. A for a discussion) of \emph{density} quantum neural networks (density QNNs) as follows:

\begin{equation} \label{eqn:density_qnns}
    \rho(\boldsymbol{\theta}, \boldsymbol{\alpha}, \boldsymbol{x}) := \sum_{k=1}^K \alpha_k U_k(\boldsymbol{\theta}_k)\rho(\boldsymbol{x})U^\dagger_k(\boldsymbol{\theta}_k)
\end{equation}
$\rho(\boldsymbol{x})$ is a data encoded initial state, which is usually assumed to be prepared via a `data-loader' unitary, $\rho(\boldsymbol{x}) = \ketbra{\boldsymbol{x}}{\boldsymbol{x}}, \ket{\boldsymbol{x}} := V(\boldsymbol{x})\ket{0}^{\otimes n}$, a collection of sub-unitaries $\{U_{k}\}_{k=1}^K$, and a distribution, $\{\alpha_k\}_{k=1}^K$, which may depend on $\boldsymbol{x}$.

For now, we treat the density state above as an abstraction and later in the text we will discuss methods to prepare the state practically and actually use the model. The preparation method will have relevance for the different applications and connections to other paradigms. Once we have chosen a state preparation method for \eqref{eqn:density_qnns}, we must choose particular specifications for the sub-unitaries. In some cases, we may recast efficiently trainable models/frameworks within the density formalism to increase their expressibility. In others, we use the framework to improve the overall inference speed of models. In this work, we assume that the sub-unitary circuit structures, once chosen, are fixed, and the only trainability arises from the parameters, $\{\boldsymbol{\theta}_k\}_{k=1}^K$ therein, as well as the coefficients, $\{\alpha_k\}_{k=1}^K$. In other words, we do not incorporate variable structure circuits learned for example via quantum architecture search.

As a generalisation, one may consider adding a data dependence into the sub-unitary \emph{coefficients}, $\boldsymbol{\alpha} \rightarrow \boldsymbol{\alpha}(\boldsymbol{x})$, while retaining the distributional requirement for all $\boldsymbol{x}$, $\sum_k\alpha(\boldsymbol{x})_k=1$. This gives us the more general family of density QNN states:
\begin{equation} \label{eqn:density_qnn_with_data_alphas}
    \rho_{\mathsf{D}}(\boldsymbol{\theta}, \boldsymbol{\alpha}, \boldsymbol{x}) = \sum_{k=1}^K \alpha_k(\boldsymbol{x}) U_k(\boldsymbol{\theta}_k)\rho(\boldsymbol{x})U^\dagger_k(\boldsymbol{\theta}_k)
\end{equation}

In the QML world, overly dense or expressive single unitary models are known to have problems related to trainability via barren plateaus~\cite{holmes_connecting_2022}. Density QNNs and related frameworks may be a useful direction to retain highly parameterised models but via a combination of smaller, trainable models. We will demonstrate this through several examples in the remainder of the text. Before doing so, in the next section, we want to appropriately cast density QNNs within the current spectrum of quantum machine learning models.

\subsection{Connection to other QML frameworks} \label{ssec:other_qml_framework_connection}

Before proceeding, we first discuss the connection to other popular QML frameworks. For supervised learning purposes, each term in the density state, $U_k(\boldsymbol{\theta}_k)\rho(\boldsymbol{x})U^\dagger_k(\boldsymbol{\theta}_k)$ is expressive enough by itself to capture most basic models in the literature. This is due to the common model definition as $f(\boldsymbol{\theta, \boldsymbol{x}}) := \Tr(\mathcal{O}U(\boldsymbol{\theta})\rho(\boldsymbol{x})U^\dagger(\boldsymbol{\theta})) = \Tr(\mathcal{O}(\boldsymbol{\theta})\rho(\boldsymbol{x}))$ for some observable, $\mathcal{O}$, i.e. the overlap between a parameterised Hermitian observable and a data-dependent state. This unifies many paradigms in quantum machine learning literature such as kernel methods~\cite{schuld_supervised_2021} and data reuploading models via gate teleportation~\cite{jerbi_quantum_2023}.
Due to the linearity of the quantum mechanics, we can write it also in this form by inserting~\eqref{eqn:density_qnn_with_data_alphas} into the function evaluation:
\begin{multline} \label{eqn:density_qnn_with_data_as_linear_model}
    f_{\mathsf{D}}(\{\boldsymbol{\theta}, \boldsymbol{\alpha}\}, \boldsymbol{x}) = \Tr\left(\mathcal{O}(\boldsymbol{\theta}, \boldsymbol{\alpha}, \boldsymbol{x})\rho(\boldsymbol{x})\right),\\ \mathcal{O}(\boldsymbol{\theta}, \boldsymbol{\alpha}, \boldsymbol{x}) := \sum_{k=1}^K \alpha_k(\boldsymbol{x}) U^\dagger_k(\boldsymbol{\theta}_k) \mathcal{O}U_k(\boldsymbol{\theta}_k)
\end{multline}

Removing this data-dependence from the coefficients simply removes the data-dependence from the observable, $\mathcal{O}(\boldsymbol{\theta}, \boldsymbol{\alpha}, \boldsymbol{x}) \rightarrow \mathcal{O}(\boldsymbol{\theta}, \boldsymbol{\alpha})$. Finally, selecting the sub-unitaries to be identical and equal to the data loading unitary, with $K$ equal to the size of the training data leads to observable $\mathcal{O}(\{\boldsymbol{x}_k\}_{k=1}^M, \boldsymbol{\alpha}) := \sum_{k=1}^M \alpha_k U^\dagger(\boldsymbol{x}_k) \mathcal{O}U(\boldsymbol{x}_k) = \sum_{k=1}^M \alpha_k \rho(\boldsymbol{x}_k)$. This is an optimal family of models in a kernel method via the representer theorem~\cite{schuld_supervised_2021}.

Next, returning to~\eqref{eqn:density_qnn_with_data_as_linear_model} and replacing the data dependence from the \emph{state} with a parameter dependence, $\rho(\boldsymbol{x}) \rightarrow \rho(\boldsymbol{\theta})$, we fall within the family of \emph{flipped} quantum models~\cite{jerbi_shadows_2024}. These are a useful model family where the role of data and parameters in the model have been flipped. This insight enables the incorporation of classical shadows~\cite{huang_predicting_2020} for, e.g. quantum training and classical deployment of QML models.

Finally, we have the framework of \emph{post-variational} quantum models~\cite{huang_post-variational_2023}, originating from the  classical combinations of quantum states ansatz~\cite{huang_near-term_2021}. In motivation, these models are perhaps more similar to the `\emph{implicit}'~\cite{jerbi_quantum_2023} models such as quantum kernel methods, where the quantum computer is used only for specific \emph{fixed}, non-trainable, operations (e.g. evaluating inner products for kernels), rather than `\emph{explicit}'~\cite{jerbi_quantum_2023} models where trainable parameters reside within unitaries, $U(\boldsymbol{\theta})$. Post-variational models involve optimising coefficients $\alpha_{k, q}$, which are injected into the model via a linear or non-linear combination of observables, $\{\mathcal{O}_{q}\}_{q=1}^Q$,  applied to (non-trainable) unitary transformed states, $\{U_{k}\rho(\boldsymbol{x})U^{\dagger}_{k}\}_{k=1}^K$. In the linear case, the output of the model is:
\begin{multline}\label{eqn:post_variational_output}
    f_{\mathsf{PV}}(\boldsymbol{\alpha}, \boldsymbol{x}) = \sum_{kq} \alpha_{kq}\Tr(\mathcal{O}_q U_{k}\rho(\boldsymbol{x})U^{\dagger}_{k})  \\
    = \sum_{kq} \alpha_{kq}\Tr(\mathcal{O}_{kq} \rho(\boldsymbol{x})), \mathcal{O}_{kq} := U_{k}\mathcal{O}_q U^{\dagger}_{k}
\end{multline}
The major benefit of post-variational models is that, similar to quantum kernel methods, the optimisation over a convex combination of parameters \emph{outside} the circuit is in principle significantly easier than the non-convex optimisation of parameters \emph{within} the unitaries. However, just like kernel methods, this comes at the limitation of an expensive forward pass through the model, which requires $\mathcal{O}(KQ)$ circuits to be evaluated. In the worst case, this should also be exponential in the number of qubits in to enable arbitrary quantum transformations on $\rho(\boldsymbol{x})$, $KQ \leq 4^n$~\cite{huang_post-variational_2023}. To avoid evaluating an exponential number of quantum circuits, it is clearly necessary to employ heuristic strategies or impose symmetries to choose a sufficiently large yet expressive pool of operators $\mathcal{O}_{kq}$. In light of this, Ref~\cite{huang_post-variational_2023} proposes ansatz expansion strategies~\cite{huang_near-term_2021} or gradient heuristics to grow the pool of quantum operations. Such techniques may be also incorporated into our proposal, but we leave such investigations to future work.

\subsection{Preparing density quantum neural networks} \label{ssec:density_qnn_preparation}

As mentioned above, we have not yet described a method to prepare the density QNN state,~\eqref{eqn:density_qnns}. \figref{fig:density_qnn_mainfig} showcases two methods of doing so. For now, we do not assume any specific choice for the sub-unitaries. There are two methods to prepare the density state. The first is via a \emph{deterministic} circuit which exactly prepares $\rho(\boldsymbol{\theta}, \boldsymbol{\alpha}, \boldsymbol{x})$, and shown in~\figref{fig:density_qnn_mainfig}b. We prove the correctness of this circuit in~Supp.~Mat. I.2. The structure of the circuit can be related directly to the corresponding \emph{linear combination of unitaries} QNN~\cite{heredge_non-unitary_2024} which prepares instead the pure state, $\sum_k \alpha_k U_k(\boldsymbol{\theta}_k)\ket{\boldsymbol{x}}$, seen in~\figref{fig:density_qnn_mainfig}a. Notably, the deterministic density QNN removes the need for ancilla postselection on a specific state, ($(\ket{0}_{\mathcal{A}}^{\otimes n})$ in the figure). In other words, while a single forward pass through an LCU QNN will only succeed with some probability $p$, the deterministic density QNN state preparation succeeds with probability $p=1$. While the circuits in~\figref{fig:density_qnn_mainfig}a and b are conceptually simple, the controlled operation of the sub-unitaries may be very expensive in practice, which is a necessity without any further assumptions. In~Supp.~Mat. I.2 we discuss certain assumptions on the structure of the sub-unitaries which may simplify the resource requirements of this preparation mechanism, specifically assuming a Hamming-weight preserving structure allows the removal of the generic controlled operation.

The second method uses the \emph{distributional} property of $\boldsymbol{\alpha}$ to only prepare the density state $\rho(\boldsymbol{\theta}, \boldsymbol{\alpha}, \boldsymbol{x})$ \emph{on average}, depicted in~\figref{fig:density_qnn_mainfig}c. In this form, the forward pass completely removes the need for ancillary qubits, and complicated controlled unitaries. 

The effect of this is threefold:
\begin{enumerate}
    \item A forward pass through the randomised density QNN (\figref{fig:density_qnn_mainfig}c) requires time which is upper bounded by the execution speed of \emph{only} the most complex unitary, $U_{k^*}$. This is illustrated in~\figref{fig:density_qnn_mainfig}c. In the language of post-variational models measuring $Q$ observables on a randomised density QNN has complexity $\mathcal{O}(Q)$, a $K$-fold improvement. 
    \item Secondly, we will show that the gain in efficiency in moving from the LCU to randomised density QNN does not come at a significant loss in model performance. We prove this, under certain assumptions, using the Hastings-Campbell Mixing lemma from randomised quantum circuit compiling.
    \item Thirdly, and related to the first two points, one can view the randomised density QNN as an \emph{explicit} version of the post-variational (in the sense of~\cite{jerbi_quantum_2023}) framework. This may be an interesting direction to study given the series of hierarchies found by Ref.~\cite{jerbi_quantum_2023} between implicit, explicit and reuploading models.
\end{enumerate}

\subsection{Gradient extraction for density QNNs} \label{eqn:denqnn_gradient_scaling}

For density QNNs to be performant in practice, they must be efficiently trainable. In other words, it should not be exponentially more difficult to evaluate gradients from such models, compared to the component sub-unitaries. In the following, we describe general statements regarding the gradient extractability from density QNNs. By then choosing the sub-unitaries to themselves be efficiently trainable (in line with a so-called \emph{backpropagation} scaling, which we will define), the entire model will also be. We formalise this as follows:
%
%
\begin{proposition}[Gradient scaling for density quantum neural networks] \label{prop:gradient_scaling_dqnn}
Given a density QNN as in~\eqref{eqn:density_qnns} composed of $K$ sub-unitaries, $\mathcal{U} = \{U_k(\boldsymbol{\theta}_k)\}_{k=1}^K$, implemented with distribution, $\boldsymbol{\alpha} = \{\alpha_k\}$, an unbiased estimator of the gradients of a loss function, $\mathcal{L}$, defined by a Hermitian observable, $\mathcal{H}$:
\begin{equation} \label{eqn:density_qnn_loss_fn}
    \mathcal{L}(\boldsymbol{\theta}, \boldsymbol{\alpha}, \boldsymbol{x}) = \Tr\Big(\mathcal{H}\rho(\boldsymbol{\theta}, \boldsymbol{\alpha}, \boldsymbol{x})\Big)
\end{equation}
can be computed by classically post-processing $\sum_{l=1}^K\sum_{k=1}^K T_{\ell k}$ circuits, where $T_{\ell k}$ is the number of circuits required to compute the gradient of sub-unitary $k$, $U(\boldsymbol{\theta}_k)$ with respect to the parameters in sub-unitary $\ell$, $\boldsymbol{\theta}_{\ell}$. Furthermore, these parameters may be shared across the unitaries, $\boldsymbol{\theta}_k = \boldsymbol{\theta}_{k'}$ for some $k, k'$. 
\end{proposition}
%

The proof is given in~Supp.~Mat. B.1, but it follows simply from the linearity of the model. Now, there are two sub-cases one can consider. First, if all parameters between sub-unitaries are independent, $\boldsymbol{\theta}_k \neq \boldsymbol{\theta}_{\ell}, ~\forall k, \ell$. This gives the following corollary, also in~Supp.~Mat. B.1 and illustrated in~\figref{fig:density_gradients}. 
%
%
\begin{corollary}\label{corr:density_qnn_gradient_independent}
Given a density QNN as in~\eqref{eqn:density_qnns} composed of $K$ sub-unitaries, $\mathcal{U} = \{U_k(\boldsymbol{\theta}_k)\}_{k=1}^K$ where the parameters of sub-unitaries are independent, $\boldsymbol{\theta}_k \neq \boldsymbol{\theta}_{\ell}, \forall k, \ell$ an unbiased estimator of the gradients of a loss function, $\mathcal{L}$, \eqref{eqn:density_qnn_loss_fn}
can be computed by classically post-processing $\sum_{k=1}^K T_{k}$ circuits, where $T_{k}$ is the number of circuits required to compute the gradient of sub-unitary $k$, $U(\boldsymbol{\theta}_k)$ with respect to the parameters, $\boldsymbol{\theta}_{k}$.
\end{corollary}
%
%
\begin{figure*}[!ht]
    \centering
    \includegraphics[width=\linewidth]{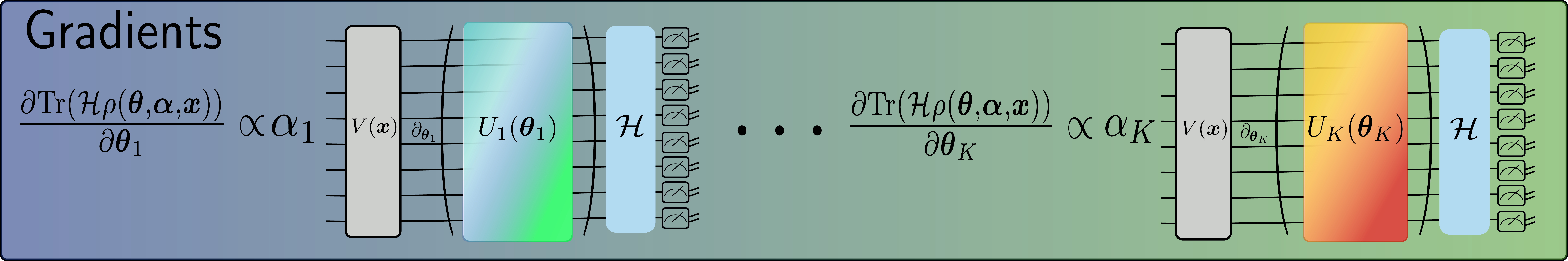}
  \caption{\textsf{\textbf{Illustration of~\corrref{corr:density_qnn_gradient_independent}.} \\ In the case where no parameters are shared across the sub-unitaries, the gradients of the density model in~\eqref{eqn:density_qnns} when measured with an observable $\mathcal{H}$ simply involves computing gradients for each sub-unitary individually. As a result, the full model introduces an $\mathcal{O}(K)$ overhead for gradient extraction. If $K=\mathcal{O}(\log(N))$ and each sub-unitary admits a backpropagation scaling for gradient extraction, the density model will also admit a backpropagation scaling.}}
\label{fig:density_gradients}
\end{figure*}

The second case is where some (or all) parameters are shared across the sub-unitaries. Taking the extreme example, $\theta_l^j = \theta_k^j =: \theta^j ~\forall k, l$ - i.e. all sub-unitaries from~\eqref{eqn:density_qnns} have the same number of parameters, which are all identical. In this case, for each sub-unitary, $l$, we must evaluate all $K$ terms in the sum so at most the number of circuits will increase by a factor of $K^2$ - we need to compute every term in the matrix of partial derivatives.

Note that this is the number of \emph{circuits} required, not the overall sample complexity of the estimate. For example, take the single layer commuting-block circuit (just a commuting-generator circuit) with $C$ mutually commuting generators. Also assume a suitable measurement observable, $\mathcal{H}$, such that the resulting gradient observables, $\{\mathcal{O}_c | \mathcal{O}_c := [G_c, \mathcal{H}]\}_{k=1}^C$, can be simultaneously diagonalized. To estimate these $C$ gradient observables each to a precision $\varepsilon$ (meaning outputting an estimate $\Tilde{o}_c$ such that $|\Tilde{o}_k - \bra{\psi}\mathcal{O}_c\ket{\psi}| \leq \varepsilon$ with confidence $1-\delta$) requires $\mathcal{O}\left(\varepsilon^{-2}\log \left(\frac{C}{\delta}\right)\right)$ copies of $\psi$ (or equivalently calls to a unitary preparing $\psi$). It is also possible to incorporate strategies such as shadow tomography~\cite{huang_predicting_2020}, amplitude estimation~\cite{brassard_quantum_2002} or quantum gradient algorithms~\cite{huggins_nearly_2022} to improve the $C$, $\delta$ or $\varepsilon$ parameter scalings for more general scenarios, though inevitably at the cost of scaling in the others.

\subsubsection{Efficiently trainable density networks} \label{sssec:backprop_density_comparison}

The results of the previous section state that moving to the density framework does not result in an exponential increase in gradient extraction difficulty, unless the number of sub-unitaries is exponential. However, what we really care for is that the models are end-to-end \emph{efficiently} trainable, meaning that overall their gradients can be computed with a \emph{backpropagation} scaling. This is the resource scaling which the (classical) backpropagation algorithm obeys, and which we ideally would strive for in quantum models. In the following, we can specialise the derived results to the cases where the component sub-unitaries have efficient gradient extraction protocols. This will render the entire density model also efficiently trainable in this regime.

This `backpropagation' scaling can be defined as follows. Specifically:
\begin{definition}[Backpropagation scaling~\cite{abbas_quantum_2023, bowles_backpropagation_2023}]\label{def:backpropagation_scaling}
    Given a parameterised function, $f(\boldsymbol{\theta}), \boldsymbol{\theta} \in \mathbb{R}^N$, with $f'(\boldsymbol{\theta})$ being an estimate of the gradient of $f$ with respect to $\boldsymbol{\theta}$ up to some accuracy $\varepsilon$. The total computational cost to estimate $f'(\boldsymbol{\theta})$ with backpropagation is bounded with:
    \begin{equation}
        \mathcal{T}(f'(\boldsymbol{\theta})) \leq c_t  \mathcal{T}(f(\boldsymbol{\theta}))
    \end{equation}
    and
    \begin{equation}
         \mathcal{M}(f'(\boldsymbol{\theta})) \leq c_m \mathcal{M}(f(\boldsymbol{\theta}))
    \end{equation}
    where $c_t, c_m = \mathcal{O}(\log(N))$ and $ \mathcal{T}(g)/\mathcal{M}(g)$ is the time/amount of memory required to compute $g$.
\end{definition}
In plain terms, a model which achieves a backpropagation scaling according to~\defref{def:backpropagation_scaling}, particularly for quantum models, implies that it does not take significantly more effort, (in terms of number of qubits, circuit size, or number of circuits) to compute gradients of the model with respect to all parameters, than it does to evaluate the model itself. 

One family of circuits which does obey such a scaling are the so-called \emph{commuting-block} quantum neural networks, defined in Ref.~\cite{bowles_backpropagation_2023}, and which contain $B$ blocks of unitaries generated by operators which all mutually commute within a block. We discuss the specific circuits in~\secref{sec:methods} but for now we specialise~\propref{prop:gradient_scaling_dqnn} to these commuting-block unitaries as follows:
%
\begin{corollary}[Gradient scaling for density commuting-block quantum neural networks] \label{corr:grad_scaling_density_commuting_block}
  Given a density QNN containing $k$ sub-unitaries, each acting on $n$ qubits. Each sub-unitary, $k$, has a commuting-block structure with $B_k$ blocks. Assume each sub-unitary has different parameters, $\boldsymbol{\theta}_{k}\neq \boldsymbol{\theta}_{\ell}~, \forall k, \ell$. Then an unbiased estimate of the gradient can be estimated by classical post-processing $\mathcal{O}(2\sum_k B_k - K)$ circuits on $n+1$ qubits.
\end{corollary}
%
%
\begin{proof}
This follows immediately from~\propref{prop:gradient_scaling_dqnn} and Theorem 5. from Ref.~\cite{bowles_backpropagation_2023}. Here, the gradients of a single $B$-block commuting block circuit can be computed by post-processing $2B-1$ circuits, where $2$ circuits are required per block, with the exception of the final block, which can be treated as a commuting generator circuit and evaluated with a single circuit.
\end{proof}

\begin{table*}[!ht]
\centering
\captionsetup{justification=raggedright}
\begin{tabular}{l|c@{\hspace{5mm}}c@{\hspace{5mm}}c@{\hspace{5mm}}c}
\toprule
\textbf{QNN ansatz} & $N_{\textsf{params}}$  & $N_{\textsf{grad}}$ & $N_{\textsf{params}}$ & $N_{\textsf{grad}}$ 
\\ & \textbf{Original}  &  \textbf{Original} & \textbf{Density} & \textbf{Density}  \\
\midrule
$D$ layer hardware efficient ~\cite{kandala_hardware-efficient_2017} & $\mathcal{O}(nD)$ & $\mathcal{O}(nD)$ &  $\mathcal{O}(nD)$  & $\mathcal{O}(D)$ \\
Equivariant XX~\cite{bowles_backpropagation_2023} & $\mathcal{O}(G)$ & $\mathcal{O}(1)$ & $\mathcal{O}(KG)$ & $\mathcal{O}(K)$ \\
HW pres.~\cite{landman_quantum_2022} - pyramid & $\mathcal{O}(n^2)$ & $\mathcal{O}(n^2)$ &  $\mathcal{O}(n)$ & $\mathcal{O}(1)$ \\
HW pres.~\cite{cherrat_quantum_2022} - butterfly & $\mathcal{O}(n\log(n))$ & $\mathcal{O}(n\log(n))$ & $\mathcal{O}(n\log(n))$ & $\mathcal{O}(\log(n))$ \\
HW pres.~\cite{hamze_parallelized_2021} - round-robin & $\mathcal{O}(n^2)$ & $\mathcal{O}(n^2)$ & $\mathcal{O}(n^2)$ & $\mathcal{O}(n)$ \\
\bottomrule
\end{tabular}
\caption{\textsf{\textbf{Summary of gradient scalings for training density quantum neural networks.}\\
Number of gradient circuits ($N_{\textsf{grad}}$) required to estimate full gradient vector for original quantum neural networks versus their density QNN counterparts each with $N_{\textsf{params}}$ parameters acting on $n$ qubits. The equivariant XX ansatz~\cite{bowles_backpropagation_2023}, an example of a commuting-generator circuit contains $G$ commuting unitaries ($G$ depends on the maximum locality chosen), $K$ versions of which can be combined to give a density version. For the HW preserving, we refer to the versions created exclusively from commuting-block unitaries, not those which take $K$ versions of the original circuit - see~\figref{fig:round_robin_density_decomposition} for the distinction with the round-robin circuit. We suppress precision factors of $\mathcal{O}(\varepsilon^{-2})$ and $\mathcal{O}(\log(\delta^{-1}))$, and we assume a direct sampling method to evaluate gradients. }}
\label{tab:summary_density_comparison}
\end{table*}

At this stage, we showcase two possibilities when constructing density networks. It should be noted that these in some sense represent extreme cases, and should not be taken as the exclusive possibilities. Ultimately, the successful models will likely exist in the middle group. The first path allows us to increase the \emph{trainability} of certain QML models in the literature. In~\tabref{tab:summary_density_comparison}, we show some results if we were to do so for some popular examples. The first step is to dissect commuting-block components from each `layer' of the respective model, then treat these components as sub-unitaries within the density formalism and then apply~\corrref{corr:grad_scaling_density_commuting_block}. In the following sections, we describe this strategy for the models in the table, beginning with the hardware efficient ansatz.

Secondly, we may simply use it as a means to increase overall model expressibility - where the component sub-unitaries, $U_k(\boldsymbol{\theta}_k)$ are any generic trainable circuits (which are independent for simplicity). Further assume each $U_k(\boldsymbol{\theta}_k)$ has identical structure with $N$ parameters acting on $n$ qubits and requiring $T_{n, N}$ gradient circuits each. Then, from~\corrref{corr:density_qnn_gradient_independent}, a density model will require $KT_{n, N}$ parameter-shift circuits. In many cases, $K$ will be a constant independent of $N$, or $n$, and furthermore this evaluation over the $K$ sub-unitaries can be done in parallel.

In the next section, and in~\figref{fig:hardware_efficient_density} we illustrate these paths using hardware efficient quantum neural networks. For the examples in the following sections (the other models referenced in~\tabref{tab:summary_density_comparison} and others), we demonstrate both of these directions.

\begin{figure*}[!ht]
    \centering
    \includegraphics[width=\linewidth]{Fig3_hardware_efficient_new.pdf}
  \caption{\textsf{\textbf{Decomposing a hardware efficient ansatz for a density QNN.} \\ $D$ layers of a hardware efficient (HWE) ansatz with entanglement generated by CNOT ladders and trainable parameters in single qubits $R_x, R_y, R_z$ gates.\\
  (bottom left) $D$ layers extracted into $D$ sub-unitaries with probabilities, $\{\alpha_d\}_{d=1}^D$ for a density QNN version. Applying the commuting-generator framework to the density version, $\rho^{\mathsf{HWE}}(\boldsymbol{\theta}, \boldsymbol{\alpha}, \boldsymbol{x})$, enables parallel gradient evaluation in $2D$ circuits versus $2nD$ as required by the pure state version, $\ket{\psi^{\mathsf{HWE}}(\boldsymbol{\theta}, \boldsymbol{\alpha}, \boldsymbol{x})}$. TO illustrate potential differences between sub-unitaries, we arbitrarily reverse CNOT directions in subsequent layers and partially accounting for low circuit depth. \\
  (bottom right) Alternatively, we can simply create a more expressive version of the hardware efficient QNN within the density framework by duplicating across $K$ sub-unitaries with probabilities $\{\alpha_k\}_{k=1}^K$ retaining $D$ layers each. In this case, the model requires $2nDK$ circuits for gradient extraction, but each sub-unitary can have independent parameters learning different features, especially if each contains different entanglement structures.}
  }
\label{fig:hardware_efficient_density}
\end{figure*}

\subsubsection{Hardware efficient quantum neural networks} \label{ssec:hardware_efficient_qnn}

To illustrate the two possible paths for model construction, we use a toy example (shown in~\figref{fig:hardware_efficient_density}) - the common but much maligned \emph{hardware efficient}~\cite{kandala_hardware-efficient_2017} quantum neural network. These `problem-independent' ans\"atze were proposed to keep quantum learning models as close as possible to the restrictions of physical quantum computers, by enforcing specific qubit connectivities and avoiding injecting trainable parameters into complex transformations. These circuits are extremely flexible, but this comes at the cost of being vulnerable to barren plateaus~\cite{mcclean_barren_2018} and generally difficult to train.

A $D$ layer hardware efficient ansatz on $n$ qubits is usually defined to have $1$ parameter per qubit (located in a single qubit Pauli rotation) per layer. The parameter-shift rule with such a model would require $2nD$ individual circuits to estimate the full gradient vector, each for $M$ measurements shots. Given such a circuit, we can construct a density version with $D$ sub-unitaries and reduce the gradient requirements from $2nD$ to $2D$ as the gradients for the single qubit unitaries in each sub-unitary can be evaluated in parallel, using the commuting-generator toolkit from~\secref{sec:methods} and \corrref{corr:grad_scaling_density_commuting_block}. This example is relatively trivial as the resulting unitaries are shallow depth (which also likely increases the ease of classical simulability) and training each corresponds only to learning a restricted single qubit measurement basis. In the~\figref{fig:hardware_efficient_density}, we take a variation of the common CNOT-ladder layout - entanglement is generated in each layer by nearest-neighbour CNOT gates. Typically, an identical structure is used in each layer, however in the figure we allow each sub-unitary extracted from each layer to have a varying CNOT control-target directionality and different single qubit rotations in each layer. This is to increase differences between each ``expert'' (see below) as each sub-unitary can generate different levels of (dis)entanglement. Secondly, as illustrated in~\figref{fig:hardware_efficient_density} one can also define a density version which is \emph{not} more trainable than the original version - in this case, we have $K$ depth-$D$ hardware efficient circuits, which according to the parameter shift rule would now require $\mathcal{O}(2KDn)$ circuits. However, the density model contains more parameters ($K$-fold more) than the original single circuit version, and possibly is more expressive as a QML model.

\subsection{LCU and the Mixing lemma} \label{ssec:the_mixing_lemma_density}

The second feature of the density framework is the relationship to linear combination of unitaries (LCU) quantum machine learning. Above, we discussed two methods of preparing the density state,~\eqref{eqn:density_qnns} and illustrated in~\figref{fig:density_qnn_mainfig}.  Now, we will demonstrate how one may translate performance guarantees from families of LCU QNNs (\figref{fig:density_qnn_mainfig}a) to the randomised version of the density QNN (\figref{fig:density_qnn_mainfig}c).

Specifically, we will show that, in at least one restricted learning scenario, we will show that if one can construct and train an LCU QNN (\figref{fig:density_qnn_mainfig}a) which has a better learning performance (in terms of e.g. classification accuracy) than any component unitary, this improved performance can be transferred to a density QNN without performance loss. This transference has an important consequence - due to the minimal requirements of implementing a randomised density QNN (\figref{fig:density_qnn_mainfig}c) on quantum hardware, relative to the LCU QNN, we can implement the more performant model much more cheaply. To do so, we will prove a result using the Hastings-Campbell mixing lemma~\cite{hastings_turning_2017, campbell_shorter_2017} from the field  compiling of complex unitaries onto sequences of simpler quantum operations.

In this context, we will adapt the Mixing lemma as follows. Assume one trains $K$ sub-unitaries $\{U_k(\boldsymbol{\theta}_k)\}$ each to be `good' models, in that they each achieve a low prediction error, $\delta_1$, to some ground truth function. Next, with the trained sub-unitaries fixed, one learns a linear combination, $\sum_k \alpha_k U_k$, with (distributional) coefficients, $\{\alpha_k\}$, by training \emph{only} the coefficients. Assume this more powerful QNN model (LCU QNN) achieves a `better' prediction error, $\delta_2 < \delta_1$. However, despite better performance, the LCU QNN is far more expensive to implement than any individual $U_k$ (as can be seen in~\figref{fig:density_qnn_mainfig}a). The logic of the Mixing lemma implies that instead of this deep circuit, we may randomise over the unitaries - create a randomised density QNN - and achieve the same error as the LCU QNN but with the same overhead as the most complex $U_k$. We formalise this as the following:
\begin{lemma}[Mixing lemma for supervised learning]\label{lemma:supervised_mixing_corr}
Let $h(\boldsymbol{x})$ be a target ground truth function, prepared via the application of a fixed unitary, $V$, $h(\boldsymbol{x}) := \Tr(\mathcal{O}V\rho(\boldsymbol{x}) V^{\dagger})$ on a data encoded state, $\rho(\boldsymbol{x})$ and measured with a fixed observable, $\mathcal{O}$. Suppose there exists $K$ unitaries $\{U_k(\boldsymbol{\theta})\}_{k=1}^K$ such that these each are $\delta_1$ good predictive models of $h(\boldsymbol{x})$:
    \begin{equation} \label{eqn:supervised_mixing_corrollary_condition_1}
        \mathbb{E}_{\boldsymbol{x}}|h(\boldsymbol{x}) - f_k(\boldsymbol{\theta}, \boldsymbol{x})| \leq 
        \delta_1, \forall k
    \end{equation}
    and a distribution $\{\alpha_k\}_{k=1}^K$ such that predictions according to the LCU model $f_{\textsf{LCU}}(\boldsymbol{\theta}, \boldsymbol{\alpha}, \boldsymbol{x}) := \Tr(\mathcal{O}\Big(\sum_k \alpha_k U_k(\boldsymbol{\theta})\Big) \rho(\boldsymbol{x})\Big(\sum_k \alpha_k U^{\dagger}_k(\boldsymbol{\theta})\Big))$ have error  bounded as:
    \begin{equation} \label{eqn:supervised_mixing_corrollary_condition_2}
        \mathbb{E}_{\boldsymbol{x}}|h(\boldsymbol{x}) - f_{\textsf{LCU}}(\boldsymbol{\theta}, \boldsymbol{\alpha}, \boldsymbol{x})| \leq \delta_2 
    \end{equation}
    for some $\delta_1, \delta_2 >0$. Then, the corresponding density QNN, $f(\boldsymbol{\theta}, \boldsymbol{\alpha}, \boldsymbol{x}) = \Tr(\mathcal{O}\rho(\boldsymbol{\theta}, \boldsymbol{\alpha}, \boldsymbol{x})),  \rho(\boldsymbol{\theta}, \boldsymbol{\alpha}, \boldsymbol{x}) = \sum_{k=1}^K\alpha_k U_k^\dagger(\boldsymbol{\theta}) \rho(\boldsymbol{x})U_k(\boldsymbol{\theta})$ can generate predictions for $h(\boldsymbol{x})$ with error:
    \begin{equation} \label{eqn:supervised_mixing_corollary}
         \mathbb{E}_{\boldsymbol{x}}|h(\boldsymbol{x}) - f(\boldsymbol{\theta}, \boldsymbol{\alpha}, \boldsymbol{x})| \leq \frac{\delta_1^2}{4\|\mathcal{O}\|_{\infty}} + 2\delta_2
    \end{equation} 
\end{lemma}

We prove this lemma in~Supp.~Mat. B.2. In the above, $\mathbb{E}_{\boldsymbol{x}}|h(\boldsymbol{x}) - g_{\boldsymbol{\theta}}(\boldsymbol{x})|$ is the expected prediction error admitted by the model, $g_{\boldsymbol{\theta}}(\boldsymbol{x})$, where the expectation is taken over the distribution the data is drawn from. Take the common predictor observable to be $\mathcal{O} = Z$, $\|\mathcal{O}\|_{\infty} = 1$ and set $\delta_1 = \delta$. Assume we find a distribution for the  LCU QNN which quadratically reduces this error $\delta_2 = \delta^2$. Then according to~\lemref{lemma:supervised_mixing_corr}, the density QNN will also be an $\mathcal{O}(\delta^2)$ good predictor, but at the same implementation cost as a single sub-unitary QNN. 

In the rest of this work, we do not primarily take the QNN $\rightarrow$ LCU QNN $\rightarrow$ density QNN route as implied by~\lemref{lemma:supervised_mixing_corr}, instead directly train the parameters of the density QNN indicating the viability of bypassing the expensive LCU directly. Though, for the sake of completeness, in~Supp.~Mat.~C we do provide one example where the LCU QNN outperforms the single-circuit QNN numercially. It may be that non-unitary (pure LCU) quantum machine learning~\cite{heredge_non-unitary_2024} and the density quantum neural networks we present here have different difficulties \emph{in practice} to achieve good model performance. Also, the above result is clearly restrictive to the setting where the model to be learned is itself a QNN, and we also know the correct observable to measure. However, it may be generalised to stronger statements where the data is generated by an arbitrary classical function. We leave both of these investigations to future work.

In the quantum compiling problem, one endeavours to produce a sequence of `simple' operations which approximate the behaviour of a target unitary or quantum channel on input states. Randomised compiling (informally via the Hastings-Campbell mixing lemma~\cite{hastings_turning_2017, campbell_shorter_2017}) allows one to carry a quadratic suppression in compiling error from a channel composed of a linear combination of unitaries, into the corresponding \emph{randomised} channel which has the same execution overhead as an individual unitary in the linear combination. If we treat the compilation task as a learning problem, as has been done in several works, spawning the subfield of \emph{variational} quantum compiling~\cite{khatri_quantum-assisted_2019, sharma_noise_2020}, one can intuitively see how the Mixing lemma may be directly applied. 

However, to argue for the benefit of density networks over individual QML circuits, we must go further and generalise to other learning tasks, for example supervised data classification/regression (which is the primary task we use density QNNs for in this work).

\subsection{Connection to classical mechanisms}  \label{sssec:classical_mechnanisms}

We have discussed features of density QNNs, and elucidated their position within the spectrum of existing quantum models. In this section, we discuss relationships or analogies between density QNNs to \emph{purely classical} mechanisms and models. We summarise these relationships in the following three observations:

\begin{itemize}
    \item \textbf{Observation 1:} The randomised state preparation method may be used as the training mode, and the deterministic state preparation method may be used as the inference mode for density QNNs, if analogies are to be made with the classical dropout mechanism.

    \item \textbf{Observation 2:} Unlike classical dropout, density QNNs do not combine an exponential number of sub-networks at inference time.

    \item \textbf{Observation 3:} Density QNNs are a quantum analogue of Mixture of Experts (MoE) models, which are subtly different from ensemble methods.
    
\end{itemize}

We expand on these observations in the following.

\subsubsection{Quantum dropout} \label{sssec:quantum_dropout}
In the previous section, we have demonstrated how density QNNs have the capacity to mitigate overfitting via more strategic parameter allocation. In light of this, one may make an analogy with the randomised density QNN and the \emph{dropout}~\cite{srivastava_dropout_2014, baldi_understanding_2013} mechanism in classical neural networks. Dropout is an effective method of combining the predictions of exponentially many classical neural networks, and also mitigates overfitting. This analogy has indeed been remarked in recent works~\cite{nguyen_theory_2022} due to the random `removal' of $K-1$ sub-unitaries in each forward pass, which, on the surface, appears similar to the randomised removal of neurons in a neural network. it has therefore been conjectured that randomised density QNNs as we have defined them may be less prone to overfitting because of this analogy. However, in the following, we will argue that this is an incorrect, or at least an incomplete comparison.

This incompleteness of the comparison arises from (at least) two sources. Firstly, there is a distinct difference (and important) between the  \emph{training} and \emph{evaluation} modes in classical dropout. In the training mode, a dropout mechanism applies a random zeroing of the nodes of a neural network, which removes outgoing weights from those nodes. This means, on any forward pass through the loss function, only a sub-network is actually activated. However, in the \emph{inference} mode, all sub-networks are effectively present, where the outgoing weights of a dropped node are re-weighted by the probability of dropping that node.

Regarding different training and inference operations, we propose the randomised implementation of the density QNN (\figref{fig:density_qnn_mainfig}c) as the state preparation method for the training phase of the model. To align with classical dropout, we then propose the \emph{deterministic} density QNN (\figref{fig:density_qnn_mainfig}b) as the inference mode. Due to the linearity of the model, this has the effect of weighting the contributions of each sub-unitary QNN, $U_k(\boldsymbol{\theta}_k)$ by the corresponding coefficients, $\alpha_k$, exactly as in classical dropout. 

However, dropout also has the key feature of efficiently combining an \emph{exponential} number of effective sub-networks at inference, which is a critical feature in boosting model performance. However, in order to maintain training efficiency (within~\corrref{corr:grad_scaling_density_commuting_block}), a density QNN may have at most $K = \mathcal{O}(\log(N))$ `sub-networks', where $N$ is the total number of circuit parameters. As such, it is unclear if at a fundamental level a density QNN can function fully as a quantum analogue for dropout. Nevertheless, we will demonstrate that the density QNN still can, in some capacity, mitigate overfitting, and perform the effective \emph{action} of dropout.

\begin{figure*}[!ht]
\includegraphics[width=\linewidth]{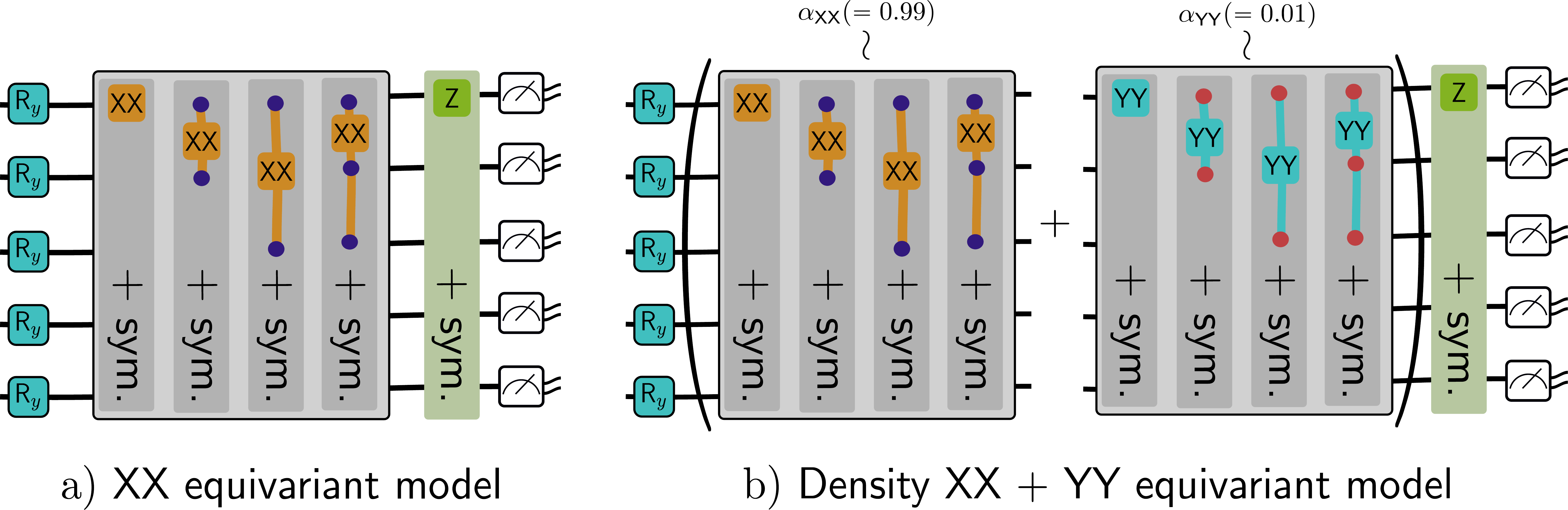}
  \caption{
  \textsf{
  \textbf{Equivariant density QNNs.} \\
    a) \textbf{The commuting-generator XX model} and a b) \textbf{XX+YY density QNN model}. The former contains up to three-body Pauli-X generated operations with twirling applied to enforce equivariance. The latter contains two sub-unitaries $U_{XX}$ (circuit (a)) and $U_{YY}$ which has the same structure but replacing Pauli-$X$ operations with Pauli-$Y$. $U_{XX}/U_{YY}$ are applied with probabilities $\alpha_{XX/YY}$. Each sub-unitary in b) are commuting-generator circuits, so each has efficiently extractable gradients.
  }
  }
  \label{fig:equivariant_models}
\end{figure*}

\begin{figure*}[!ht]
\centering
    \includegraphics[width=\linewidth]{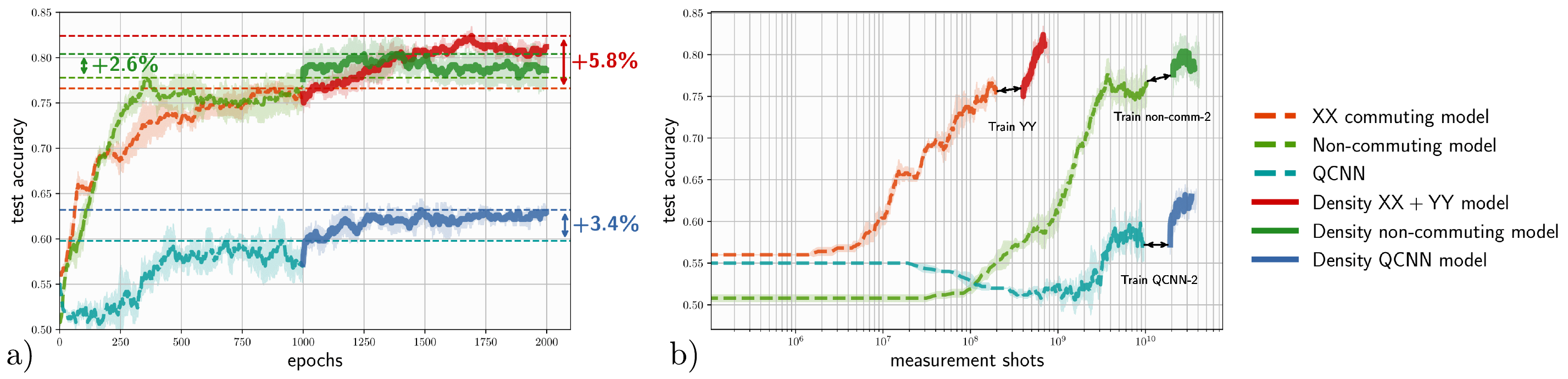}
  \caption{
  \textsf{
  \textbf{Numerics on noisy bars and dots dataset.} \\
    We create density QNN versions of the non-trivial models from Ref.~\cite{bowles_backpropagation_2023}; 1) the commuting-generator circuit in \figref{fig:equivariant_models}, 2) a `non-commuting' equivariant QNN and 3) the quantum convolutional neural network~\cite{cong_quantum_2019}, all on $10$ qubits. In all cases the density QNN is initialised from (separately) pretrained (for $1000$ epochs) base versions, and training continues for another $1000$. For the non-commuting and QCNN density models the sub-unitaries have identical structure but trained independently.  We show mean and standard deviation in test accuracy vs. a) \textbf{training epochs} and b) \textbf{number of overall shots} over $5$ independent training runs, from the same initialisation. In all cases, after base model performance saturation, the density QNN improves the final result. The gaps in (b) are to account for the extra measurement overhead to initialise and train the second sub-unitaries. This is $U_{YY}$ for the density model and $U_2$ (non-comm-2/QCNN-2) for the other two models. In all cases, the density version is initialised with $\alpha_1/\alpha_{XX} = 0.99$ $\alpha_2/\alpha_{YY} = 0.01$ which are also trainable.
  }
  }
  \label{fig:equivariant_model_numerics}
\end{figure*}

\subsubsection{Density quantum neural networks as a mixture of experts} \label{ssec:moe_density_qnn}

The next, and final, comparison to classical methods we discuss is an interpretation of the density QNN framework as a \emph{mixture of experts} (MoE) model~\cite{jacobs_adaptive_1991, jordan_hierarchical_1993}. MoE models have achieved success even at the level of large scale machine learning models such as Mixtral~\cite{jiang_mixtral_2024}, particularly where sparsity is required. The MoE framework is flexible and powerful as it allows an increase in the capacity of a model with minimal increase in computational effort. This same methodology drives the density QNN framework of this work. A MoE contains a set of `\emph{experts}', $\{f_1, \dots, f_K\}$, each of which is `responsible' for a different training case. A \emph{gating} network, $\mathcal{G}$, decides which expert should be used for a given input. In the simplest form, the MoE output, $F(\boldsymbol{x})$, is a weighted sum of the experts, according to the gating network output $F(\boldsymbol{x}) = \sum_k \mathcal{G}_k(\boldsymbol{x})f_k(\boldsymbol{x})$ where $\mathcal{G}_k(\boldsymbol{x})$ is the $k^{th}$ output of the gating network, given input $\boldsymbol{x}$. A typical difference between MoE models and ensemble models is that, for the latter every element of the ensemble is evaluated for every input, while for an MoE, only a subset are activated (e.g. the top-$k$ experts for a given input). We elaborate on this in~Supp.~Mat. H. 

We can interpret the density QNN exactly as a form of MoE as follows. By allowing the sub-unitary coefficients to be \emph{data-dependent} ($\alpha_k \rightarrow \alpha_k(\boldsymbol{x})$), as in \eqref{eqn:density_qnn_with_data_as_linear_model}, we can predict them using a gating network in an MoE. Now, the sub-unitaries become the `experts' and the overall model has increased capacity to learn which sub-unitary (expert) is more relevant for the task at hand (the given input $\boldsymbol{x}$). If the data were quantum states, we could use another quantum neural network as a gating network, or perhaps classical shadows~\cite{huang_predicting_2020} with a classical neural network. We choose a simple yet traditional model for the data-dependent gating mechanism we describe in~\secref{sec:methods}. There has been extensive investigation into variations of the MoE models both for deep~\cite{eigen_learning_2014, shazeer_outrageously_2017} and shallow models~\cite{cao_support_2003, lima_hybridizing_2007} though we leave thorough investigation of different approaches in the quantum world to future work. For example, one could employ techniques regulating the dependence on the overall mixture on any individual subsets of unitaries~\cite{eigen_learning_2014} which can occur when the distribution tends to become quite peaked, meaning the gating network chooses to rely only on a small fraction of experts for a given input.

\subsection{Numerical results} \label{ssec:numerics}
Now, we demonstrate that density networks can be successful \emph{in practice}, through three examples. The first is an demonstration where an efficiently trainable model can be made more expressive (see~\figref{fig:hardware_efficient_density}, right), and in the second we choose an example of an interpretable model family for which we construct \emph{both} efficiently trainable versions  (\figref{fig:hardware_efficient_density}, left) and more expressive (\figref{fig:hardware_efficient_density}, right). The third, and final, main example demonstrates the ability of density quantum neural networks to mitigate overfitting compared to standard single unitary QML models. We summarise the takeaways of each set of numerics as follows:

\begin{enumerate}
    \item \textbf{Model}: Equivariant quantum neural networks.
    \begin{itemize}
       \item Takeaways: Backpropagation scaling models can be made more expressive and performant. Model initialisations via sub-unitary pre-training is possible and successful.
     \end{itemize}
    \item  \textbf{Model}: Hamming-weight preserving (orthogonal) quantum neural networks.
    \begin{itemize}
       \item Takeaways: Non-backpropagation scaling models can be made trainable with minimal or no performance loss. Increasing sub-unitary expressivity monotonically improves overall performance.
    \end{itemize}
    \item  \textbf{Model}: Data reuploading quantum neural networks.
    \begin{itemize}
       \item Takeaways: Distributing parameters over shallower sub-unitaries mitigates overfitting to training data.
    \end{itemize}
\end{enumerate}

\subsubsection{Equivariant quantum neural networks} \label{ssec:numerics_equivariant_qnn}

\begin{figure}[!ht]
\centering
\includegraphics[width=0.95\linewidth]{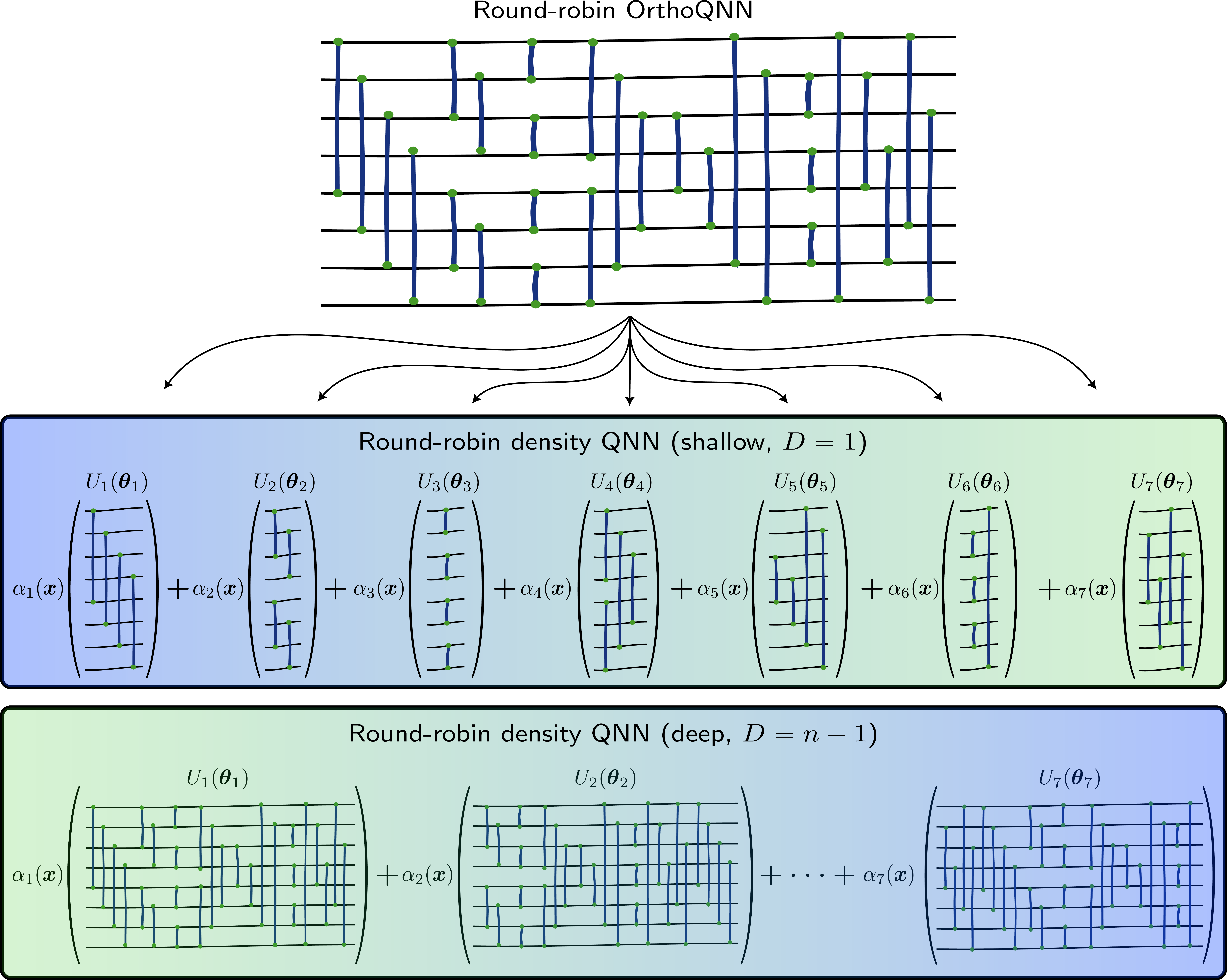}
  \caption{
  \textsf{
  \textbf{Hamming weight preserving density QNN with round-robin structure.} \\
  Illustration of extraction possibilities from an orthogonal QNN with round-robin connectivity, on $8$ qubits. We can extract a minimally expressive `shallow' (depth $D=1$)  or a maximally expressive `deep' (depth $D=\frac{n(n-1)}{2}$) sub-unitary for each `expert'. Both cases have $K=n-1$ sub-unitaries, which have $\frac{n}{2}$ (shallow) or $\frac{n(n-1)}{2}$ (deep) parameters. The shallow version has gradients which can be evaluated in parallel using $n-1$ circuits with gradient complexity stated in~\tabref{tab:summary_density_comparison}. The deep version assumes a cycled connectivity as seen in the figure with gradient extraction complexity $\mathcal{O}(n^3)$. We can interpolate between the two regimes for different sub-unitary depths, $D$, between $1\leq D\leq n-1$, but efficient gradient extraction is lost for $D\geq 2$. An MoE gating network can be used to predict the probability coefficients, $\{\alpha_k(\boldsymbol{x})\}$.
 }
  }
  \label{fig:round_robin_density_decomposition}
\end{figure}

For the first example, we build a density version of a commuting-generator (a special case of commuting-block circuit) ansatz on the simplified classification problem posed by Ref.~\cite{bowles_backpropagation_2023}. Here, the challenge is classifying \emph{bars} vs.~\emph{dots}, in a noisy setting, described in~\secref{sec:methods}.

For the base QNN, we use the original `XX' ansatz (denoted $U_{XX}(\boldsymbol{\theta})$) of Ref.~\cite{bowles_backpropagation_2023}, see~\figref{fig:equivariant_models}a) and compare against a density QNN version with two sub-unitaries, $\{U_1, U_2\}$ and coefficients $\{\alpha_1, \alpha_2\}$. We define  $U_1 := U_{XX}(\boldsymbol{\theta})$. The second sub-unitary, $U_2 := U_{YY}(\boldsymbol{\theta})$, has the exact same structure as $U_{XX}(\boldsymbol{\theta})$, but with the $XX$ generators replaced by Pauli-$Y$ generators. The measurement operator is chosen to be the translation invariant operator $\sum_i Z_i$ to suit the translation symmetry in the synthetic problem. The output function of the model is then:

\begin{equation}
f(\boldsymbol{\theta},\boldsymbol{\alpha}, \boldsymbol{x}) := \sum_{i=1}^d \sum_{k=1}^K \Tr(Z_i\rho(\boldsymbol{\theta}, \boldsymbol{\alpha}, \boldsymbol{x}))
\end{equation}
with 

\begin{multline}
    \rho(\boldsymbol{\theta}, \boldsymbol{\alpha}, \boldsymbol{x}) := 
    \alpha_1 U_{XX}(\boldsymbol{\theta})\Big[\bigotimes_{j=1}^d \ket{\boldsymbol{x}}\bra{\boldsymbol{x}}^{ry}_{j} \Big]U^{\dagger}_{XX}(\boldsymbol{\theta}) + \\
     \alpha_2 U_{YY}(\boldsymbol{\theta})\Big[\bigotimes_{j=1}^d \ket{\boldsymbol{x}}\bra{\boldsymbol{x}}^{ry}_{j}\Big] U^{\dagger}_{YY}(\boldsymbol{\theta})
\end{multline}

Each generator that appears in $U_{XX/YY}(\boldsymbol{\theta})$ is of the form $\textsf{sym}(X_1)$ or $\textsf{sym}(X_1\dots X_k)$ ($\textsf{sym}(Y_1)$ and $\textsf{sym}(Y_1\dots Y_k)$ respectively) for some $k\leq K$. The operation $\textsf{sym}$ is the twirling operation used to generate equivariant quantum circuits (in this case, equivariance with respect to translation symmetry). For example, $\textsf{sym}(X_1X_2)$ is a sum of all pairs of $X$ operators on the state with no intermediate trivial qubit, $\textsf{sym}(X_1X_3)$ is a sum of all pairs on the state separated by exactly \emph{one} trivial qubit (in either direction, visualising the qubits as a $1\textnormal{D}$ chain with closed boundary conditions) and so on. However, note that in this case since the $YY$ ansatz ($U_{XX/YY}(\boldsymbol{\theta})$) is in the same basis as the data encoding ($R_y$), we have classical simulatability for computing expectation values since the initial circuit is effectively Clifford~\cite{bowles_backpropagation_2023}. The density circuits are visualised in~\figref{fig:equivariant_models}b which we adapt from~\cite{bowles_backpropagation_2023}, and the results of the experiment can be seen in~\figref{fig:equivariant_model_numerics}a and~\figref{fig:equivariant_model_numerics}b. We also construct density versions of the two other non trivial models considered in Ref.~\cite{bowles_backpropagation_2023}, the quantum convolutional neural network (QCNN)~\cite{cong_quantum_2019} and a `non-commuting` equivariant circuit (see Ref.~\cite{bowles_backpropagation_2023}, Fig. 5c). In both these latter cases, we again consider two sub-unitaries $U_1$ and $U_2$ which have identical structure but different learned parameterisations, $U_1 = U_2$, $\boldsymbol{\theta}^*_1 \neq \boldsymbol{\theta}^*_2$. In initialising the density QNN, we bias the model towards the `first' sub-unitary with probability $99\%$ (as in ~\figref{fig:equivariant_models}b)). The density QNN outperforms its base counterpart in all cases, and the best tradeoff between efficiency and model performance is found for the density $XX+YY$ model. We describe the experiment details in~\secref{sec:methods}.

\subsubsection{Orthogonal quantum neural networks} 
\label{ssec:numerics_ortho_qnn}

For the second example, we turn to the Hamming weight (HW) preserving quantum neural network ($U(1)$ equivariant), and specifically their orthogonal quantum neural networks variants. As mentioned above, these models have desirable properties from a machine learning point of view - they are interpretable and can stabilise training. While we focus on the easiest to simulate version of HW preserving models here, all of the below is applicable to the more difficult to classically simulate compound QNNs~\cite{cherrat_quantum_2022} (Supp.~Mat. E.2) and general Hamming weight preserving unitaries. 

The data encoded state for orthogonal QNNs on $n$ qubits, $\rho(\boldsymbol{x}) := \ketbra{\boldsymbol{x}}{\boldsymbol{x}}, \ket{\boldsymbol{x}} := \sum_j x_j \ket{\mathbf{e}_j}$, is a unary amplitude encoding of the vector $\boldsymbol{x}$ ($\mathbf{e}_j$ is a basis vector with a single $1$ in position $j$ and zeros otherwise). The unitary that acts on this state consists of two qubit gates known as reconfigurable beam splitter (RBS) gates (see~\secref{sec:methods}). There are various configurations for these gates which parameterise some subset of $SO(n)$ matrices. To begin, we choose the \emph{round-robin} architecture, shown in \figref{fig:round_robin_density_decomposition} (top), composed of $\mathcal{O}(n)$ layers of where in each layer qubit is connected to only one other, alternating so eventually each qubit is connected to every other. Evaluating gradients for this QNN specification via the parameter shift rule requires $\mathcal{O}(n^2)$ circuits. We construct the `shallow' version of a round-robin density QNN by simply extracting each layer into the sub-unitaries, with $K=\mathcal{O}(n)$. Then each $U_k$ is a depth $D=1$ commuting generator circuit whose gradients can be evaluated in parallel with $\mathcal{O}(n)$ circuits via~\corrref{corr:grad_scaling_density_commuting_block}. In the spirit of~\figref{fig:hardware_efficient_density} (right), we also take $K=\mathcal{O}(n)$ copies of the \emph{full} round robin circuit to construct a more expressive model. This model will require $\mathcal{O}(Kn^2) = \mathcal{O}(n^3)$ circuits for gradient evaluation, but also has $\mathcal{O}(n)$ more parameters than the `vanilla' round-robin model. We can also interpolate between these two regimes by scaling the complexity of each sub-unitary with a round-robin layer depth between $1 \leq D \leq n-1$. In all cases, we demonstrate the MoE formalism by introducing a gating network to predict the coefficients, $\{\alpha(\boldsymbol{x})\}$, details given in~\secref{sec:methods}.

In all cases, since the unitaries are Hamming weight preserving, the output states, $\ket{\boldsymbol{y}^{k}}$ from each sub-unitary are of the form $\ket{\boldsymbol{y}} = \sum_j y_j \ket{\mathbf{e}_j}$ for some vector $\boldsymbol{y}$. This output is related to the input vector via some orthogonal matrix transformation $O^U$, $\boldsymbol{y} = O^U\boldsymbol{x}$ where the elements of $O^U$ can be computed via the angles of the RBS gates in the circuit (see~\secref{sec:methods}). The typical output of such a layer is the vector $\boldsymbol{y}$ itself, for further processing in a deep learning pipeline. As a result the output of a `orthogonal' density QNN is a linear combination of orthogonal transformations, $\boldsymbol{y} = \sum_k\alpha_k \boldsymbol{y}_k =  \sum_k\alpha_k O^{U_k}\boldsymbol{x}$, benefitting from both stable orthogonal training within each sub-unitary, and the generality afforded by the linear combination. In contrast, adding more (e.g. RBS) gates directly to the pure state version circuit (i.e. the `vanilla' round-robin QNN) would not increase the expressivity as the output would only correspond to another (different) orthogonal matrix. The results of these experiments are shown in~\figref{fig:round_robin_results}.

\subsubsection{Mitigating overfitting with density QNNs} \label{ssec:density_overfitting}

\begin{figure*}[!ht]
\centering
    \includegraphics[width=0.9\linewidth]{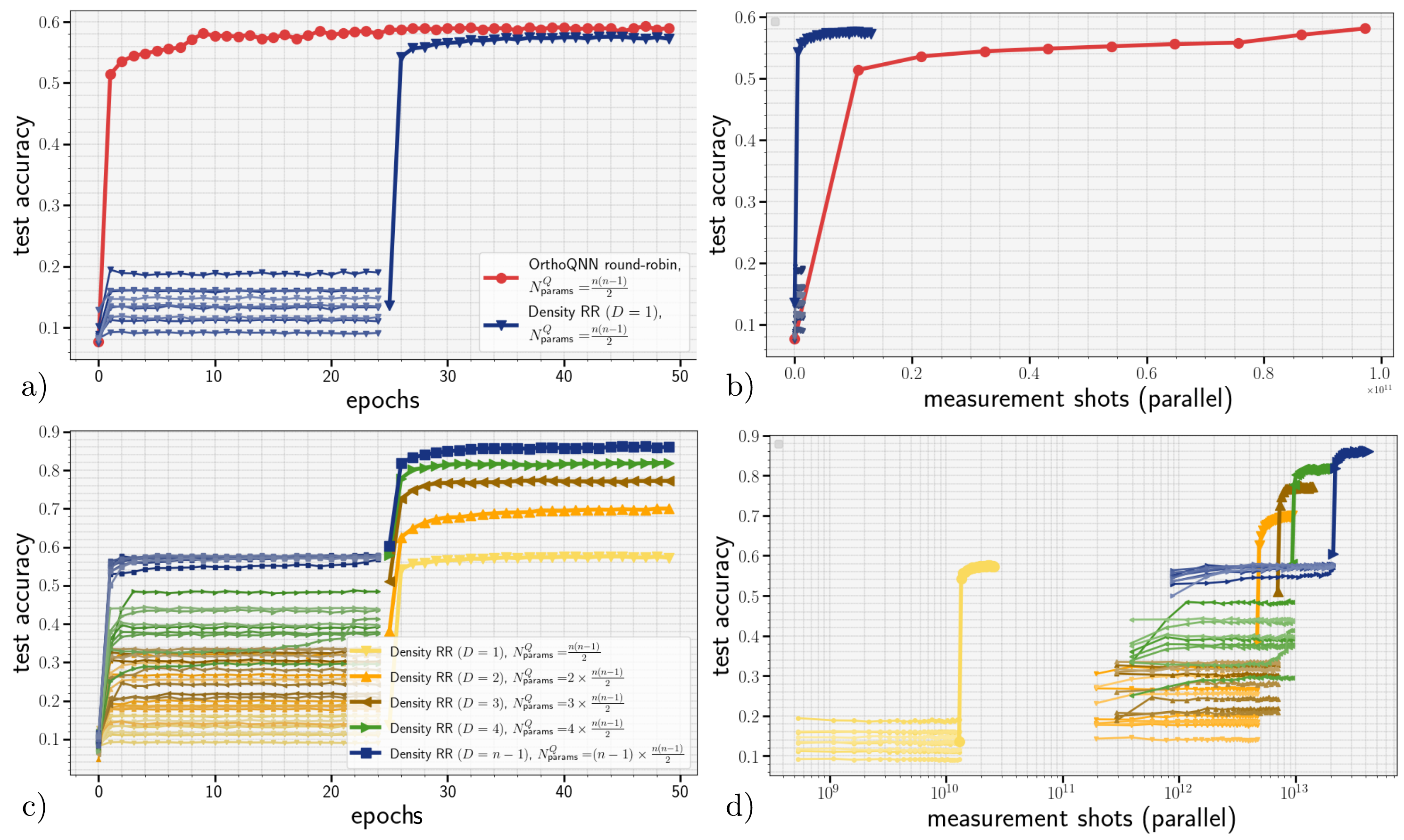}
\caption{
  \textsf{
  \textbf{Round-robin Hamming weight preserving density QNN results.} \\
  \textbf{Top row)} Round-robin OrthoQNN versus shallow ($D=1$) round-robin density QNN, as a function of training epochs (a) and measurement shots (b) on $10$ qubits. For $D=1$ the density QNN is a commuting generator circuit and so requires only $n-1$ circuit evaluations for each iteration whereas the OrthoQNN requires $\frac{4n(n-1)}{2}$ circuits as per the parameter-shift rule for gradient evaluation. \\
  \textbf{Bottom row)} increasing complexity of density QNN by increasing round robin depth, $D$, from $D \in \{1, 2, 3, 4, (n-1)\}$ as a function of epochs (c) and measurement shots (b, log scale). In all cases for the density QNN, $n-1 = 9$ sub-unitaries are trained \emph{in parallel} for $25$ epochs ($9$ thin lines) before initialising the density QNN, and training continues for a further 25. Performance of density QNN increases monotonically with sub-unitary depth, $D$.
  }
  }
  \label{fig:round_robin_results}
\end{figure*}

We have argued above that the density QNN framework possesses similarities and differences to the dropout mechanism. Nevertheless, we can demonstrate that indeed, density QNNs have the capacity to do what dropout does - which is to mitigate overfitting. To demonstrate this, we use \emph{data reuploading}~\cite{perez-salinas_data_2020} quantum neural networks. Data reuploading is a poweful QML technique which allows the construction of universal quantum classifiers, even with single qubits. We compare a single, deep (`vanilla') data reuploading QNN to a density version with approximately the same number of parameters, and test which model is more likely to overfit training data, generated by Chebyshev polynomials - see~\secref{sec:methods} for details. In summary, we find that given a budget of $N$ parameters, distributing these across a density QNN with shallow sub-unitaries will lead to better results (from the perspective of overfitting) than increasing the circuit depth of a single quantum classifier.
See~\secref{sec:methods} for a definition of data reuploading and~Supp.~Mat. G.3 for a discussion of data reuploading with density QNNs.

For both density and vanilla reuploading models, we use arbitrary single qubit rotations with $3$ parameters as the trainable operations. For the vanilla version, these are repeated for $L$ `reuploads' so the model has $3L$ parameters. For the density QNN, we choose $K=5$ sub-unitaries, the depth of each ($D$) is chosen such that $5\times 3D \approx 3L$. Both models are illustrated in~\figref{fig:data_reuploading_overfitting}a, and we also shown the truncated Fourier series learned by each of the $5$ sub-unitaries in the density case. The results can be seen in~\figref{fig:data_reuploading_overfitting}b which shows the test error, as measured by mean squared error (MSE) and the generalisation gap (difference between train and test errors).

We see as the vanilla reuploading model becomes more expressive (with increasing $L$), it overfits and test error grows. In contrast, the density model (with approximately the same number of parameters) does not overfit and the test error remains small. This means that if, without any prior, one would choose an overparameterised model (see~\cite{larocca_theory_2023, peters_generalization_2023} for further discussions on overparameterisation which may be incorporated in future work), the density model would perform better. In~Supp.~Mat. J we provide further detail on this, and supplementary numerics.

\begin{figure*}[ht!]
\centering
    \includegraphics[width=\linewidth]{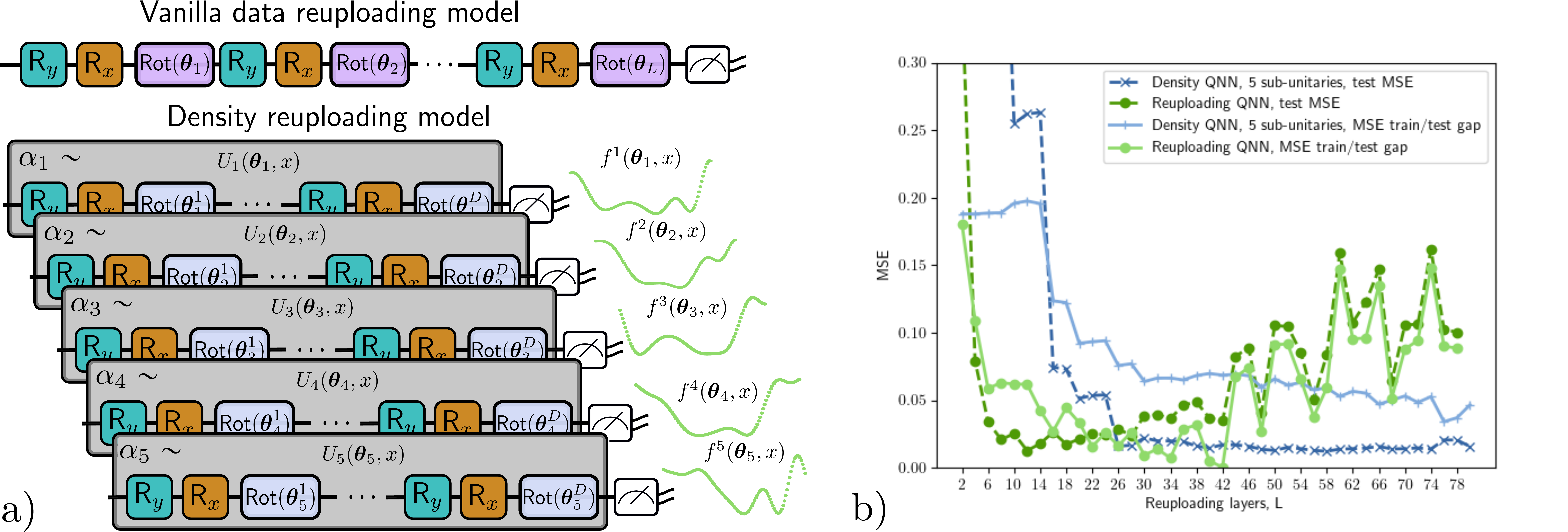}
  \caption{
  \textsf{
  \textbf{Overfitting mitigation with density QNNs - data reuploading.} \\
  a) A ``vanilla'' single qubit data reuploading QNN with $L$ reuploading layers versus a density QNN with $5$ sub-unitaries, each with $D$ data reuploads. Also shown is the partial Fourier series representations of each sub-unitary, $\{f^k(\boldsymbol{\theta}_k, x)\}_{k=1}^5$ when learning the Chebyshev polynomial $T_2(x)$.
  b) Generalisation gap between train and test mean square errors (MSE) in a regression problem for learning the Chebyshev polynomial $T_2(x)$. The number of parameters between both models is kept approximately the same for each $L$, so that $5\times 3D \approx 3L$. For a fixed number of parameters, the test error and train/test gap (generalisation error) for the vanilla data reuploading model diverges, while the density QNN test error remains small.
  }
  }
  \label{fig:data_reuploading_overfitting}
\end{figure*}

\subsection{Drawbacks and limitations} \label{ssec:comparison_to_other_methods}

To conclude our results, we discuss some potential drawbacks and limitations with the density framework, specifically those which it does not solve relative to single-unitary variational learning models. A major problem with the latter is the existence of barren plateaus, or regions of problematic gradients more generally. As discussed in Ref.~\cite{bowles_backpropagation_2023} the general relationship between commuting generator circuits, and barren plateaus is still an open question - particularly in finding a circuit architecture which has a suitable relationship between input states, circuit generators and measurement observable to render the dynamical Lie algebra polynomial, the circuit non-classically simulatable and gradients non-vanishing. This is not a question resolved via the density framework either, although if such an architecture was to be found, it could similarly be uplifted as other specific models studied in this work. However, an interesting research direction provided by the density framework is to allow a combination of different sub-unitaries, which may have different and independent gradient behaviour. However, in this case it is possible that the sub-unitary with most favourable gradients would dominate the training (a feature known as representation or expert collapse in the mixture of expert literature~\cite{lee-thorp_sparse_2022, riquelme_scaling_2021, chi_representation_2022}). However, it has also been demonstrated how barren plateau issues can be mitigated or eliminated via clever initialisation strategies

Besides the barren plateau issue, there are still a number of open questions with variational quantum learning models which are also potentially inherited by non-unitary (LCU, dissipative or density) learning models including; the effect of quantum or measurement noise, lack of data or problem-dependent circuit architectures and training dynamics related to hyperparameter tuning and loss function choices. All of these are fruitful areas for future study.

\section{Discussion} \label{sec:discussion}

Efficiently trainable quantum machine learning models, particularly those operable on quantum hardware, are crucial for the field's advancement. Here, we introduce density quantum neural networks (density QNNs), a generalisation of traditional pure parameterised quantum circuit models, incorporating classical randomisation. We demonstrate that the gradient complexity of these models hinges on the sub-unitaries' gradient evaluation. By employing commuting-block circuits for these sub-unitaries, the model requires a constant number of gradient circuits while potentially enhancing expressivity. Additionally, investigating the density formalism's impact on the efficient classical simulability of the model in specific scenarios is essential. We also highlight the connections between density QNN formalism and other primitives in quantum and classical machine learning, such as quantum-native dropout, the mixture of experts formalism, and measurement-based and post-variational quantum machine learning. Exploring the intersections among these distinct learning mechanisms presents promising research opportunities. Most significantly, we provide a direct connection between linear combination of unitaries (LCU) QML and density QNNs via the Hastings-Campbell Mixing lemma from randomised quantum compiling. This indicates that performance benefits from the LCU can be directly translated into the density framework, for a dramatically reduced cost.

Inspired by interpretable models such as orthogonal quantum neural networks, and efficiently trainable circuits such as commuting generator QNNs, we developed density versions of these families for numerical experimentation. The density versions demonstrate increased learning performance, or more efficient trainability, or possibly both, over their pure state counterparts. Future research should explore the expressivity of density QNNs with various sub-unitary families and seek other efficiently trainable sub-unitary types. Successful examples can be promptly applied to the density model, as illustrated with hardware-efficient QNNs.

This work underscores the transformative potential of density quantum neural networks in advancing quantum machine learning, laying a robust foundation for future innovations.

\section{Methods} \label{sec:methods}

\subsection{Commuting-block quantum neural networks} \label{ssec:commuting_block_circuits}

Commuting-block QNNs are circuits decomposed into $B$ \emph{blocks} or layers, so the ansatz has the following form:
\begin{equation} \label{eqn:commuting_block_circuits}
    \mathcal{U}(\boldsymbol{\theta}) = \prod_{b=1}^B\prod_{j=1}^{N_b} U_b(\boldsymbol{\theta}^b_j) = \prod_{b=1}^B\prod_{j=1}^{N_b}  e^{i\theta^b_j G^b_j}
\end{equation}
Each block, $b\in [1, \dots B]$, with $\prod_{j=1}^{N_b}  e^{i\theta^b_j G^b_j}$ contains generators, $G^b_j$ such that $[G^b_i, G^b_j] = 0~\forall i, j$. Increasing the number of blocks was shown to increase the dimension of the dynamical Lie algebra (DLA)~\cite{larocca_diagnosing_2022}, which is a concept useful in probing barren plateaus and expressibility in quantum neural networks. Assume $b=1$ for simplicity. Then, given a measurement operator, $\mathcal{H}$, each mutually commuting generator defines a gradient observable, $\mathcal{O}_k := [G_k, \mathcal{H}]$, such that $i \bra{\psi(\boldsymbol{\theta}, \boldsymbol{x})}[G_k, \mathcal{H}]\ket{\psi(\boldsymbol{\theta}, \boldsymbol{x})} = \partial_{\boldsymbol{\theta}_k}\mathcal{L}$. If all generators commute or anticommute with $\mathcal{H}$, i.e. $[G_k, \mathcal{H}] = 0$ or $\{G_k, \mathcal{H}\} = 0$ for all $k$, one can show that the gradient operators, $\mathcal{O}_k$, all mutually commute~\cite{bowles_backpropagation_2023}. As a result, $\mathcal{O}_k$ are simultaneously diagonalizable and their statistics can be extracted in parallel, using a single circuit query, which is diagonalized by appending a `diagonalizing' unitary to it. One can also extract higher order gradient information (second derivatives etc.) but with a decreasing precision~\cite{bowles_backpropagation_2023}. If $G_i~\forall i$ is a tensor product of Pauli operators, implemented in depth $T(n)$, the gradient circuit has depth $T(n) + \mathcal{O}\big(\frac{n}{\log(n)}\big)$ in general, but the overhead can be constant if $G_i$ is supported on a constant number of qubits, admitting a backpropagation scaling.

\subsection{Hamming weight preserving quantum neural networks} \label{sec:ortho_qnns}
\emph{Hamming weight} preserving unitaries, or $U(1)$ equivariant circuits~\cite{cerezo_does_2024} are a useful family of QNNs which admit favourable properties such as interpretability. Such circuits can act on states which have a fixed ($k$) Hamming weight ($k$ is equal to the number of $1$'s in a computational basis state) input or on superpositions of different Hamming weight states. In the latter case, the unitary can be written in a block diagonal form with the $k^{th}$ block, acting independently on the Hamming weight $k$ subspace~\cite{cherrat_quantum_2022, cherrat_quantum_2023, monbroussou_trainability_2023}. The dimension of the Hamming weight $k$ subspace is $\binom{n}{k}$. However, the dimension of the DLA (and hence the behaviour of barren plateaus etc.) depends on the type of operation performed within this space~\cite{monbroussou_trainability_2023}. However, for small $k$ even brute force subspace simulation is sufficient for efficient classical simulation. Though we note that DLA or brute force simulation are not the only possible classical simulation methods for these or generic circuits. For example, one could consider Clifford proximity or entanglement content in the resulting states and tailor simulation appropriately. Focusing on the $k=1$ HW space, we get \emph{orthogonal} quantum neural networks (OrthoQNNs)~\cite{landman_quantum_2022, cherrat_quantum_2022, cherrat_quantum_2023}, which have been useful in mitigating unstable gradient dynamics for classical models and may offer a cubic to quadratic speedup in parameterising and training (classical) orthogonal neural networks~\cite{landman_quantum_2022}. 
 
Hamming weight preserving unitaries (OrthoQNNs as a special case)~\cite{landman_quantum_2022, kiani_projunn_2022, hamze_parallelized_2021} can be built from \emph{reconfigurable beam splitter} (RBS) gates, or Givens' rotations, which have the following (two qubit) form:
\begin{multline} \label{eqn:rbs_gate_definition}
    RBS(\theta) := e^{-i\frac{\theta}{2}G_{RBS}}  :=  e^{-i\frac{\theta}{2}\left(Y \otimes X - X \otimes Y\right)} \\ 
    = \left(
    \begin{array}{cccc}
    1 & 0 & 0 & 0 \\
    0 & \cos(\theta) & -\sin(\theta) & 0 \\
    0 & \sin(\theta) & \cos(\theta) & 0 \\
    0 & 0 & 0 & 1 
    \end{array}
    \right)
\end{multline}
Different forms of OrthoQNNs can be defined, each parameterising a (potentially restricted) orthogonal matrix shown in~\figref{fig:ortho_qnn_architectures}. These are the explicit `pyramid', `X' and  `butterfly' ans\"atze used above and shown in~\figref{fig:ortho_qnn_architectures}. Replacing the $RBS$ gates in these circuits with a so-called \emph{fermionic} beam splitter ($FBS$) gate gives \emph{compound} QNNs, which we discuss in~Supp.~Mat. E.2. These compound QNNs can be represented by compound matrices acting on higher Hamming weight ($k>1$) or superpositions thereof.

\begin{figure}
    \centering
    \includegraphics[width=\columnwidth, height=0.6\columnwidth]{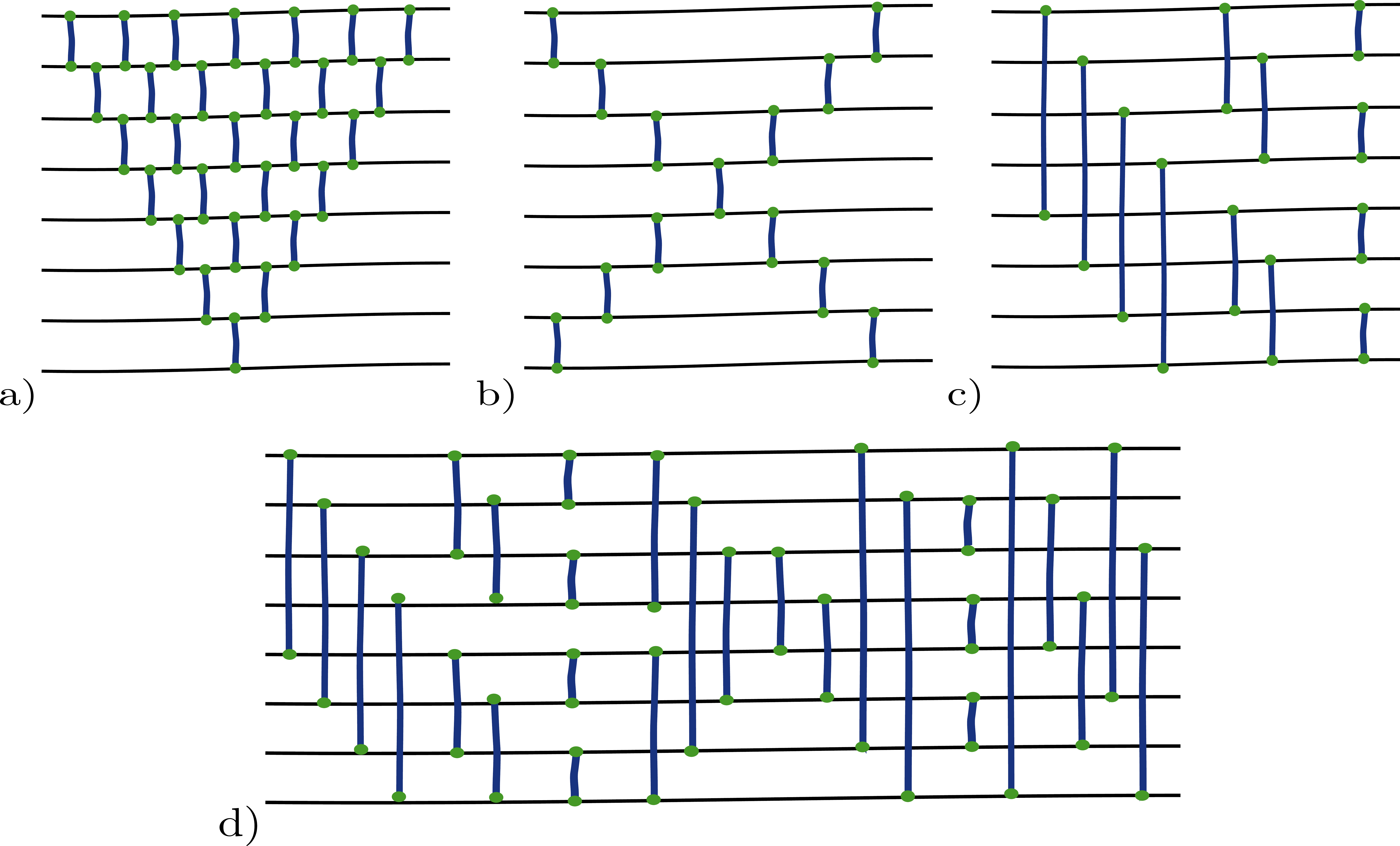}
    \caption{\textsf{
    \textbf{Ans\"atze for orthogonal quantum neural networks~\cite{landman_quantum_2022, cherrat_quantum_2022}}. a) Pyramid circuit, b) X circuit, c) Butterfly circuit and d) Round-robin circuit. Each gate corresponds to an $RBS$ gate with (potentially different) parameter $\theta$. With respect to the number of qubits, $n$, the depths of each of these layers is $2n-1, n-1, \log(n)$ and $n-1$ respectively.}
    }
    \label{fig:ortho_qnn_architectures}
\end{figure}

\subsection{Data reuploading} \label{sec:data_reuploading}
Data reuploading~\cite{perez-salinas_data_2020} quantum neural networks incorporate unitaries which inject the \emph{same} data into the parameterised state multiple times. This allows the learned functions to be highly non-linear in the data, and can produce universal classifiers. Further, it is well known that such models can be expressed as \emph{partial Fourier series} in the data,~\cite{schuld_effect_2021, landman_classically_2023, mhiri_constrained_2024}. A single reuploading layer, $\ell$, can be written as follows:
\begin{equation}\label{eqn:single_reup_layer}
    \mathcal{E}^{\ell}_{\boldsymbol{\theta}^{\ell}, \boldsymbol{x}}(\rho) := U(\boldsymbol{\theta}^{\ell})V_{\ell}(\boldsymbol{x}) \rho V_{\ell}^{\dagger}(\boldsymbol{x})U^\dagger(\boldsymbol{\theta}^{\ell})
\end{equation}
where $\rho$ is some data-independent initial state, usually $\rho = \ketbra{0}{0}^{\otimes n}$.
Then, the function output from $L$ such uploads is: 
\begin{align}\label{eqn:data_reuploading_full}
\mathcal{E}^{L}_{\boldsymbol{\theta}, \boldsymbol{x}}(\rho) &:=   \mathcal{E}^{L}_{\boldsymbol{\theta}^L, \boldsymbol{x}} \left( \cdots \mathcal{E}^{2}_{\boldsymbol{\theta}^{2}, \boldsymbol{x}}\left(\mathcal{E}^{1}_{\boldsymbol{\theta}^{1}, \boldsymbol{x}}(\rho)\right)\right)\\
    f(\boldsymbol{\theta}, \boldsymbol{x}) &= \Tr\left(\mathcal{O}\mathcal{E}^{L}_{\boldsymbol{\theta}, \boldsymbol{x}}(\rho)\right) = \sum_{\boldsymbol{\omega}\in \Omega} c_{\boldsymbol{\omega}}(\boldsymbol{\theta}) e^{i \boldsymbol{\omega}^\top \boldsymbol{x}}
\end{align}
which is a sum of Fourier coefficients with frequencies, $\Omega := \{\boldsymbol{\omega}\}$. Here, it is assumed that the encoding unitaries are \emph{Hamiltonian encodings} of the data, $V_{\ell}(\boldsymbol{x}) = \prod\nolimits_{d=1}^{D} \exp\left( i \boldsymbol{x}_d H^\ell_d\right)$ for some (potentially highly non-local) Hermitian operators, $\{H^\ell_d\}_{d, \ell}$.

Now, it is known that the expressivity of such models are determined by the eigenspectrum of the Hamiltonians used in the encoding unitaries. These generate the Fourier frequencies available to the model, $\Omega$, while the trainable parameters in the unitaries, $\boldsymbol{\theta}$, determine the coefficients $c_{\boldsymbol{\omega}}(\boldsymbol{\theta})$ which dictate how these frequencies are combined. It is straightforward to generalise this model into a density counterpart, but we explicitly do so in Supp. Mat. G.3 for completeness.

\subsection{Experiment details} 
\subsubsection{Equivariant density QNN} 
\label{ssec:equivariant_methods}
The bars and dots dataset is as follows. Each datapoint is a $d$-dimensional vector (bar or dot) with either alternating $+1$ or $-1$ values (dot) or sequential periods of $+1$ or $-1$ of length $\lfloor\frac{d}{2}\rfloor$ (bar). Gaussian noise with mean $=0$ and variance $\sigma^2$ is added to each vector. In Ref.~\cite{bowles_backpropagation_2023}, the translation invariance of the data enables the application of a commuting-generator equivariant ansatz, where each generator consists of a symmetrised Pauli-$X$ string containing up to $K$-body terms. The measurement observable is also a symmeterised Pauli  $Z$  string with $K=1$, meaning $\mathcal{H} = \sum_{i=1}^d Z_i$. Each bar/dot, $\boldsymbol{x} := [x_1, \dots, x_d]^\top \in \mathbb{R}^d$, is encoded as a Pauli Y rotation per qubit, $x_i \mapsto R_y(\frac{\pi x_i}{2})\ket{0} = \cos(\frac{\pi x_i }{4} )\ket{0} + \sin(\frac{\pi x_i }{4})\ket{1} =: \ket{\boldsymbol{x}}^{ry}_i$. This angle encoding uses the same number of qubits as the unary amplitude encoding in the OrthoQNN above. 

For the numerics in~\figref{fig:equivariant_model_numerics}, we use $10$ qubits in simulation. We increase the noise to $\sigma=1.8$ to raise the problem difficulty. We train the model with train and test data of sizes of $1000, 100$ respectively and a batch size of $20$. The Adam optimiser with a learning rate of $0.001$ is used in all cases. We initialise the weighting parameters, $\{\alpha_1, \alpha_2\} = \{\alpha_{XX}, \alpha_{YY}\} = \{0.99, 0.01\}$ to bias the model towards the (pre-trained) equivariant XX model, and $\boldsymbol{\alpha}$ are also trainable. The density model is compared to its commuting $XX$ counterpart, plus a non-commuting model (a model which does not obey the commuting generator properties) and the quantum convolutional neural network (QCNN)~\cite{cong_quantum_2019} with the same structures as described in Ref.~\cite{bowles_backpropagation_2023}. The formulae to compute the number of shots required by each model (commuting $XX$, non-commuting, QCNN) is given in Ref.~\cite{bowles_backpropagation_2023}. The number of shots for the density QNN $XX + YY$ model being $2N_{XX} + 2N_{YY}$ where $N_{XX/YY}$ is the number of shots to train the XX, YY models separately (over $500$ epochs for each model). We adapt the code of~\cite{bowles_josephbowlesbackprop_scaling_2023} to generate the results for the commuting, non-commuting and QCNN models which is built using Pennylane~\cite{bergholm_pennylane_2022}.

\subsubsection{Orthogonal density QNN} \label{ssec:orthogonal}
For all the density orthogonal QNN results, we take $60,000$ train and $10,000$ test $28\times 28$ pixel images from the MNIST dataset. For the round-robin decomposition (\figref{fig:round_robin_density_decomposition}) each MNIST image is downscaled to $10$ qubits using principal component analysis.

For simplicity in simulation, we compute the action of each sub-unitary on the data state, $\ket{\boldsymbol{x}}$ (which is the data-loaded feature vector using the unitary $V(\boldsymbol{x})$), and classically combine them. This bypasses the need to explicitly reconstruct the density matrix, $\rho(\boldsymbol{\theta}, \boldsymbol{\alpha}, \boldsymbol{x})$, though we stress in a real quantum implementation one would apply the randomisation over sub-unitaries, or create the density sta te explicitly on the quantum computer as in~\figref{fig:density_qnn_mainfig}. There are various choices one could make to create the data state, $\ket{\boldsymbol{x}}$, for all the above results we assume a \emph{parallel} vector loader~\cite{johri_nearest_2021, cherrat_quantum_2022} but the choice of loader will depend on the available quantum hardware. 

With the addition of the data dependent mixture of experts (MoE) gating network to predict the distribution of sub-unitary coefficients in~\figref{fig:round_robin_density_decomposition}, the following parameterisation is chosen. Specifically, ``MoE'' refers to coefficients defined as a gating network with the following form $\alpha_k(\boldsymbol{x}) = \texttt{softmax}_k(\texttt{Linear}(\boldsymbol{x}))$. We investigate the addition of a nonlinearity in~Supp.~Mat. H.2). Also, we could replace the input to the gating network to be a more general feature vector, $\boldsymbol{x} \rightarrow \sigma(\boldsymbol{x})$, i.e. the output of another neural network.

For all numerical results involving orthogonal quantum neural networks, we used QCLearn, a custom in-house quantum machine learning software package with model tailored classical simulation, primarily built in JAX~\cite{bradbury_jax_2018}.

\subsubsection{Data reuploading density QNN} \label{ssec:reuploading_methods}

For the numerical results demonstrating overfitting mitigation in\figref{fig:data_reuploading_overfitting} with data reuploading, we generate training and test data from the family of Chebyshev polynomials of the first kind, $\{T_n(x)\}$, and compare each models ability to learn the underlying function, defined as:

\begin{multline}\label{eqn:chebyshev_polynomials}
    T_0(x) = 1, T_1(x) = x, \\ T_{n+1}(x) = 2xT_n(x) - T_{n-1}(x), 
\end{multline}

We choose the ground truth function to be $h(x) := T_{4}(x)$ and add Gaussian noise with mean $0$ and variance $0.25$ to the training data to nudge both models towards visible overfitting.  The arbitrary single qubit operation we use for reupload $\ell$ in both reuploading models (density/vanilla) is $U_{\ell}(\boldsymbol{\theta}_{\ell}) = R_z(\theta^1_{\ell})R_y(\theta^2_{\ell})R_z(\theta^3_{\ell})$. We also again use Pennylane~\cite{bergholm_pennylane_2022} for these results.

\section{Data Availability} \label{sec:data_availibility}
The data corresponding to this manuscript is available upon reasonable request.

\section{Code Availability} \label{sec:code_availibility}
The code corresponding to this manuscript is available upon reasonable request.

\section{Author contributions} \label{sec:contributions}
The idea was conceived by B.C.. All authors contributed to the code for numerical results. Numerical experiments performed by B.C., N. M. and S. R.. B. C., N. M., S. R. and I.K. contributed to the theoretical results. I.K. supervised the project.

\section{Competing interests} \label{sec:competing_interests}
The authors declare no competing interests.



\appendix

\onecolumn

\addcontentsline{toc}{section}{Appendix} 
\renewcommand{\thepart}{}
\renewcommand{\partname}{}

\part{Supplementary Material}

\parttoc %

\titleformat{\section}[block]{\normalfont\Large\bfseries}{\thesection}{1em}{}
\titleformat{\subsection}[block]{\normalfont\large\bfseries}{\thesubsection}{1em}{}
\titleformat{\subsubsection}[block]{\normalfont\normalsize\bfseries}{\thesubsubsection}{1em}{}

\section{Quantum machine learning frameworks}\label{app_ssec:other_qml_frameworks}

In the main text, we introduced the density QNN as a \emph{framework}. Before proceeding, it is useful to make a distinction between a quantum machine learning (QML) model \emph{framework}, \emph{specification} and \emph{instance} as these concepts are frequently used interchangeably. Depending on preference, frameworks may be more or less concrete. For example, in Ref.~\cite{caro_generalization_2022}, a quantum machine learning model (QMLM) refers to the collection of (discrete and continuous) trainable parameters, $\boldsymbol{\theta}$ with the only a-priori  restriction being the existence of \emph{some} mapping from parameters to a (completely positive trace preserving) quantum channel, $\boldsymbol{\theta}\rightarrow \mathcal{E}_{\boldsymbol{\theta}}$. However, it will be useful to define a \emph{framework} with slightly more structure. This allows for comparison between different model families, each of which may have complementary advantages. The \emph{specification} of the QML framework refers to the choice of quantum circuit \emph{ans\"{a}tze} - a mapping from the parameters to a set of circuit operations, $\boldsymbol{\theta} \rightarrow \Sigma(\boldsymbol{\theta})$. Usually, $\Sigma$ will contain a set of parametrised unitaries, $\{U_k(\phi_k)\}_{k=1}^K$, measurement observables, $\{\mathcal{O}_j(\psi_j)\}_{j=1}^J$ and functions, $f(\boldsymbol{\alpha}), g(\boldsymbol{\beta}), \dots$ which dictate quantities such as control flow and variable structure within the model. The combination of unitaries, observable and classical functions are combined with the restriction of forming the general channel, $\mathcal{E}_{\boldsymbol{\theta}}$, acting on quantum or classical data, $\mathcal{X} = \{\rho, \rho(\boldsymbol{x})\}$, to produce
$\mathcal{E}_{\boldsymbol{\theta}}(\mathcal{X})$. Finally, this is generally measured with some (potentially different to those above) observable to produce outputs $\Tr(\mathcal{O}\mathcal{E}_{\boldsymbol{\theta}}(\mathcal{X}))$. However, note in the case of classical data, the model will also include the operations which load the data into the state, $\boldsymbol{x} \rightarrow \rho(\boldsymbol{x})$, and these also may have trainable parameters. 

Finally, an \emph{instance} is the fixing of a particular parameter setting, $\boldsymbol{\theta}^* \rightarrow \Sigma(\boldsymbol{\theta}^*)$, usually used for inference modes of the model when trained.

\section{Proofs} \label{app_sec:proofs}
In this section we give the explicit proofs from the theoretical results in the main text. 

\subsection{Proof of\texorpdfstring{~\propref{prop:gradient_scaling_dqnn}}{}} \label{app_ssec:proof_dqnn_gradient_scaling}
\begin{propositionunnum}[Gradient scaling for density quantum neural networks (\textbf{\propref{prop:gradient_scaling_dqnn} repeated})]\label{app_thm:gradient_scaling_dqnn_app}
Given a density QNN as in~\eqref{eqn:density_qnns} composed of $K$ sub-unitaries, $\mathcal{U} = \{U_k(\boldsymbol{\theta}_k)\}_{k=1}^K$, implemented with distribution, $\boldsymbol{\alpha} = \{\alpha_k\}_{k=1}^K$, an unbiased estimator of the gradients of a loss function, $\mathcal{L}$, defined by a Hermitian observable, $\mathcal{H}$:
\begin{equation} \label{eqn:density_qnn_loss_fn_app}
    \mathcal{L}(\boldsymbol{\theta}, \boldsymbol{\alpha}, \boldsymbol{x}) = \Tr\Big(\mathcal{H}\rho(\boldsymbol{\theta}, \boldsymbol{\alpha}, \boldsymbol{x})\Big)
\end{equation}
can be computed by classically post-processing $\sum_{l=1}^K\sum_{k=1}^K T_{\ell k}$ circuits, where $T_{\ell k}$ is the number of circuits required to compute the gradient of sub-unitary $k$, $U(\boldsymbol{\theta}_k)$ with respect to the parameters in sub-unitary $\ell$, $\boldsymbol{\theta}_{\ell}$. Furthermore, these parameters can also be shared across the unitaries, $\boldsymbol{\theta}_k = \boldsymbol{\theta}_{k'}$ for some $k, k'$. 
\end{propositionunnum}

\begin{proof}
Assume for simplicity that the number of parameters in each sub-unitary is the same, $B_k = B ~ \forall k$ and $N_{B_k} = N_{B} = N$. Furthermore assume the number of blocks is $B=1$. Then for notational purposes we can write the following $K \times K \times N$ tensor, with the $j^{th}$ `slice' across the last dimension being:
\begin{equation} \label{eqn:sub_gradient_matrix}
[\Delta \mathcal{L}(\boldsymbol{\theta}, \boldsymbol{\alpha}, \boldsymbol{x})]_{j} := 
 \left(
\begin{array}{ccc}
    \partial_{j1}\mathcal{L}'(\boldsymbol{\theta}_1, \boldsymbol{x})&  \cdots &  \partial_{jK}\mathcal{L}'(\boldsymbol{\theta}_1, \boldsymbol{x}) \\
     \vdots & \ddots & \vdots \\
     \partial_{j1}\mathcal{L}'(\boldsymbol{\theta}_K, \boldsymbol{x}) & \cdots    &\partial_{jK}\mathcal{L}'(\boldsymbol{\theta}_K, \boldsymbol{x})
\end{array} \right) 
\end{equation}
where $\mathcal{L}'$ is the loss function evaluated using only a single term of the density sum:
\begin{align*}
\mathcal{L}'(\boldsymbol{\theta}_k, \boldsymbol{x}) &:= \Tr\Big(\mathcal{H}U_k(\boldsymbol{\theta}_k)\ketbra{\boldsymbol{x}}{\boldsymbol{x}}U^\dagger_k(\boldsymbol{\theta}_k)\Big) \\
\implies 
\partial_{\ell}\mathcal{L}'(\boldsymbol{\theta}_k, \boldsymbol{x}) &:= \frac{\partial \mathcal{L}'(\boldsymbol{\theta}_k, \boldsymbol{x})}{\partial \boldsymbol{\theta}_{\ell}}, 
\partial_{j\ell}\mathcal{L}'(\boldsymbol{\theta}_k, \boldsymbol{x}) := \frac{\partial \mathcal{L}'(\boldsymbol{\theta}_k, \boldsymbol{x})}{\partial \theta^j_{\ell}}
\end{align*}
In other words, the $j^{th}$ slice of $\Delta \mathcal{L}(\boldsymbol{\theta}, \boldsymbol{\alpha}, \boldsymbol{x})$ is a matrix, where the rows and columns are indexed by the terms in the model~\eqref{eqn:density_qnns} - the diagonal terms are the gradients of the $k^{th}$ sub-unitary with respect to the $k^{th}$ set of parameters, while the off-diagonal terms are the gradients of the $k^{th}$ term with respect to the $\ell^{th}$ ($\ell \neq k$) set of parameters.

Now, we can plug in the definition of the model~(\eqref{eqn:density_qnns}), taking the gradient w.r.t the $k^{th}$ sub-unitaries parameters (a vector of size $N$):
\begin{multline} \label{eqn:density_gradient_computation}
    \left[\frac{\partial \mathcal{L}(\boldsymbol{\theta}, \boldsymbol{\alpha}, \boldsymbol{x})}{\partial \boldsymbol{\theta}_{\ell}}\right]_j = \frac{\partial \Tr\Big(\mathcal{H}\rho(\boldsymbol{\theta}, \boldsymbol{\alpha}, \boldsymbol{x})\Big)}{\partial \theta^j_{\ell}} 
    = \sum_{k=1}^K \alpha_k  \frac{\partial \Tr\Big(\mathcal{H} U_k(\boldsymbol{\theta}_k)\ketbra{\boldsymbol{x}}{\boldsymbol{x}}U^\dagger_k(\boldsymbol{\theta}_k)\Big)}{\partial \theta^j_{\ell}}  \\
    = \sum_{k=1}^K \alpha_k \partial_{j\ell}\mathcal{L}'(\boldsymbol{\theta}_k, \boldsymbol{x}) = \sum_{k=1}^K \alpha_k [[\Delta \mathcal{L}(\boldsymbol{\theta}, \boldsymbol{\alpha}, \boldsymbol{x})]_{j}]_{\ell, k}
\end{multline} 
where $[[\Delta \mathcal{L}(\boldsymbol{\theta}, \boldsymbol{\alpha}, \boldsymbol{x})]_{j}]_{\ell, k}$ is the $\ell, k$ element of~\eqref{eqn:sub_gradient_matrix}.
Hence, assuming we can compute the $k^{th}$ column of $\Delta \mathcal{L}(\boldsymbol{\theta}, \boldsymbol{\alpha}, \boldsymbol{x})$ with $T_k$ circuits, we can estimate~\eqref{eqn:density_gradient_computation} by computing the gradient with respect to each sub-unitary $U_k$ and summing the resulting (weighted by $\alpha_k$) vectors. 

Finally, relating to the unbiasedness of the gradient estimator, this follows directly from Ref.~\cite{sweke_stochastic_2020} where an expectation value can be constructed for each $i, j, k$ from $S$ measurements which has mean $[[\Delta \mathcal{L}(\boldsymbol{\theta}, \boldsymbol{\alpha}, \boldsymbol{x})]_{j}]_{\ell, k}$.
\end{proof}
As discussed in the main text, there are two sub-cases one can consider. First, if all parameters between sub-unitaries are independent, $\boldsymbol{\theta}_k \neq \boldsymbol{\theta}_{\ell}, ~\forall k, \ell$. Here, the computation is simpler as taking a gradient with respect to the parameters of sub-unitary, $U_{k}$, results in all other columns of \eqref{eqn:sub_gradient_matrix} vanishing, i.e. $[[\Delta \mathcal{L}(\boldsymbol{\theta}, \boldsymbol{\alpha}, \boldsymbol{x})]_{j}]_{\ell, k} = 0,~\forall \ell \neq k$. Hence, we only need to extract the diagonal terms from~\eqref{eqn:sub_gradient_matrix}, $\partial_{jk}\mathcal{L}'(\boldsymbol{\theta}_k, \boldsymbol{x})$ for each $j$. This observation then gives~\corrref{corr:density_qnn_gradient_independent} in the main text.

\subsection{Proof of\texorpdfstring{~\lemref{lemma:supervised_mixing_corr}}{}} \label{app_sec:density_mixing_proof}
Here we prove~\lemref{lemma:supervised_mixing_corr} from the main text, repeated here for convenience:
\begin{lemmaunnum}[Mixing lemma for supervised learning (\textbf{\lemref{lemma:supervised_mixing_corr} repeated})]\label{thm:supervised_mixing_corr_app}
Let $h(\boldsymbol{x})$ be a target ground truth function, prepared via the application of a fixed unitary, $V$, $h(\boldsymbol{x}) := \Tr(\mathcal{O}V\rho(\boldsymbol{x}) V^{\dagger})$ on a data encoded state, $\rho(\boldsymbol{x})$ and measured with a fixed observable, $\mathcal{O}$. Suppose there exists $K$ unitaries $\{U_k(\boldsymbol{\theta})\}_{k=1}^K$ such that these each are $\delta_1$ good predictive models of $h(\boldsymbol{x})$:
    \begin{equation} \label{eqn:supervised_mixing_corrollary_condition_1_app}
        \mathbb{E}_{\boldsymbol{x}}|h(\boldsymbol{x}) - f_k(\boldsymbol{\theta}, \boldsymbol{x})| \leq 
        \delta_1, \forall k
    \end{equation}
    and a distribution $\{\alpha_k\}_{k=1}^K$ such that predictions according to the LCU model $f_{\textsf{LCU}}(\boldsymbol{\theta}, \boldsymbol{\alpha}, \boldsymbol{x}) := \Tr(\mathcal{O}\Big(\sum_k \alpha_k U_k(\boldsymbol{\theta})\Big) \rho(\boldsymbol{x})\Big(\sum_k \alpha_k U^{\dagger}_k(\boldsymbol{\theta})\Big))$ have error  bounded as:
    \begin{equation} \label{eqn:supervised_mixing_corrollary_condition_2_app}
        \mathbb{E}_{\boldsymbol{x}}|h(\boldsymbol{x}) - f_{\textsf{LCU}}(\boldsymbol{\theta}, \boldsymbol{\alpha}, \boldsymbol{x})| \leq \delta_2 
    \end{equation}
    for some $\delta_1, \delta_2 >0$. Then, the corresponding density QNN, $(\rho(\boldsymbol{\theta}, \boldsymbol{\alpha}, \boldsymbol{x})) = \sum_{k=1}^K\alpha_k U_k^\dagger(\boldsymbol{\theta}) \rho(\boldsymbol{x})U_k(\boldsymbol{\theta})$ can generate predictions for $h(\boldsymbol{x})$ with error:
    \begin{equation} \label{eqn:supervised_mixing_corollary_app}
         \mathbb{E}_{\boldsymbol{x}}|h(\boldsymbol{x}) - f(\boldsymbol{\theta}, \boldsymbol{\alpha}, \boldsymbol{x})| \leq \frac{\delta_1^2}{4\|\mathcal{O}\|_{\infty}} + 2\delta_2
    \end{equation} 
\end{lemmaunnum}

\begin{proof}

Firstly, we have the Mixing lemma itself:
\begin{lemma}[Hastings-Campbell Mixing lemma]
    Let $V$ be a target unitary with corresponding channel $\mathcal{V}(\rho(\boldsymbol{x})) = V\rho V^{\dagger}$. Suppose there exists $K$ unitaries $\{U_k(\boldsymbol{\theta})\}_{k=1}^K$ such that:
    \begin{equation} \label{eqn:mixing_lemma_condition_1}
        \|U_k -V\| \leq \varepsilon_1
    \end{equation}
    and a distribution $\{\alpha_k\}_{k=1}^K$ such that:
    \begin{equation} \label{eqn:mixing_lemma_condition_2}
        \left\|\sum_k \alpha_k U_k - V\right\| \leq \varepsilon_2
    \end{equation}
    for some $\varepsilon_1, \varepsilon_2 >0$. Then, the corresponding channel, $\Lambda(\rho) = \sum_{k=1}^K\alpha_k U_k \rho U_k^\dagger$ approximates $\mathcal{V}$ as:
    \begin{equation} \label{eqn:mixing_lemma_result}
        \left\|\Lambda(\rho) - \mathcal{V} \right\|_{\diamond} \leq \varepsilon^2_1 + 2\varepsilon_2
    \end{equation}
\end{lemma}

In a general QML form, these conditions can be immediately uplifted to the following, assuming $\rho$ is data dependent, and the sub-unitaries are trainable:
    \begin{equation} \label{eqn:mixing_lemma_result_data_dep_trainable_unitary}
        \|U_k(\boldsymbol{\theta}_k) -V\| \leq \varepsilon_1, \left\|\sum_k \alpha_k U_k(\boldsymbol{\theta}_k) - V\right\| \leq \varepsilon_2 \implies \left\|\Lambda(\rho(\boldsymbol{x})) - \mathcal{V}(\rho(\boldsymbol{x})) \right\|_{\diamond} \leq \varepsilon^2_1 + 2\varepsilon_2, \forall \boldsymbol{x}
    \end{equation}
Now, in supervised learning scenarios we are interested in minimising the prediction or generalisation error, between the ground truth function, $h(\boldsymbol{x})$, and the actual prediction of our model, $f(\boldsymbol{\theta}, \boldsymbol{x})$. The error then for a single input is $|h(\boldsymbol{x}) - f(\boldsymbol{\theta}, \boldsymbol{x})| = |h(\boldsymbol{x}) - \Tr(\mathcal{O}\rho(\boldsymbol{\theta}, \boldsymbol{x}))|$, if we assume $ f(\boldsymbol{\theta}, \boldsymbol{x})$ is a standard linear quantum model. 

The full error is taken as the expectation of this quantity over the entire data distribution $\underset{\boldsymbol{x}}{\mathbb{E}}|h(\boldsymbol{x}) - \Tr(\mathcal{O}\rho(\boldsymbol{\theta}, \boldsymbol{x}))|$. We assume as above that $\rho(\boldsymbol{\theta}, \boldsymbol{x}) = U(\boldsymbol{\theta})\ketbra{\boldsymbol{x}}{\boldsymbol{x}}U^{\dagger}(\boldsymbol{\theta})$ where $\ket{\boldsymbol{x}}$ is some encoding of the data.

Now, assuming we can embed the ground truth also into the evaluation of an expectation value such that $h(\boldsymbol{x}) = \Tr(\mathcal{O}V \ketbra{\boldsymbol{x}}{\boldsymbol{x}}V^{\dagger})$ with the same data encoding as above, for some $V$. 

In this simplified scenario, a solution is to optimise the parameters, $\boldsymbol{\theta}$, such that $U(\boldsymbol{\theta}) = V$ which is exactly the quantum compilation strategy. Slightly more generally, we could have different observables for the ground truth and model in which case we optimise, $\boldsymbol{\theta}$ such that  $f(\boldsymbol{\theta}, \boldsymbol{x}) = \Tr(\mathcal{O}U(\boldsymbol{\theta})\ketbra{\boldsymbol{x}}{\boldsymbol{x}}U^{\dagger}(\boldsymbol{\theta}))$ is close to $ h(\boldsymbol{x}) = \Tr(\mathcal{O}'V \ketbra{\boldsymbol{x}}{\boldsymbol{x}}V^{\dagger})$. 

Returning to the simplified scenario, and assuming the first condition of the Mixing lemma is satisfied - we have trained, individually, an ensemble of models $\{U(\boldsymbol{\theta}_k)\} \rightarrow f_k(\boldsymbol{\theta}, \boldsymbol{x})$ such that $\|U_k - V\| \leq \varepsilon_1$, where $V$ represents the ground truth as above. Now, we have:
\begin{align}
     |h(\boldsymbol{x}) -f_k(\boldsymbol{\theta}, \boldsymbol{x})| = |h(\boldsymbol{x}) - \Tr(\mathcal{O}\rho_k(\boldsymbol{\theta}, \boldsymbol{x}))|  
    &= \left|\Tr(\mathcal{O}U_k(\boldsymbol{\theta})\ketbra{\boldsymbol{x}}{\boldsymbol{x}}U^{\dagger}_k(\boldsymbol{\theta})) - \Tr(\mathcal{O}V \ketbra{\boldsymbol{x}}{\boldsymbol{x}}V^{\dagger})\right| \label{eqn:sup_learn_error_proof_1} \\
    &=  \left|\Tr\left(\mathcal{O}\left(U_k(\boldsymbol{\theta})\ketbra{\boldsymbol{x}}{\boldsymbol{x}}U_k^{\dagger}(\boldsymbol{\theta}) - V \ketbra{\boldsymbol{x}}{\boldsymbol{x}}V^{\dagger}\right)\right)\right|\label{eqn:sup_learn_error_proof_2}\\
    &\leq \|\mathcal{O}\|_{\infty}\cdot \|U_k(\boldsymbol{\theta})\ketbra{\boldsymbol{x}}{\boldsymbol{x}}U^{\dagger}_k(\boldsymbol{\theta}) - V \ketbra{\boldsymbol{x}}{\boldsymbol{x}}V^{\dagger}\|_{1} \label{eqn:sup_learn_error_proof_3}\\
    &\leq 2 \|\mathcal{O}\|_{\infty}\cdot\|U_k(\boldsymbol{\theta}) - V\|_{1}\label{eqn:sup_learn_error_proof_4}\\
    &\leq 2 \varepsilon_1\|\mathcal{O}\|_{\infty}, \qquad \forall k \label{eqn:sup_learn_error_proof_5}
\end{align}
where the first inequality follows from H\"{o}lders inequality and the second can be proven using the expansion $U\rho U^{\dagger} - V\rho V^{\dagger} = U\rho U^{\dagger} - V\rho U^{\dagger} + V\rho U^{\dagger} - V\rho V^{\dagger}$, the triangle inequality and the fact that $\|\rho\|_1 = \Tr(|\rho|) \leq 1$ for density matrices. Alternatively see Lemma B.5 of~\cite{caro_generalization_2022}.

Next, we assume that we have found a suitable distribution $\boldsymbol{\alpha} = \{\alpha_k\}$ such that $\|\sum_{k=1}^K\alpha_k U_k(\boldsymbol{\theta}) - V\| \leq \varepsilon_2$. Using this linear combination of the unitaries (LCU) now as the `enhanced' model, $f_{\textsf{LCU}}(\boldsymbol{\theta}, \boldsymbol{\alpha}, \boldsymbol{x}) = \Tr(\mathcal{O}\left(\sum_{k=1}^K\alpha_k U_k(\boldsymbol{\theta})\right) \rho \left(\sum_{k'=1}^K\alpha_{k'} U^{\dagger}_{k'}(\boldsymbol{\theta})\right))$. Here, the (non-physical) state $\mathcal{U}_{\textsf{LCU}}(\rho) := \left(\sum_{k=1}^K\alpha_k U_k(\boldsymbol{\theta})\right) \rho(\boldsymbol{x}) \left(\sum_{k'=1}^K\alpha_{k'} U^{\dagger}_{k'}(\boldsymbol{\theta})\right)$ is the channel formed by applying the LCU to the pure state, $\sum_{k=1}^K\alpha_k U_k(\boldsymbol{\theta}) \ket{\boldsymbol{x}}$. This can be done directly using the standard LCU circuits from~\figref{fig:density_qnn_mainfig}a.

Using similar logic to the above, we have:
\begin{align}
    |h(\boldsymbol{x}) - f'(\boldsymbol{\theta}, \boldsymbol{\alpha}, \boldsymbol{x})| &= \left|h(\boldsymbol{x}) - \sum_{k, k'} \alpha_k\alpha_{k'}  \Tr(\mathcal{O}\rho_{k, k'}(\boldsymbol{\theta}, \boldsymbol{x}))\right|  \label{eqn:sup_learn_error_proof_combo_1}\\
    &\leq \|\mathcal{O}\|_{\infty}\cdot \left\|\left(\sum_k \alpha_k U_k(\boldsymbol{\theta})\right)\ketbra{\boldsymbol{x}}{\boldsymbol{x}} \left(\sum_{k'} \alpha_{k'} U^{\dagger}_k(\boldsymbol{\theta})\right) - V \ketbra{\boldsymbol{x}}{\boldsymbol{x}}V^{\dagger}\right\|_{1} \label{eqn:sup_learn_error_proof_combo_2}\\
     &= \|\mathcal{O}\|_{\infty}\cdot \left\|(\mathcal{U}_{\textsf{LCU}} - \mathcal{V}) (\ketbra{\boldsymbol{x}}{\boldsymbol{x}})\right\|_{1} \label{eqn:sup_learn_error_proof_combo_3}\\
     &\leq 2\|\mathcal{O}\|_{\infty} \cdot \left\| \sum_{k=1}^K\alpha_k U_k(\boldsymbol{\theta}) - V\right\|
     \leq 2\varepsilon_2 \|\mathcal{O}\|_{\infty}
\end{align}
Putting these two together, we have the two assumptions from the Mixing lemma recast into a supervised learning scenario:
\begin{align}
    |h(\boldsymbol{x}) - f_k(\boldsymbol{\theta}, \boldsymbol{x})| &\leq 2\varepsilon_1  \|\mathcal{O}\|_{\infty} =: \delta_1 , \forall k, \boldsymbol{x} \\
    |h(\boldsymbol{x}) - f_{\textsf{LCU}}(\boldsymbol{\theta}, \boldsymbol{\alpha}, \boldsymbol{x})| &\leq 2\varepsilon_2  \|\mathcal{O}\|_{\infty} =: \delta_2, \forall \boldsymbol{x}
\end{align}
Now, we can recast the result of the Mixing lemma itself into the supervised scenario. From the lemma we have that $\left\|\Lambda_{\boldsymbol{\theta}, \boldsymbol{\alpha}} - \mathcal{V} \right\|_{\diamond} \leq \varepsilon^2_1 + 2\varepsilon_2$ where $\Lambda_{\boldsymbol{\theta}, \boldsymbol{\alpha}}$ is the `density' channel, $\Lambda_{\boldsymbol{\theta}, \boldsymbol{\alpha}}(\rho(\boldsymbol{x})) = \rho(\boldsymbol{\theta}, \boldsymbol{\alpha}, \boldsymbol{x}) = \sum_{k=1}^K \alpha_k U_k(\boldsymbol{\theta}_k)\ketbra{\boldsymbol{x}}{\boldsymbol{x}}U^\dagger_k(\boldsymbol{\theta}_k)$. Taking the output of the density model to be $f(\boldsymbol{\theta}, \boldsymbol{\alpha}, \boldsymbol{x}) = \Tr\left(\mathcal{O}\Lambda_{\boldsymbol{\theta}, \boldsymbol{\alpha}}(\rho(\boldsymbol{x})\right))$, we have:
\begin{align}
    |h(\boldsymbol{x}) - f(\boldsymbol{\theta}, \boldsymbol{\alpha}, \boldsymbol{x})| &= \left|h(\boldsymbol{x}) - \sum_{k} \alpha_k  \Tr(\mathcal{O}\rho_{k}(\boldsymbol{\theta}, \boldsymbol{x}))\right| \\
    & \leq \|\mathcal{O}\|_{\infty}\cdot \left\|\left(\sum_k \alpha_k U_k(\boldsymbol{\theta})\ketbra{\boldsymbol{x}}{\boldsymbol{x}} U^{\dagger}_k(\boldsymbol{\theta})\right) - V \ketbra{\boldsymbol{x}}{\boldsymbol{x}}V^{\dagger}\right\|_{1} \\
    & \leq \|\mathcal{O}\|_{\infty} (\varepsilon^2_1 + 2\varepsilon_2)
\end{align}
The last inequality follows directly from the definition of the diamond norm and the Mixing lemma - since the lemma holds for the diamond norm, it holds also for all states, $\rho$ with $\|\rho\|_{1} \leq 1$:  $\|( \Lambda\otimes \mathds{I} - \mathcal{V} \otimes \mathds{I})(\rho)\|_{1} \leq \varepsilon_1^2 + 2\varepsilon_2$. Note the diamond norm includes a tensor product of the channel with the identity. As mentioned in~\cite{campbell_shorter_2017}, this does not change the logic of the proofs, so we could image all channels tensored with the identity and acting on a larger system. Finally, since all three derived conditions hold for all data, $\boldsymbol{x}$, taking expectations over $\boldsymbol{x}$ on both sides and plugging in $\delta_1, \delta_2$ for $\varepsilon_1, \varepsilon_2$ concludes the proof.
\end{proof}

\section{Linear combination of unitaries QNN numerics} \label{app_sec:lcu_QNN_numerics}
In the main text and the previous section, we have advocated the possibility of translating benefits from a linear combination of unitaries QNN (LCU QNN) to a density QNN via the Mixing lemma. However, the rationale for doing so rests on the assumption that an LCU QNN is more expressive than a single-unitary QNN. In other words, that it is possible to find suitable distributional coefficients $\{\alpha_k\}_{k=1}^K$ such that $\sum_k \alpha_k U_k(\boldsymbol{\theta}_k)$ is more expressive, and a better QML model, than any $U_k(\boldsymbol{\theta}_k)$ alone. 

An explicit example showing this is possible in the randomised compiling case was given in~\cite{campbell_shorter_2017}, and in this section we also show that is possible in a learning scenario. 
\begin{figure*}[!ht]
    \includegraphics[width=0.8\linewidth]{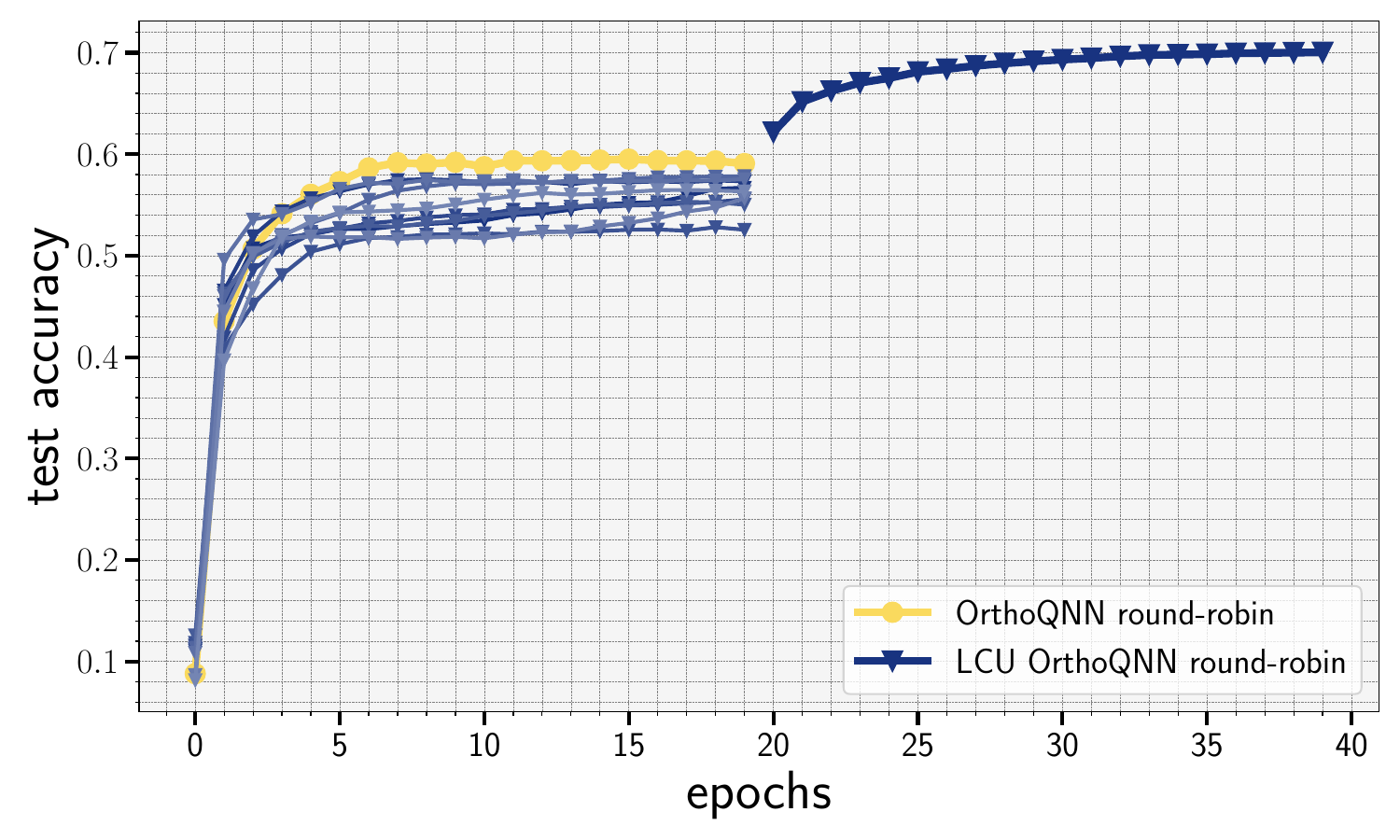}
  \caption{
  \textsf{
  \textbf{LCU QNN versus single unitary QNN.}
  A single circuit round-robin QNN (yellow) versus an initialised LCU QNN (blue). For the LCU QNN, we train $10-1=9$ sub-unitaries for $20$ epochs (to get initial accuracy $\delta_1$ in the Mixing lemma), then initialise the LCU QNN and continue training only the distributional coefficients $\{\alpha_k\}_{k=1}^K$, which are parameterised via a neural network. This gives the improved accuracy $\delta_2$ from the Mixing lemma, which can in priciple be translated to the density QNN.
  }
  }
  \label{fig:lcu_QNN_vs_QNN}
\end{figure*}

In \figref{fig:lcu_QNN_vs_QNN}, we show the results of an LCU QNN outperforming a single unitary QNN for classifying MNIST digits, downscaled to $10$ qubits. We first train all the sub-unitaries in the LCU QNN, then fix them to initialise the full LCU QNN, with a uniform distribution for $\{\alpha_k\}_{k=1}^K$. We then train a neural network to predict the coefficients, which is data dependent, exactly as in~\appref{app_sec:moe_qnns}. We see that the final LCU QNN can outperform any single circuit QNN in terms of MNIST test accuracy. For all sub-circuits, we use the fully parameterised $D=\frac{n(n-1)}{2}$ round-robin orthogonal circuits from~\figref{fig:ortho_qnn_architectures} and perform hyperparameter optimisation using \texttt{optuna}.

Note, in the above, since we perform simulation, we do not take into account the overhead required to postselect the ancilla registers into the $\ket{0}^{\otimes n}$ state and therefore we do not show total measurement shots needed. This is because we are only interested in the final expressivity of models and we leave the drawbacks of the LCU model with implementing the postselection to future work. Nevertheless, this should serve as further motivation to bypass the LCU QNN altogether in favour of the density QNN, which can offer similar performance boosts, but without requiring post-selection.

\section{Hyperparameters} \label{app_subsec:hyperparams}

For hyperparameter optimisation, we use the hyperparameter tuning package $\texttt{optuna}$~\cite{akiba_optuna_2019}. For the experiments in the main text relating to OrthoQNNs and their density counterparts, we use the following choice of hyperparameters:
\begin{itemize}
    \item $\texttt{batch\_size} \in \{32, 48, 64, 80, \dots, 256\}$,  $\texttt{learning\_rate} \in [1\times 10^{-4}, 1\times 10^{-2}]$, 
    \item $\texttt{optimiser} \in \{\texttt{Adam}, \texttt{SGD}, \texttt{RMSProp}\}$, $\texttt{regularisation} \in \{\ell_2, \texttt{None}\}$
\end{itemize}
The (initial) learning rate is sampled uniformly on a logarithmic scale, and the regularisation (either $\ell_2$ norm regularisation or no regularisation) is added to the loss function, which in all cases is the cross entropy loss between model outputs and the true labels. The specific plot (\figref{fig:density_versus_ortholinear_mainfig}b) shows the average over the best $\frac{x}{32}$ hyperparameter runs, those which achieve at least $80\%$ accuracy during training. For each model we have \figref{fig:density_versus_ortholinear_mainfig}ai) OrthoQNN: $\frac{30}{32}$, \figref{fig:density_versus_ortholinear_mainfig}aii) Density OrthoQNN: $\frac{28}{32}$, \figref{fig:density_versus_ortholinear_mainfig}aiii) Density OrthoQNN (compressed): $\frac{30}{32}$. Increasing the criterion to $90\%$ accuracy, the relative success drops to OrthoQNN: $\frac{26}{32}$,  Density OrthoQNN: $\frac{20}{32}$, Density OrthoQNN (compressed): $\frac{22}{32}$ respectively.

\section{Parameter-shift rules for Hamming weight preserving QNNs} \label{app:ortho_qnns}
As mentioned in the main text, orthogonality of weight matrices is a desirable feature in classical machine learning, but is difficult to maintain while training via gradient descent. To combat this, orthogonality preserving methods include: 1) projecting to the Stiefel manifold (the manifold of orthogonal matrices) via, e.g. singular value decompositions (SVDs), 2) performing gradient descent directly in the space of orthogonal matrices, 3) exponentiation and optimisation of an anti-symmetric generator matrix or 4) adding orthogonality regularisation terms to the loss function to be optimised. The first three techniques are theoretically expensive, typically using $\mathcal{O}(n^3)$ complexity to orthogonalise an $n\times n$ weight matrix, while the latter regularisation technique will only enforce approximate orthogonality. 

Orthogonal \emph{quantum} neural networks (OrthoQNNs) were proposed by Ref.~\cite{landman_quantum_2022} as an alleviation to this, and have two possible `modes' of operation for machine learning purposes. The first mode is in a \emph{quantum-inspired} (classical) mode where they can be used as completely classical models for orthogonal neural networks. This is due to the special nature of the gates used within the circuits - specifically all operations within an OrthoQNN are \emph{Hamming weight} (HW) preserving, inheriting the property from the underlying $RBS$ or $FBS$ gates. Applied on a \emph{unary} data encoding, $\ket{\psi}_{\mathsf{unary}} = \sum_{i=1}^n x_i \ket{e_i}, e_j:= 00\cdots 1_j \cdots 00$), an OrthoQNN, $\ket{\phi}_{\mathsf{unary}} = U_{\mathsf{pyr}}(\boldsymbol{\theta})\ket{\psi}_{\mathsf{unary}}$, will preserve the unary nature of the input so $\ket{\phi}_{\mathsf{unary}}$ will also be exclusively supported on the $n$ unary basis elements. This restriction to an $n$ dimensional subspace enables efficient classically simulability, depending on the input state. Therefore, they can also be trained in a purely classical mode without an exponential overhead. A method for performing \emph{layerwise} (on those subsets of gates which can be applied in a single timestep, or \emph{moment}, in parallel - see~\figref{fig:density_versus_ortholinear_mainfig}ai) for an example of such a decomposition) training classically was proposed also in~\cite{landman_quantum_2022} which enables the incorporation of such layers in backpropagation pipelines, with an overhead scaling with the number of layers ($\mathcal{O}(n)$ for the pyramid/round-robin circuits or $\mathcal{O}(\log(n))$ for the butterfly circuit).

The second mode is the fully `quantum' mode - where the orthogonal layers are evaluated and trained on quantum hardware. Here, automatic differentiation through layers is not possible as in the classical scenario, and one must resort to the $\mathcal{O}(n^2)$ parameter-shift rule, as discussed extensively in the main text. Is it possible to have a method scaling as the classical $\mathcal{O}(n)$ layerwise training for such models (without resorting to decomposing the circuits into a density model)? If we could apply the commuting-block argument of Ref.~\cite{bowles_backpropagation_2023}, this would be achievable.

However, this is not straightforward - the gates \emph{within} a layer obviously commute (since they act on different qubits), but $RBS$ gates \emph{between} layers do not obey the required \emph{fixed} commutation relation. This is because the generators of two $RBS$ which share a single qubit neither completely commute nor anticommute, which can be seen as follows with two $RBS$ gates, acting on qubits $0, 1$ ($G_a$) and $1, 2$ ($G_b$):
\begin{align*}
    G_a &:= Y_0\otimes X_1 \otimes \mathds{1}_2 -X_0 \otimes Y_1 \otimes \mathds{1}_2, \qquad  G_b := \mathds{1}_0 \otimes Y_1\otimes X_2 -\mathds{1}_0 \otimes X_1 \otimes Y_2 \\
    G_a\times G_b &= Y_0\otimes X_1 Y_1\otimes X_2  - X_0 \otimes Y_1Y_1 \otimes X_2  
     - Y_0\otimes X_1 X_1\otimes Y_2 + X_0 \otimes Y_1 X_1 \otimes Y_2 \\
     &= G_{\text{comm}} + G_{\text{anti-comm}}
\end{align*}
So we have a commuting part ($X_1X_1$ or $Y_1Y_1$) and an anti-commuting part ($X_1Y_1$ or $Y_1X_1$) on the shared qubit. As a result, $[G_a, G_b] = 2G_{\text{anti-comm}} \neq 0$ and similarly $\{G_a, G_b\}$ does not vanish. We leave a search for an efficient $U(1)$-equivariant-circuit-specific training protocol to future research, and in the next sections focus on the parameter-shift rule.

\subsection{Parameter-shift rule for orthogonal quantum neural networks} \label{app_ssec:ortho_qnns_parameter_shift}
Here, we show the circuits required for computing gradients of the Hamming weight preserving ans\"atze in~\figref{fig:ortho_qnn_architectures} in the main text. Since all gates in these circuits are simply $RBS$ gates, we need only to compute the parameter-shift for the $RBS$ gate. 

This result has been derived before~\cite{anselmetti_local_2021} in the context of fermionic quantum simulation, but we include it here for completeness, and to explicitly keep track of constant factors due to our slightly different parameterisation of the $RBS$ gate. Specifically the general gate in Ref.~\cite{anselmetti_local_2021} is defined as $\exp\left(-i\theta/2 Q\right)$ whereas we define the $RBS$ gate in the form $\exp\left(-i\theta Q\right)$. This factor does not matter in practice for training such models, as the optimiser can adapt the parameters accordingly, but it is important for debugging purposes to generate the correct formulae for the appropriate gates. We also assume the output of the model is the expectation values of Hermitian observables, $\mathcal{B}_0$ relative to a single pure state created by a sequence of unitaries, We can also assume a data encoding with unitary $V(\boldsymbol{x}), \ket{\boldsymbol{x}} :=V(\boldsymbol{x})\ket{0}^{\otimes n}$:
\begin{align} \label{eqn:function_eval_param_shift}
    f(\boldsymbol{\theta}) &= \bra{\psi(\theta, \boldsymbol{x})}\mathcal{B}_0\ket{\psi(\theta, \boldsymbol{x})} = \bra{\boldsymbol{x}} U^{\theta_{1}\dagger}\cdots U^{\theta_{j}\dagger}\cdots U^{\theta_{J-1}\dagger}U^{\theta_{J}\dagger} \mathcal{B}_0 U^{\theta_{J}}U^{\theta_{J-1}}\cdots U^{\theta_{i}}\cdots U^{\theta_{1}}\ket{\boldsymbol{x}} \\
    &= \bra{\boldsymbol{x}} \mathcal{U}^\dagger_{[1:j]}(\theta) \mathcal{B}_{[j+1:J]}\mathcal{U}_{[1:j]}(\theta)\ket{\boldsymbol{x}} \nonumber \\
    \mathcal{B}_{[j+1:J]}  &:=  U^{\theta_{j+1}\dagger} \cdots U^{\theta_{J-1}\dagger}U^{\theta_{J}\dagger}  \mathcal{B}_0 U^{\theta_{J}}U^{\theta_{J-1}}\cdots U^{\theta_{j+1}} \nonumber
\end{align}

Each $RBS$ gate is of the form $RBS(\theta) = U(\theta) = \exp\left(-i\theta Q\right)$ with $Q := 1/2\left(Y\otimes X -X \otimes Y\right)$. This generator has the property that $Q^3 = Q$ and has eigenvalues $\{0, \pm 1\}$. Therefore, according to~\cite{schuld_evaluating_2019, anselmetti_local_2021} we can write $U(\theta)$ as:
\begin{align*}
    U(\theta) = \mathds{1} + \left(\cos(\theta) - 1\right)Q^2 - i\sin\left(\theta\right)Q
\end{align*}
Following the logic of Ref.~\cite{anselmetti_local_2021}, the expression for the gradient with respect to $\theta$ will involve the commutator, $[\mathcal{B}, Q]$ (suppressing the indices on $\mathcal{B})$:
\begin{equation}
    \frac{\partial f(\boldsymbol{\theta})}{\partial \theta} = \bra{\boldsymbol{x}}\mathcal{U}^\dagger_{[1:j]}(\theta) \left(-i [\mathcal{B}_{[j+1:J]}, Q]\right)\mathcal{U}_{[1:j]}(\theta)\ket{\boldsymbol{x}}
\end{equation}
Defining, $U(\theta)(\mathcal{B}) := U^{\dagger}(\theta)\mathcal{B}U(\theta)$ so:
\begin{align*}
   U(\pm\theta)(\mathcal{B})   &= \left[\mathcal{B} + \left(\cos(\theta) - 1\right)Q^2\mathcal{B} \pm i\sin\left(\theta\right)Q\mathcal{B}\right]\left[\mathds{1} + \left(\cos(\theta) - 1\right)Q^2 \mp i\sin\left(\theta\right)Q\right]\\
   & := \left[\mathcal{B} + \delta Q^2\mathcal{B} \pm \gamma QB\right]\left[\mathds{1} + \delta Q^2 \mp \gamma Q\right]\\
   & = \mathcal{B} + \delta \mathcal{B}Q^2 \mp \gamma \mathcal{B}Q +\left[\delta Q^2\mathcal{B} + \delta^2 Q^2\mathcal{B} Q^2 \mp \gamma \delta Q^2\mathcal{B} Q\right] + \left[\pm \gamma Q\mathcal{B} \pm \delta \gamma Q\mathcal{B}Q^2 -\gamma^2 Q\mathcal{B}Q\right]\\
\end{align*}
It turns out that we can extract the commutator above by taking linear combinations of the following $U(\theta)(\mathcal{B}) - U(-\theta)(\mathcal{B})$ for different values of the angles. Computing this term gives:
\begin{align*}
U(\theta)(\mathcal{B}) - U(-\theta)(\mathcal{B}) & =   -2\gamma \mathcal{B}Q  -2\gamma \delta Q^2B Q + 2\gamma QB + 2\delta \gamma QBQ^2 \\
& = -2\gamma [\mathcal{B}, Q] - 2\delta\gamma [Q, Q\mathcal{B}Q]\\
& = -2i\sin\left(\theta\right) [\mathcal{B}, Q] -2i\sin\left(\theta\right) \left(\cos\left(\theta\right) -1\right) [Q, Q\mathcal{B}Q]
\end{align*}

Evaluating the commutator by taking a linear combination of the above expression with two different angles, $\pm\alpha, \pm\beta$ gives:
\begin{align*}
-i[\mathcal{B}, Q] =& d_1\left[U(\alpha)(\mathcal{B}) - U(-\alpha)(\mathcal{B})\right] - d_2\left[U(\beta)(\mathcal{B}) - U(-\beta)(\mathcal{B})\right]\\
=& d_1(-2i\sin\left(\alpha\right) [\mathcal{B}, Q] -2i\sin\left(\alpha\right) \left(\cos\left(\alpha\right) -1\right) [Q, Q\mathcal{B}Q])  \\
&+ d_2(2i\sin\left(\beta\right) [\mathcal{B}, Q] +2i\sin\left(\beta\right) \left(\cos\left(\beta\right) -1\right) [Q, Q\mathcal{B}Q])\\
=& -i[\mathcal{B}, Q]\left(2d_1\sin\left(\alpha\right) - 2d_2\sin\left(\beta\right)\right) 
-2i[Q, Q\mathcal{B}Q] \left[d_1\sin\left(\alpha\right) \left(\cos\left(\alpha\right) -1\right) - d_2\sin\left(\beta\right) \left(\cos\left(\beta\right) -1\right) \right]
\end{align*}
The coefficient of $[\mathcal{B}, Q]$ should $=1$, while the coeffecient of $[Q, Q\mathcal{B}Q]$ should $=0$. Therefore, we get the conditions:
\begin{align*}
   2d_1\sin\left(\alpha\right) - 2d_2\sin\left(\beta\right) = 1 &\implies  d_1\sin\left(\alpha\right) - d_2\sin\left(\beta\right) = \frac{1}{2} \\
    d_1\sin\left(\alpha\right) \left(\cos\left(\alpha\right) -1\right) - d_2\sin\left(\beta\right) \left(\cos\left(\beta\right) -1\right) = 0
    & \implies d_1\sin\left(2\alpha\right)  - d_2\sin\left(2\beta\right) = 1
\end{align*}
To solve these, we can take $d_1 = 1, d_2 = \frac{(\sqrt{2}-1)}{2},  \alpha = \frac{\pi}{4}$ and  $\beta = \frac{\pi}{2}$   
Then we arrive at a four term gradient rule for $RBS$ gates:
\begin{equation}\label{eqn:rbs_gradient_functional_parameter_shift}
    \frac{\partial f(\boldsymbol{\theta})}{\partial \theta_i} = \left[ f\left(\theta_i+\frac{\pi}{4}\right) - f\left(\theta_i- \frac{\pi}{4}\right)\right] - \frac{\sqrt{2}-1}{2}\left[f\left(\theta_i+ \frac{\pi}{2} \right) - f\left(\theta_i- \frac{\pi}{2}\right)\right]
\end{equation}
\subsection{Parameter-shift rule for \texorpdfstring{$FBS$}{} gates} \label{app_ssec:ortho_qnns_parameter_shift_fbs}
A useful generalisation of the $RBS$ gates defined in the main text is to so-called \emph{Fermionic} beam splitter~\cite{kerenidis_quantum_2022} (FBS) gates, which are defined as follows:
\begin{equation} \label{eqn:fbs_gate_definition}
    FBS(\theta)_{ij}\ket{\mathbf{s}} =  \\ 
    \left(
    \begin{array}{cccc}
    1 & 0 & 0 & 0 \\
    0 & \cos(\theta) & (-1)^{f_{i,j, \mathbf{s}}}\sin(\theta) & 0 \\
    0 & (-1)^{f_{i,j, \mathbf{s}}+1}\sin(\theta) & \cos(\theta) & 0 \\
    0 & 0 & 0 & 1
    \end{array}
    \right)
\end{equation}
The $FBS$ gate acts on two qubits, $i,j$, and is defined along with the overall $n$ computational basis state  $\mathbf{s} \in \{0, 1\}^n$ it acts on. Here, $f_{i,j, \mathbf{s}} := f(i, j, \mathbf{s}) := \sum_{i<k<j}s_k$. If the parity, $\bigoplus_{i<k<j} s_k$ between qubits $i, j$ is odd, we have $FBS(\theta) = RBS(\theta)$ gate, and equal to $RBS(-\theta)$ otherwise. $FBS$ gates, in contrast to the strictly two-local $RBS$ gates may generally be maximally non-local acting on all qubits at once, due to need to compute the parity term $f_{i,j, \mathbf{s}}$ quantumly. The $FBS$ gates in particular are useful in creating \emph{subspace} states~\cite{kerenidis_quantum_2022}, a useful primitive that can, for example, accelerate \emph{determinant} sampling machine learning methods~\cite{kazdaghli_improved_2023, thakkar_improved_2024}. Finally, as mentioned in the main text, applying such $FBS$ gates on higher order Hamming weight initial states, or superpositions of different Hamming weight states results in \emph{compound} matrices of order $k$ acting on the $\binom{n}{k}$ dimensional Hamming weight $k$ subspace.

How does one then train these $FBS$ gates with a version of the parameter-shift rule? This is perhaps not obvious at a first glance as from~\eqref{eqn:fbs_gate_definition} as it is not trivial to write as a single operation of the form $e^{i\theta G}$ for some Hermitian generator $G$, and if we did so, $G$ would be a multi-qubit operation. 

However, using the correspondence between the $FBS$ and the $RBS$ gate from Ref.~\cite{kerenidis_quantum_2022}, it becomes clear that one can simply derive the parameter-shift rule for an $FBS$ gate from that of an $RBS$ gate, and it turns out to have the same functional form as~\eqref{eqn:rbs_gradient_functional_parameter_shift}. This fact means that if, on quantum hardware, one has a native way to implement $FBS$ gates (which are generally non-local), evaluating the gradients of these gates does not require any extra circuit resources over simply evaluating the function itself. One may simply compute gradients using the same $FBS$ gates but with shifted parameters, even though the $FBS$ gate itself is not in an obvious form for the parameter-shift requirements. If the $FBS$ gate instead decomposed into a primitive with perhaps \emph{multiple} $RBS$ gates (but each with some fraction of the total angle), this fact may not materialise.

The correspondence is the following (Proposition 2.6 in Ref.~\cite{kerenidis_quantum_2022}):
\begin{equation} \label{eqn:fbs_rbs_relation}
    FBS(\theta)_{ij} = \mathcal{P}(i+1, j)CZ_{i+1, j} RBS(\theta)_{ij} CZ_{i+1, j}\mathcal{P}^{\dagger}(i+1, j)
\end{equation}
where $\mathcal{P}(i+1, j)$ is a circuit which computes the parity of the qubits $k$ qubits between $i$ and $j$ into qubit $i+1$, such that $|i-j| = k+1$. 

Then let $f_{FBS}(\boldsymbol{\theta})$ be the function implemented by a quantum circuit with trainable parameters in $FBS$ gates. We have 
\begin{align*}
    \frac{\partial f_{FBS}(\boldsymbol{\theta})}{\partial \theta_i} &= \left[f'_{RBS}\left(\theta_i+\frac{\pi}{4}\right) - f'_{RBS}\left(\theta_i- \frac{\pi}{4}\right)\right] - \frac{\sqrt{2}-1}{2}\left[f'_{RBS}\left(\theta_i+ \frac{\pi}{2} \right) - f'_{RBS}\left(\theta_i- \frac{\pi}{2}\right)\right]\\
    &= \left[f_{FBS}\left(\theta_i+\frac{\pi}{4}\right) - f_{FBS}\left(\theta_i- \frac{\pi}{4}\right)\right] - \frac{\sqrt{2}-1}{2}\left[f_{FBS}\left(\theta_i+ \frac{\pi}{2} \right) - f_{FBS}\left(\theta_i- \frac{\pi}{2}\right)\right]
\end{align*}
where $f'_{RBS}$ is the formally-equivalent circuit to $f_{FBS}(\boldsymbol{\theta}$, but replacing all $FBS(\theta)_{ij}$ with the expressions using $RBS$ gates~\eqref{eqn:fbs_rbs_relation}. In other words, we take the $FBS$ circuit, write all $FBS$ in terms of corresponding $RBS$ gates and gates with no parameters, evaluate the gradients with respect to the $RBS$ parameters - which end up as circuits with identical form, but with simply shifted parameters - then rewrite the $RBS$ circuits back in the form of $FBS$ gates (shifted by the same amount), again using the relation above.

\section{Measurement protocol \& gradients for orthogonal \texorpdfstring{QNNs}{}} \label{sssec:ortho_qnn_measurement}
In the main text, it was stated that the gradients of all orthogonal-inspired density QNNs could be evaluated more quickly than their OrthoQNN counterparts. However, to adapt Theorem 1 of Ref.~\cite{bowles_backpropagation_2023} fully, 
we require the generators commute \emph{and} the measurement operator obeys the commutation relation specified therein. We will see how this analysis also raises a subtlety not addressed in previous works.

Assume we have a quantum neural network ansatz generated with $N$ unitaries with generators, $G$ as follows:
\begin{equation} \label{eqn:qnn_generic}
    \ket{\psi(\boldsymbol{\theta}, \boldsymbol{x})} = \mathcal{U}(\boldsymbol{\theta})V(\boldsymbol{x})\ket{0}^{\otimes n} , \qquad
    \mathcal{U}(\boldsymbol{\theta}) = \prod_{j=1}^N U_j(\theta_j) = \prod_{j=1}^N e^{i\theta_j G_j}
\end{equation}
Next, recall from Ref.~\cite{bowles_backpropagation_2023} that circuits for which all the generators in~\eqref{eqn:qnn_generic} commute (i.e. commuting-generator circuits) with each other, $[G_i, G_j] = 0~\forall i, j$ and given a measurement observable, $\mathcal{H}$, each generator (assuming they mutually commute) defines a gradient observable, $\mathcal{O}_k := [G_k, \mathcal{H}]$:
\begin{equation} \label{eqn:gradient_equation}
    \frac{\partial \mathcal{L}}{\partial \theta_k} 
    = i \bra{\psi(\boldsymbol{\theta}, \boldsymbol{x})}[G_k, \mathcal{H}]\ket{\psi(\boldsymbol{\theta}, \boldsymbol{x})}
\end{equation}
In the following, we show that the appearance of the commutator in~\eqref{eqn:gradient_equation} has implications for creating commuting-generator circuits from orthogonal quantum neural networks.

\subsection{Orthogonal QNN measurement protocol} \label{ssssec:measurement_protocol_orthoqnn}

Recall from the main text, the output of an orthogonal quantum neural network, $\boldsymbol{y}$, is also a unary encoding, to fully extract the state we only need $n$ amplitudes - those corresponding to the unary bitstrings, $\mathbf{e}_j := 0\cdots 1_j \cdots 0$. This can be done using an $\ell_{\infty}$-norm tomography procedure~\cite{kerenidis_quantum_2020, landman_quantum_2022} as follows. Firstly, the probabilities of the unary states are extracted, $p(\mathbf{e}_j) := y_j^2$ via direct measurement of the circuit. Then two auxiliary circuits are evaluated to extract the signs of the amplitudes. These first/second appends a layer of $RBS(\frac{\pi}{4})$ gates onto the odd/even-controlled qubits after the OrthoQNN, (see Figure 18 in \cite{landman_quantum_2022} for details). The three circuits are measured in the Pauli $Z$ basis to extract the final amplitudes $y_j$. While stated in previous works that it is sufficient to measure all simultaneously in the $Z$ basis to extract the unary amplitudes, we show in the following section this is not true if one is \emph{also} interested in hardware trainability. We must be more careful in the measurement protocol.

Focusing on probability extraction, $p(\mathbf{e}_j) := y_j^2$ (the sign evaluation follows similar logic) we have three choices, which impact practicality and trainability. We can $1)$ measure $\mathds{1}^{\otimes (j-1)}  \otimes Z_j \otimes \mathds{1}^{\otimes (n-j)}$ on each qubit individually, $2)$ perform a global measurement, $Z^{\otimes (j-1)}  \otimes Z_j \otimes Z^{\otimes (n-j)}$ over all qubits at once or $3)$ measure non-overlapping qubit pairs $\mathds{1}_1 \otimes Z_2 \otimes \mathds{1}_3 \otimes Z_4 \dots \mathds{1}_{n-1} \otimes Z_n$ (or equivalently swapping $\mathds{1} \leftrightarrow Z$ for each pair). In terms of forward passes, option $(1)$ adds an $\mathcal{O}(n)$ complexity, while options $2/3)$ add an $\mathcal{O}(1)$ complexity (option $(2)$ requires only a single circuit to run, while option $(3)$ requires $2$ circuits, one to characterise $\mathds{1}_{j} \otimes Z_{j+1}$ and the other for $Z_{j} \otimes \mathds{1}_{j+1}$ for each pair, $\{j, j+1\}$. In all cases, the probabilities are estimated by counting the number of `$1$'s on each qubit, $j$. However, while each technically resolves the same information (the probabilities of the unary states), there is a fundamental difference with respect to trainability. 

 \subsection{Orthogonal QNN gradient extraction} \label{ssssec:density_orthoqnn_gradient_ext}
 
The commuting nature of the generators, $G_j \propto Y_{j_1}\otimes X_{j_2} -X_{j_1} \otimes Y_{j_2}$ enables us to apply the results of Ref.~\cite{bowles_backpropagation_2023} if the measurement Hamiltonian, $\mathcal{H}$ either commutes or anticommutes with $G_j$. If $\mathcal{H}$ commutes, $[\mathcal{H}, G_j] = 0$ and therefore the gradient expression in \eqref{eqn:gradient_equation} is zero. This will be the case if we measure the global version of the observable, $Z^{\otimes n}$ since $[Z_j^{\otimes 2}, Y_{j_1}\otimes X_{j_1} -X_{j_2} \otimes Y_{j_2}] = 0$. However, if we measure each qubit individually, we will have either $Z_{j_1} \otimes \mathds{1}_{j_2}$ or $\mathds{1}_{j_1} \otimes Z_{j_2}$, both of which anticommute with $G_j$ and we will have a non-zero gradient\footnote{In fact this is also true for the final layer of gates in a vanilla OrthoQNN circuit.}.

In the latter case, we have two observables, $\mathcal{O}_{j_1} := 2i(\mathds{1}_{j_1} \otimes Z_{j_2})\frac{1}{2}(Y_{j_1}\otimes X_{j_2} X_{j_1} \otimes Y_{j_2}) = Y_{j_1}\otimes Y_{j_2} + X_{j_1} \otimes X_{j_2}$ and $\mathcal{O}_{j_2} := 2i(Z_{j_1} \otimes \mathds{1})\frac{1}{2}(Y_{j_1}\otimes X_{j_2} -X_{j_1} \otimes Y_{j_2}) = -X_{j_1}\otimes X_{j_2} - Y_{j_1} \otimes Y_{j_2} = -\mathcal{O}_{j_1}$. So, by evaluating the gradient with respect to the $j^{th}$ element $y_{j}$ of the output vector $\boldsymbol{y}$, the gradient with respect to $y_{j+1}$ simply points in the opposite direction ($\partial y_j / \partial \theta_{j} = -\partial y_{j+1} / \partial \theta_{j}$). As a result, we only need to measure gradient observables for one of the qubits upon which the $RBS$ gate is supported. Therefore, an optimal choice is to use measurement scheme $(3)$, measuring  $Z$  on only a single qubit out of each pair. A forward pass then requires $2$ extra circuits (over the $Z^{\otimes n}$ measurement protocol), but a full gradient evaluation requires only a single extra circuit (since we only need to extract gradients for $\frac{n}{2}$ qubits, one from each pair).


It is also in fact sufficient to compute gradients for only a \emph{single} layer (specifically the widest) of the odd-even (uncompressed) decomposition. This is due to 1) the commuting nature of the generators, and 2) the fact the generators are all the same for each gate. The sub-circuit $U_1$ in~\figref{fig:density_versus_ortholinear_mainfig}aii) contains $n$ $RBS$ gates on the first two qubits in series. Each of these are generated by the same operator $G^1_{12} = G^2_{12} = \cdots = G^n_{12}  := G_{12}$. Since the generators commute, each gradient observable is of the form: $i \bra{\boldsymbol{x}}\mathcal{U}_1(\boldsymbol{\theta})[G_{12}, \mathcal{H}]\mathcal{U}_1^{\dagger}(\boldsymbol{\theta})\ket{\boldsymbol{x}}$ (\eqref{eqn:gradient_equation}), and are therefore identical. The same applies for the other generators $G_{j, j+1}$ in $U_1$ and $U_2$ and so the density QNN in~\figref{fig:density_versus_ortholinear_mainfig}aii) has gradients which all be computed using two different circuits (and classical post-processing), assuming no parameter sharing. 

\subsection{Orthogonal QNN diagonalisation} \label{ssssec:density_orthoqnn_diagon}

Finally, we need to compute the diagonalization unitary needed to extract parallel gradient information. Fortunately, for orthogonal inspired density QNNs, this is simple. As mentioned above, we measure gradient observables,  $\mathcal{O} :=  X \otimes X + Y \otimes Y$ where the forward measurement is (for example) $\mathds{1}\otimes Z$. Diagonalization then results in a matrix $P$ such that $P\mathcal{O}P^{-1} = D$ for a diagonal matrix, $D$:
\begin{equation*}
\mathcal{O} = {\left(\begin{array}{cccc}
    0 & 0 & 0 & 0 \\
    0 & 0 & 2 & 0 \\
    0 & 2 & 0 & 0 \\
    0 & 0 & 0 & 0 
    \end{array} \right)
    }
    \implies P = {\left(\begin{array}{cccc}
            1 & 0 & 0 & 0 \\
            0 & \frac{1}{\sqrt{2}} & \frac{1}{\sqrt{2}} & 0 \\
            0 & -\frac{1}{\sqrt{2}} & \frac{1}{\sqrt{2}} & 0 \\
            0 & 0 & 0 & 1 
            \end{array}\right)}
\end{equation*}

which is simply an $RBS(\frac{\pi}{4})$ gate appended to each $RBS$ gate in each orthogonal sub-unitary. An efficient implementation could use the fact that these diagonalization circuits are the same as those required for $\ell_{\infty}$ tomography above.

\section{Features of density quantum neural networks} \label{app_sec:misceallanea}
In this section, we expound some comments regarding features of density QNNs, and their extensions.

\subsection{Incorporating data} \label{app_sec:data_density_qnn}

First, the loss function in~\eqref{eqn:density_qnn_loss_fn} assumes only a single datapoint, $\ket{\boldsymbol{x}}$ is evaluated. We can also take expectations of this loss with respect to the training data: 
\begin{align} \label{eqn:density_qnn_loss_fn_over_data}
    \mathcal{L}(\boldsymbol{\theta}, \boldsymbol{\alpha}) = \underset{\boldsymbol{x}}{\mathbb{E}}\Big[\Tr\Big(\mathcal{H}\rho(\boldsymbol{\theta}, \boldsymbol{\alpha}, \boldsymbol{x})\Big)\Big] 
    &\approx \frac{1}{S}\sum_{s=1}^S \delta_i \Tr\Big(\mathcal{H}\rho(\boldsymbol{\theta}, \boldsymbol{\alpha}, \boldsymbol{x}_i)\Big)\\
    &= \frac{1}{S}\sum_{k=1}^K\sum_{s=1}^S \alpha_k \delta_i \Tr\Big(\mathcal{H} U_k(\boldsymbol{\theta}_k)\ketbra{\boldsymbol{x}_s}{\boldsymbol{x}_s}U^\dagger_k(\boldsymbol{\theta}_k)\Big)
\end{align}
for $S$ data samples. Therefore, we increase the number of circuits we must run by a factor of $S$, each of which will have a gradient cost of 
$\sum_{l=1}^K\sum_{k=1}^K T_{\ell k}$. One can view~\eqref{eqn:density_qnn_loss_fn_over_data} as creating an `average' data state $\sum_{s=1}^S  \delta_s \ketbra{\boldsymbol{x}_s}{\boldsymbol{x}_s}$, where $\{\delta_i\}$ is the empirical distribution over the data and then applying the density sub-unitaries with their corresponding distribution, $\{\alpha_k\}_k$, or by estimating the elements of the stochastic $K\times S$ matrix with elements $\Tr\Big(\mathcal{H} U_k(\boldsymbol{\theta}_k)\ketbra{\boldsymbol{x}_s}{\boldsymbol{x}_s}U^\dagger_k(\boldsymbol{\theta}_k)\Big)$. 

In reality, we also will estimate the trace term in~\eqref{eqn:density_qnn_loss_fn_over_data} with $M$ measurements shots from the circuit. Incorporating this with the above density model and its gradients, one could also define an extreme gradient descent optimiser in the spirit of~\cite{sweke_stochastic_2020}, sampling over datapoints, measurement shots, measurement observable terms, parameter-shift terms and in our case, sub-unitaries to estimate the loss function and its gradients in a single circuit run. In this extreme, an estimator for the gradients (assuming they are computed via a linear combination of pure state expectation values, i.e. the parameter-shift rule) can be written in the following form:
\begin{align} \label{eqn:density_qnn_loss_fn_over_data_gradient}
    \frac{\partial \mathcal{L}(\boldsymbol{\theta}, \boldsymbol{\alpha})}{\partial \boldsymbol{\theta}_j} = \sum_{p=1}^{P} \gamma_p \mathcal{L}(\boldsymbol{\theta} + \boldsymbol{\beta}^p_j, \boldsymbol{x}) =  \underset{\boldsymbol{x}}{\mathbb{E}}\Big[\sum_{i=1}^{P} \gamma_i \Big[\Tr\Big(\mathcal{H}\rho(\boldsymbol{\theta}  + \boldsymbol{\beta}^p_j, \boldsymbol{\alpha}, \boldsymbol{x})\Big)\Big] \Big]\\
    \approx \sum_{p=1}^{P} \sum_{k=1}^{K}\sum_{s=1}^{S} \gamma_p \alpha_k \delta_s \underbrace{\Tr\Big(\mathcal{H} U_k(\boldsymbol{\theta}_k + \boldsymbol{\beta}^p_j )\ketbra{\boldsymbol{x}_s}{\boldsymbol{x}_s}U^\dagger_k(\boldsymbol{\theta}_k + \boldsymbol{\beta}^p_j)\Big)}_{\text{estimated with } M \text{ measurements}}
\end{align}
where $\boldsymbol{\beta}^p_j := [0, \dots, \underbrace{\beta^p}_{j}, \dots, 0]^{\top} = \beta^p\boldsymbol{e}_j \in \mathbb{R}^N$ is a unit vector in direction $j$. In the parameter-shift rule, the coefficients $\boldsymbol{\gamma} = \{\gamma_p\}_{p=1}^P$ depend on the unitary, $U_j(\theta_j)$. If $G_j$ has two unique eigenvalues, e.g. $G_j \in\{X, Y, Z\}$ we have $P = 2$ and $\boldsymbol{\beta}^{\pm}_{j} = \pm \frac{\pi}{2}, \gamma_{\pm} = \frac{1}{2}$ for every parameter. Hence, for a QNN with $N$ trainable parameters, we assume a forward pass (single loss evaluation) as a constant time operation, $ \mathcal{T}(QNN(\boldsymbol{\theta})) = \mathcal{O}(1)$. However, the gradient requires $\mathcal{O}(N)$ extra `shifted' circuit evaluations so $ \mathcal{T}(QNN'(\boldsymbol{\theta})) = \mathcal{O}(N)$ (ignoring other parameters). This is similar to the forward AD gradient scaling, but far removed from the efficiency of backpropagation. 

\subsection{Model expressivity} \label{app_sec:model_expressivity}

Second, the density model is at least as expressive as its `pure' state counterpart. The standard unitary QNN model for corresponds to $K=1$, where we have $\alpha_1 = 1\ \&\ \alpha_j = 0\forall j\neq 1$ and    
\begin{align*}
\mathcal{L}(\boldsymbol{\theta}, \boldsymbol{x}) &= \Tr\Big(\mathcal{H}\rho(\boldsymbol{\theta}, \boldsymbol{x})\Big) \\
&= \Tr\Big(\mathcal{H}U_1(\boldsymbol{\theta}_1)\ketbra{\boldsymbol{x}}{\boldsymbol{x}}U^\dagger_1(\boldsymbol{\theta}_1)\Big)  
= \Tr\Big(\bra{\boldsymbol{x}}U^\dagger_1(\boldsymbol{\theta}_1)\mathcal{H}U_1(\boldsymbol{\theta}_1)\ket{\boldsymbol{x}}\Big)
= \bra{\psi(\boldsymbol{\theta}, \boldsymbol{x})}\mathcal{H}\ket{\psi(\boldsymbol{\theta}, \boldsymbol{x})}
\end{align*}

Now, if one assumes independent parameters per sub-unitary in~\eqref{eqn:density_qnns}, one might ask - what are the training dynamics of a model whose gradient for subsections of parameters are completely independent? There actually exists an analogous classical situation. Taking a very simple linear layer in a neural network \emph{without the activation}, $f(W, \mathbf{b}) = W\boldsymbol{x} + \mathbf{b}$. The gradients with respect to the parameters, $W, \mathbf{b}$, are independent of each other, $\partial_W f =\boldsymbol{x}, \partial_{\mathbf{b}} f = \mathds{1}$, as in the above. However, this no longer is the case when adding an activation and multiple layers in a deep network: $f(W, \mathbf{b}) = \sigma(W\boldsymbol{x} + \mathbf{b})$: $\partial_W f = \partial_{W\boldsymbol{x} + \mathbf{b}} (\sigma(W\boldsymbol{x} + \mathbf{b})) \times\boldsymbol{x}$ and $\partial_{\mathbf{b}} f = \partial_{W\boldsymbol{x} + \mathbf{b}}(\sigma(W\boldsymbol{x} + \mathbf{b})) \times \mathds{1}$, where (depending on the activation) the non-differentiated parameters still propagate into the gradients of the differentiated ones.

\subsection{Density reuploading quantum neural networks} \label{app_sec:reuploading_density}
From the main text, we can generalise for completeness data reuploading density models as Fourier series. In this case, each sub-unitary reuploading layer, $\ell$, is of the form:
\begin{equation}\label{eqn:single_reup_layer_density}
    \mathcal{E}^{\ell, k}_{\boldsymbol{\theta}^{\ell}, \boldsymbol{x}}(\rho) := U_k(\boldsymbol{\theta}^{\ell}_k)V_{\ell}(\boldsymbol{x}) \rho V_{\ell}^{\dagger}(\boldsymbol{x})U^\dagger_k(\boldsymbol{\theta}^{\ell}_k)
\end{equation}

where $\rho$ is again some data-independent initial state. Given this, we define $L$ layered density reuploading QNNs as follows:

\begin{definition}[Density QNNs with data reuploading] \label{def:density_qnns_reuploading_def}
Given a classical data point (vector or otherwise), $\boldsymbol{x}$, we can define a density QNN incorporating data reuploading as:
\begin{align} \label{eqn:density_qnns_reuploading}
    \rho^R(\boldsymbol{\theta}, \boldsymbol{\alpha}, \boldsymbol{x}) &= \sum_{k=1}^K \alpha_k \mathcal{E}^{L, k}_{\boldsymbol{\theta}^{\ell}, \boldsymbol{x}}(\ketbra{0}{0}^{\otimes n}) \\ 
    \mathcal{E}^{L, k}_{\boldsymbol{\theta}, \boldsymbol{x}}(\rho) := \underbrace{\mathcal{E}^{L, k}_{\boldsymbol{\theta_L}, \boldsymbol{x}}\left(\cdots\mathcal{E}^{2, k}_{\boldsymbol{\theta}_2, \boldsymbol{x}}\left(\mathcal{E}^{1, k}_{\boldsymbol{\theta}_1, \boldsymbol{x}}(\rho)\right)\right)}_{L\ \textsf{times}}, &\qquad
    \mathcal{E}^{r, k}_{\boldsymbol{\theta}, \boldsymbol{x}}(\rho) := U_k(\boldsymbol{\theta}^r_k)V{\ell}(\boldsymbol{x}) \rho V_{\ell}^{\dagger}(\boldsymbol{x})U^{\dagger, r}_k(\boldsymbol{\theta}_k)
\end{align}
\end{definition}
For two reuploads the density state is:
\begin{equation} \label{eqn:two_reupload_density_QNN}
    \rho^2(\boldsymbol{\theta}, \boldsymbol{\alpha}, \boldsymbol{x}) =
    \sum_{k=1}^K \alpha_k U(\boldsymbol{\theta}^2_k)V_{\boldsymbol{x}}U(\boldsymbol{\theta}^1_k)\ketbra{\boldsymbol{x}}{\boldsymbol{x}}U^\dagger(\boldsymbol{\theta}^1_k)V_{\boldsymbol{x}}U^\dagger(\boldsymbol{\theta}^2_k)
\end{equation}
Where we drop the subscript $k$ on the sub-unitaries for compactness, and $V_{\boldsymbol{x}}\ket{0}^{\otimes n} := V({\boldsymbol{x}})\ket{0}^{\otimes n} := \ket{\boldsymbol{x}}$. Now, in the case of sub-unitaries which are decomposed into (commuting-)blocks, somewhat confusingly, $\boldsymbol{\theta}$ is a rank $4$ tensor, $\boldsymbol{\theta} := \{\theta^{r, b}_{k, j}\}$ where $r, b, k, j$ indexes the upload, block, sub-unitary and parameter respectively.
Now unfortunately it is not clear in general how to efficiently train the above models in~\defref{def:density_qnns_reuploading_def} even if each unitary is of a commuting-block form of Ref.~\cite{bowles_backpropagation_2023}. This is because the arbitrary nature of the data encoding unitary, $V(\boldsymbol{x})$ prohibits a fixed commutation relation between unitaries in subsequent uploads. 

In the above, we assume that the unitaries applied are the \emph{same} in successive reuploads - i.e. if the random variable we employ chooses sub-unitary, $U_k$, to apply to the initial state for the first upload, the subsequent sub-unitaries applied are also $U_k$, although with potentially different parameters, $\boldsymbol{\theta}_k^1 \neq \boldsymbol{\theta}_k^2$. In this case the distribution $\{\alpha_k\}$ represents the distribution of reuploading `sequences'. One could also account for different distributions over each reupload. For example, we apply one upload to the state and get $U(\boldsymbol{\theta}^1_k)\ketbra{\boldsymbol{x}}{\boldsymbol{x}}U^\dagger(\boldsymbol{\theta}^1_k)$ with probability $\alpha_k$. If we choose a different (but independent) distribution, $\{\beta_{k'}\}$, to select $U(\boldsymbol{\theta}^2_{k'})$, the resulting state would be $U(\boldsymbol{\theta}^2_{k'})V_{\boldsymbol{x}}U(\boldsymbol{\theta}^1_k)\ketbra{\boldsymbol{x}}{\boldsymbol{x}}U^\dagger(\boldsymbol{\theta}^1_k)V_{\boldsymbol{x}}U^\dagger(\boldsymbol{\theta}^2_{k'})$ with probability $\alpha_k\beta_{k'}$, and so on.

The function produced from each sub-unitary, $k$, can be written as a partial Fourier series: 
\begin{align}\label{eqn:data_reuploading_full_fourier}
    f^k(\boldsymbol{\theta}, \boldsymbol{x}) &= \Tr\left(\mathcal{O}\mathcal{E}^{L, k}_{\boldsymbol{\theta}, \boldsymbol{x}}(\rho)\right) = \sum_{\boldsymbol{\omega}\in \Omega} c^k_{\boldsymbol{\omega}}(\boldsymbol{\theta}) e^{i \boldsymbol{\omega}_k^\top \boldsymbol{x}}
\end{align}
which is a sum of Fourier coefficients with frequencies, $\Omega_k := \{\boldsymbol{\omega}_k\}$. As in the main text we assume a Hamiltonian encoding for each $V_{\ell}(\boldsymbol{x})$.

Now, it is straightforward to see that a density QNN with this form of data reuploading is a model of the form:
\begin{equation}\label{eqn:density_fourier_term_full}
g(\boldsymbol{\theta}, \boldsymbol{\alpha}, \boldsymbol{x}) = \Tr\left(\mathcal{O}\rho^L(\boldsymbol{\theta}, \boldsymbol{\alpha}, \boldsymbol{x})\right) = \sum_{k=1}^K \alpha_k f^k(\boldsymbol{\theta}, \boldsymbol{x}) = \sum_{k=1}^K \alpha_k \sum_{\boldsymbol{\omega}_k\in \Omega_k} c_{\boldsymbol{\omega}_k}(\boldsymbol{\theta}) e^{i \boldsymbol{\omega}_k^\top \boldsymbol{x}}
\end{equation}
So the model is a \emph{randomised} linear combination of partial Fourier series. 

\section{Mixture of Experts} \label{app_sec:moe_qnns}

In the main text, we introduced \emph{data-dependent} coefficients for the sub-unitaries, $\boldsymbol{\alpha} \rightarrow \boldsymbol{\alpha}(\boldsymbol{x})$. We showed that this could dramatically improve performance, even outperforming an orthogonal quantum layer with the same number of parameters in the quantum circuit. As mentioned, we also referred to this data dependence as an interpretation of the density QNN framework in the \emph{mixture of experts} (MoE) formalism~\cite{jacobs_adaptive_1991, jordan_hierarchical_1993}. In this section, we study this MoE interpretation in more detail, and test it on the butterfly decomposition for a density QNN, along with the odd-even decomposition from the main text. 

The MoE framework contains a set of experts, $\{f_1, \dots, f_K\}$, each of which is `responsible' for a different training case. A \emph{gating} network, $w$, decides which expert should be used for a given input. In the simplest form, the MoE output, $F(\boldsymbol{x})$, is a weighted sum of the experts, according to the gating network output $F(\boldsymbol{x}) = \sum_k w_k(\boldsymbol{x})f_k(\boldsymbol{x})$. The specific implementation of these gating an expert networks has been the study of much classical research. For example, each expert could be a neural networks~\cite{shazeer_outrageously_2017} with millions of parameters, or shallow models such as support vector machines~\cite{collobert_parallel_2001} or Gaussian processes~\cite{theis_generative_2015, deisenroth_distributed_2015} for example. For example, in introducing a \emph{hierarchical} version of these networks, Ref.~\cite{jordan_hierarchical_1993} used a simple linear layer with trainable layer, $W_g$, for the gating network:
\begin{equation}\label{eqn:moe_gating_simple_softmax}
    w_k(\boldsymbol{x}) = \texttt{softmax}_k(W_g\boldsymbol{x})
\end{equation}
Ref.~\cite{jordan_hierarchical_1993} investigated stacking MoE layers with such gating networks controlling linear experts with an activation function. One could also explore more complex gating functions as Ref.~\cite{shazeer_outrageously_2017}, which used `noisy top-$k$ gating' - adding Gaussian noise to the linear gating followed by passing only the influence of the top $k$ experts to the \texttt{softmax}. Adding sparsity reduces the network computation and the Gaussian noise improves the `\emph{unbalanced expert utilisation}' problem - unbalanced expert loads occur when the gating learns to only rely a small handful of experts, an effect which is compounded via training, since this subset will receive more and more examples as training progresses. All of these extensions could be explored within the quantum scenario in future work.

\subsection{Density quantum neural network as a mixture of experts} \label{app_ssec:density_qnn_moe}
As mentioned in the main text, in the density QNN framework, the distribution of sub-unitaries, $\alpha_j$, act as a weighting over sub-unitaries. If the method of parameterising the distribution is efficient and (efficiently) trainable, the model will select the sub-unitary which is most effective at extracting information from the data. By adding the data-dependence to the $\alpha$ coefficients, these can be interpreted as the output of a gating network choosing the most suitable sub-unitary for a given input. Note, use of classical neural networks to drive the evolution of quantum neural networks is not a new concept in itself~\cite{verdon_learning_2019, wilson_optimizing_2021}, and has even now a relatively long history. However, we believe the interpretation we give here in the context of a MoE is novel.

For example, in quantum data applications, one could imagine classifying directly states corresponding to fixed $k$-body Hamiltonians, as in the quantum phase recognition problem~\cite{cong_quantum_2019}. A sequence of sub-unitaries (experts) could be defined with specific entangling characteristics: $U_1$ contains only $1$-body terms (single qubit rotations), $U_2$ contains $2$-body terms (two qubit gates), $U_3$ contains $3$-body terms and so on. We create a density QNN with probabilities $\{\alpha_1, \alpha_2, \alpha_3, \dots \}$. It is clear that if the state to be classified is a product state, then the model is sufficient to learn the weighting $\alpha_1 = 1, \alpha_j = 0~\forall j\neq 1$, which conversely will not be sufficient for more strongly entangled inputs.

\begin{figure*}[!ht]
    \includegraphics[width=\linewidth]{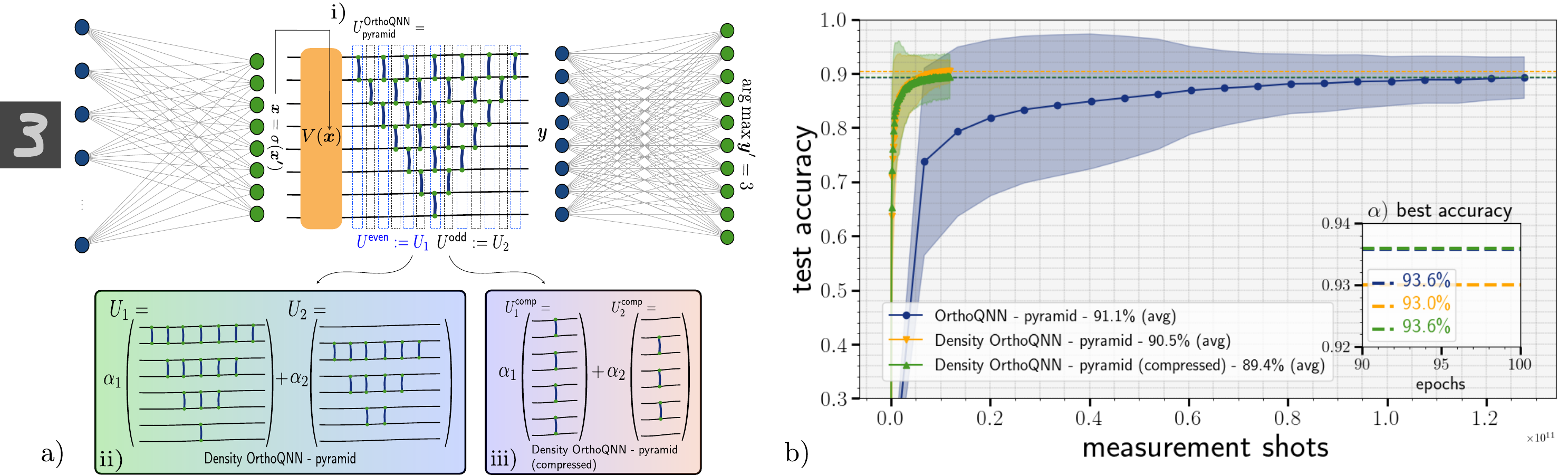}
  \caption{
  \textsf{
  \textbf{Pure state Hamming weight preserving (orthogonal) QNN with a pyramid ansatz, and the density QNNs derived from it.} a) The full models used for numerical results. MNIST data is flattened with a $784 \times 8$ linear layer outputting, $\boldsymbol{x}'$. A ReLU activation gives $\boldsymbol{x} := \sigma(\boldsymbol{x}') := \mathsf{ReLU}(\boldsymbol{x}')$. A data loader, $V(\boldsymbol{x})$, is used to encode $\boldsymbol{x}$ in a unary amplitude encoding (suitably normalised), $\ket{\boldsymbol{x}} \propto \sum_j x_j\ket{\mathbf{e}_j}, \mathbf{e}_j := 0\cdots 1_j \cdots 0$. This is then processed by either an ortholinear QNN $U^{\textsf{OrthoQNN}}$ with parameters $\boldsymbol{\theta}$, or a density version as in aii) \& aiii). An outcome vector, $\boldsymbol{y}$, is extracted, and postprocessed into $10$ label classes (digits $\ell \in \{0, 1, \dots, 9\}$) via a $8 \times 10$ linear layer and a \texttt{softmax} function. b) shows the comparison between three models, `OrthoQNN - pyramid', `Density OrthoQNN  - pyramid' and `Density OrthoQNN - pyramid  (compressed)' which correspond to (a) i), ii) or iii) respectively. The y-axis shows the test accuracy over all the $10,000$ MNIST test images, compared to the number of measurement shots required to train the model, i.e. using the parameter-shift rule for i) and the commuting-generator diagonalising circuit for ii) and iii). The main plot shows the mean (solid line) and standard deviation (shaded region) over the best fraction out of $32$ overall hyperparameter optimisation runs for all three models using $\texttt{optuna}$. We define `best' to be those runs which achieve $>80\%$ test accuracy, and we give more details in~\appref{app_subsec:hyperparams}. Inset ($\alpha$) shows the \emph{best} test accuracies by each model over all hyperparameters in the last $10$ epochs.
  }
  }
  \label{fig:density_versus_ortholinear_mainfig}
\end{figure*}

\subsubsection{Odd-even pyramid decomposition} \label{sssec:odd_even_pyramid}

Before including this data-dependence to uplift the density QNN to a full mixture of experts, let us first describe a simpler decomposition using the example of the pyramid circuit~\figref{fig:ortho_qnn_architectures}a). This example will give the greatest gradient query speedup, and is an alternative decomposition to the `layerwise' approach for the round-robin circuits in the main text. We also illustrate how one may ``dress'' such circuits with feature pre- and post-processing classical layers to generate competitive performance. Given the pyramid circuit~\figref{fig:ortho_qnn_architectures}a), we define the odd-even decomposition with $U_1:= U^{\textsf{even}}$ and $U_2 := U^{\textsf{odd}}$. $U^{\textsf{even}}$ contains the circuit moments where each gate within has an even-numbered qubit as its first qubit (the `control') and $U^{\textsf{odd}}$ contains odd qubit-controlled gates only. The resulting density QNN state is then (initialised with a uniform distribution weighting, $\boldsymbol{\alpha}$ - though we also allow these to be trainable):
\begin{equation} \label{eqn:odd_even_decomp_equation}
    \rho(\boldsymbol{\theta}, \boldsymbol{\alpha}=\left\{\frac{1}{2}, \frac{1}{2}\right\}, \boldsymbol{x}) = 
    \frac{1}{2}\left(U^{\textsf{even}}(\boldsymbol{\theta})\ketbra{\boldsymbol{x}}{\boldsymbol{x}}U^{\textsf{even}}(\boldsymbol{\theta})^{\dagger}\right) 
    + \frac{1}{2}\left(U^{\textsf{odd}}(\boldsymbol{\theta})\ketbra{\boldsymbol{x}}{\boldsymbol{x}}U^{\textsf{odd} }(\boldsymbol{\theta})^{\dagger}\right)
\end{equation}

 All gates in $U_1$ and $U_2$ mutually commute with each other trivially. The input state, $\ket{\boldsymbol{x}} = \sum_j x_j \ket{\mathbf{e}_j}$, is a unary ($\mathbf{e}_j$ is a basis vector with a single $1$ in position $j$ and zeros otherwise) amplitude encoding of the vector $\boldsymbol{x}$. Since the unitaries are Hamming weight preserving, the output states, $\ket{\boldsymbol{y}^{\textsf{odd}}}, \ket{\boldsymbol{y}^{\textsf{even}}}$ from each sub-unitary are of the form $\ket{\boldsymbol{y}} = \sum_j y_j \ket{\mathbf{e}_j}$ for some vector $\boldsymbol{y}$. This output is related to the input vector via some orthogonal matrix transformation $O^U$, $\boldsymbol{y} = O^U\boldsymbol{x}$ where the elements of $O^U$ can be computed via the angles of the reconfigurable beam splitter (RBS) gates in the circuit (see~\secref{sec:methods}). The typical output of such a layer is the vector $\boldsymbol{y}$ itself, for further processing in a deep learning pipeline. For our purposes in gradient-based training, due to the linearity and the purity of the individual output states, $ \ket{\boldsymbol{y}^{\textsf{even}/\textsf{odd}}}$, we can deal with both individually and classically combine the results.
 
\subsection{Logarithmic butterfly decomposition} \label{ssec:attention_orthoqnn_mnist_butterfly}
Similarly to the round-robin circuit from the main text, the `butterfly' circuit (seen in butterfly~\figref{fig:ortho_qnn_architectures}c) on $8$ qubits), can decompose layerwise into $\log_2(n)$ sub-unitaries in the density framework. As shown in~\tabref{tab:summary_density_comparison} this gives a gradient query advantage from $\mathcal{O}(n\log(n))$ circuits to $\mathcal{O}(\log(n))$. This is because we decompose the butterfly into $\log(n)$ layers, each with $n/2$ gates, each of which must be parameter-shifted in the original circuit. In the density framework, each `layer' can be gradient-evaluated independently.
\begin{figure*}[!ht]
    \centering
    \includegraphics[width=\linewidth]{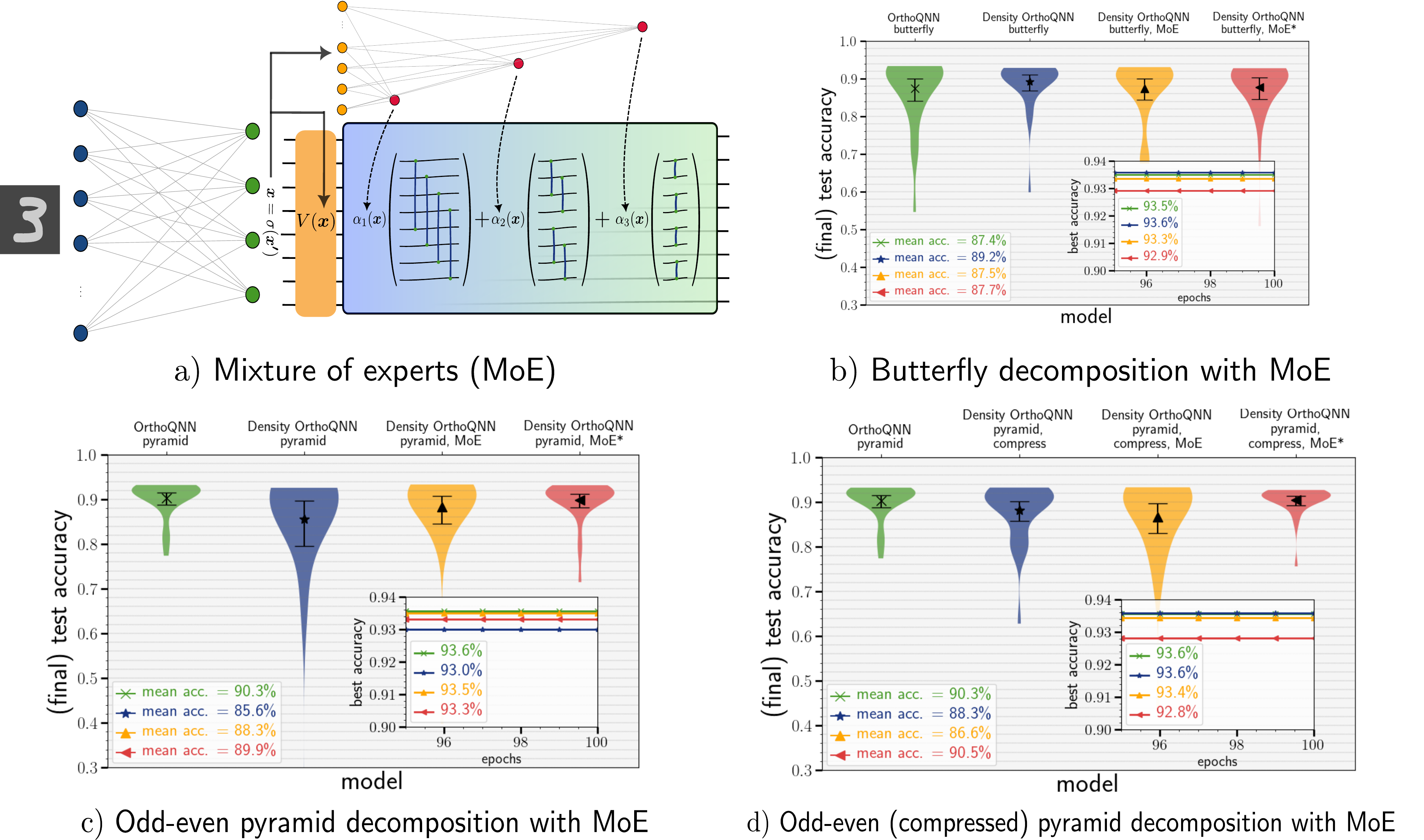}
    \caption{
  \textsf{
  \textbf{Density QNN with data dependent sub-unitary weighting parameters - mixture of experts.} \\
  a) Introducing data-dependence into sub-unitary decomposition via mixture of experts (MoE), for the decomposed butterfly $U(1)$ equivariant unitary (orthogonal QNN or OrthoQNN). A simple linear layer takes as input $\boldsymbol{x} := \sigma(\boldsymbol{x}')$, and outputs $\boldsymbol{\alpha}(\boldsymbol{x})= \{\boldsymbol{\alpha}_k(\boldsymbol{x})\}$, again with $\sum_k\boldsymbol{\alpha}_k(\boldsymbol{x}) = 1~\forall\boldsymbol{x}$. We omit the final (classical) post-processing layer. We test this for the b) \textbf{butterfly decomposition}, c) \textbf{odd-even decomposition}, and d) \textbf{odd-even compressed decomposition}. Again, we perform hyperparameter optimisation over $32$ trails on the MNIST dataset. Errorbars on violin plots show mean and $95\%$ confidence intervals bootstrapped using $1000$ samples. We notice, for those models which are very `shallow', or have few parameters - namely the butterfly and the compressed pyramid extraction, a mixture of expert gating distribution does not seem to help performance. However, if we have a larger number of parameters as in the odd-even extraction of the pyramid circuit (see~\figref{fig:density_qnn_mainfig}a) in the main text), an MoE (data-dependent learnable weighted distribution) \emph{does} appear to help the model learn, achieving both a higher absolute test accuracy over all hyperparameters, and a better average accuracy. 
  }
  }
  \label{fig:attention_density}
\end{figure*}

The specific gating distribution (or the sub-unitary coefficients) $\{\alpha_1(\boldsymbol{x}), \alpha_2(\boldsymbol{x}), \alpha_3(\boldsymbol{x})\}$ we use are as follows: 
\begin{align} \label{eqn:moe_gating_quantum_numerics}
    \textnormal{``MoE''} &\implies \alpha_k(\boldsymbol{x}) = \texttt{softmax}_k(\texttt{Linear}(\sigma(\boldsymbol{x}')),  \\  \textnormal{``MoE*''} &\implies \alpha_k(\boldsymbol{x}) = \texttt{softmax}_k(\mathsf{GELU}(\texttt{Linear}(\sigma(\boldsymbol{x}')))
\end{align}
We use the ``MoE'' version for the round-robin decomposition in the main text.

We again test using MNIST data for the butterfly (\figref{fig:attention_density}b),  odd-even (\figref{fig:attention_density}c) and compressed odd-even (\figref{fig:attention_density}d) respectively. In these figures, we plot the average accuracies achieved at the end of $100$ training epochs, over hyperparameter runs. The errorbars in the plots represented $95\%$ confidence intervals which are bootstrapped using $1000$ samples. The insets show the best accuracies achieved over all runs for each model.

We make some observations for these results. First, in some cases, the vanilla density model can outperform it's pure version, e.g. comparing best accuracy achieved by the density QNN-butterfly versus the OrthoQNN-butterfly in~\figref{fig:attention_density}b (inset). Secondly, for models with larger parameter counts per sub-unitary, e.g. the uncompressed odd-even decomposition, the data-dependent attention \emph{does} appear to help (with/without activation) - and can boost performance to the level of the original model. Thirdly, an explicit activation function does not have a conclusive impact in performance - it can make the model perform better on \emph{average} over hyperparameter optimisation as in Figs.~\ref{fig:attention_density}c, \ref{fig:attention_density}d, but ultimately the \emph{best} accuracies are found either by the original model or the density model without activation (insets). The final observation is that making $\boldsymbol{\alpha}(\boldsymbol{x})$ data dependent (including a MoE), and predictable by a neural network, reduces the variance of the training over hyperparameter runs, which can be observed in all scenarios.

\section{Dropout in quantum machine learning} \label{app:dropout_qml}
Dropout is a technique in classical machine learning~\cite{srivastava_dropout_2014} to effectively and efficiently combine the predictions of an exponentially large number of networks, avoiding the overhead of needing to train many networks individually and combine their results \textit{ex post facto}. It also regularises the output model and prevents overfitting by avoiding the network learning very complex and specific relationships between all neurons. With dropout, each neuron learns to solve the problem with only a random, small, collection of partner neurons at any time - leading to information being shared across the entire network.

At each training forward pass, dropout randomly removes every neuron in the (classical) network with probability $p$ by sampling a $\mathsf{Bernoulli}(p)$ $0/1$ random variable for each neuron. This effectively severs all input and output weight connections to this neuron so they do not contribute to the output. This effectively samples sub-networks from the parent network on each forward pass.

The closest analogue to this behaviour in quantum neural networks (specifically parameterised quantum circuits (PQCs) is to randomly drop \emph{gates} in a trainable circuit (equivalently randomly set their parameters to zero)~\cite{scala_general_2023, kobayashi_overfitting_2022} which has variations known as \emph{entangling} or \emph{rotation} dropout. As noted by~\cite{scala_general_2023}, even this notion of `quantum' dropout is not completely analogous to its classical counterpart, since it only removing single qubit gates make not sever temporal connections between qubits (as in the classical case) due to entanglement.
\subsection{Dropout interpretation of density \texorpdfstring{QNNs}{}} \label{app:dropout_density}

Nevertheless, as mentioned in the main text, the density QNN framework is also sometimes referred to a quantum version of dropout, since it bears some surface similarities. Given some `dropout' probabilities (the probabilities of the sub-unitaries, $\{\alpha_k\}_{k=1}^K$), each forward pass involves $K-1$ sub-unitaries $\{U_{k'}(\theta_{k'})\}_{k'\neq k}$ being `dropped out', and avoids the model relying too heavily on any specific trainable operation (or subset of parameters).

We argue that this model, as it is presented in the na\"ive form, is \emph{not} sufficiently close to mimic dropout for (at least) one crucial reason. A key feature of dropout is the different training and inference behaviour. Common deep learning packages such as \texttt{pytorch} have specific methods for models, \texttt{.train()} and \texttt{.eval()} which, when activated, imply different behaviour for layers such as dropout. More specifically, in the training phase, a dropout layer randomly drops neurons with probability $p$, However in the the \emph{evaluation}/inference phase, dropout has the behaviour that the \emph{full} network is applied, but with the adaptation that the weight matrix is scaled by the probability $p$. As a result, the actual output of inference through the network at test time, is the same as the \emph{expectation} of inference through the network in training.

In the density QNN framework, this presents a problem. In order to more correctly mimic the behaviour of dropout, we need the density QNN to have a \texttt{.eval()} mode where a single forward pass is equivalent to the \emph{on average} evaluation of the density QNN in an \texttt{.eval()} mode.

For evaluation, this means we literally need to prepare the density state, 
\begin{equation} \label{eqn:density_qnns_app_dropout_1}
    \rho(\boldsymbol{\theta}, \boldsymbol{\alpha}, \boldsymbol{x}) = \sum_{k=1}^K \alpha_k U_k(\boldsymbol{\theta}_k)\ketbra{\boldsymbol{x}}{\boldsymbol{x}}U^\dagger_k(\boldsymbol{\theta}_k)
\end{equation}
on a quantum computer, which may be highly non-trivial in general (and in general exponential), particularly in the case of quantum data. In the following section, we give an \texttt{.eval()} mode to do this in a specific case, so the model more closely resembles a dropout network, but taking such an interpretation will come with limitations, as we discuss in the following.

First, we separate the two operation modes of a density QNN explicitly as follows:
\begin{itemize}
    \item {\textbf{\textsf{Train:}} For each datapoint $\boldsymbol{x}$, create the state $\ket{\boldsymbol{x}}$ by applying the loading unitary, $V(\boldsymbol{x})$ to the initial state, $\ket{0}^{\otimes n}$. Then sample and index, $k \sim \alpha_k$, and apply sub-unitary $U_k(\boldsymbol{\theta}_k)$ with probability $\alpha_k$ to the state $\ket{\boldsymbol{x}}$. Measuring the output observable $\mathcal{H}$ will, on expectation, evaluate $\Tr(\mathcal{H} \rho(\boldsymbol{\theta}, \boldsymbol{\alpha}, \boldsymbol{x}))$. Each sub-unitary is trained individually, as described in the main text. However, including parameter sharing between the sub-unitaries may be more reminiscent of dropout. This is implemented via the \emph{randomised} density QNN in~\figref{fig:density_qnn_mainfig}c in the main text.
    }
    \item {\textbf{\textsf{Test:}} In evaluation mode, we must \emph{directly} create the state $\rho(\boldsymbol{\theta}, \boldsymbol{\alpha}, \boldsymbol{x})$ and then measure the observable $\mathcal{H}$. This is implemented via the \emph{deterministic} density QNN in~\figref{fig:density_qnn_mainfig}b in the main text which prepares the state via the linear combination of unitaries. We show two examples,~\figref{fig:matrix_loader_dropout_app}a shows a generic case of a density QNN which may have large depth and~\figref{fig:matrix_loader_dropout_app}b specialises to a generative application and Hamming weight preserving unitaries, which can have a much more conservative depth scaling. This uses a mixed unary/binary representation on the qubits to create the state $\rho(\boldsymbol{\theta}, \boldsymbol{\alpha}, \boldsymbol{x})$ on the bottom register, $\mathcal{B}$. Clearly, such an implementation sacrifices the efficient and shallow implementation from the training phase, but it also provides a relatively general method to implement density quantum neural networks in a less NISQ-friendly manner. We describe the details of the circuit operation in the following sections~\appref{app:dropout_density_matrix_loader},~\appref{app:dropout_density_matrix_loader_generative}.
}
\end{itemize}

\begin{figure*}[!ht]
        \includegraphics[width=\linewidth]{Fig13_matrix_loaders_unary_generic.pdf}
    \caption{
    \textsf{
    \textbf{Circuits for preparing density QNN state.}
    a) is the deterministic state preparation from~\figref{fig:density_qnn_mainfig}b while b) is the simplification in the case the sub-unitaries are Hamming weight preserving and behaves as a matrix loader, potentially suitable for generative modelling.
    }
    }
    \label{fig:matrix_loader_dropout_app}
\end{figure*}

\subsection{Evaluation circuit for density \texorpdfstring{QNNs}{}} \label{app:dropout_density_matrix_loader}
Here we describe how the deterministic circuit in~\figref{fig:density_qnn_mainfig}b prepares the state~\eqref{eqn:density_qnns_app_dropout_1}. In the special case we desribe below, such circuits can function as \emph{matrix loaders}. The matrix loader~\cite{cherrat_quantum_2022} was originally intended to load an $n\times d$ matrix, $\boldsymbol{x}$, into an overall Hamming weight $2$ state using two unary qubit registers; one to index the matrix rows and the other to index the columns as follows: $\ket{\boldsymbol{x}} = \frac{1}{\|\boldsymbol{x}\|}\sum_{i=1}^n\sum_{j=1}^d \boldsymbol{x}_{i, j}\ket{\mathbf{e}_i}\ket{\mathbf{e}_j}$. Intuitively, this works by loading first the column indices to the top register. Controlled on these `row' indices, `row' loaders are applied (Figure 5 in~\cite{cherrat_quantum_2022}) to the qubits in the bottom register ($\mathcal{B}$) in the figure. Due to the unary encoding on the top register, each control will be only activated corresponding to the row that qubit is indexing. 

We adapt this idea here to prepare the state $\rho(\boldsymbol{\theta}, \boldsymbol{\alpha}, \boldsymbol{x})$. First, we load the distribution of sub-unitaries, $\{\alpha_k\}$ onto the top $K$ qubits. This produces the Hamming weight $1$ state in the register $\mathcal{A}$. We can simultaneously prepare the initial data state $\ket{\boldsymbol{x}}$ by applying $V(\boldsymbol{x})$ on the register $\mathcal{B}$:
\begin{equation}
    \mathsf{Load}\left(\sqrt{\boldsymbol{\alpha}}\right)\ket{0}_{\mathcal{A}}^{\otimes n}V(\boldsymbol{x})\ket{0}_{\mathcal{B}}^{\otimes n} = \sum_{k=1}^K \sqrt{\alpha_k} \ket{\mathbf{e}_k}_{\mathcal{A}} \ket{\boldsymbol{x}}_{\mathcal{B}}
\end{equation}
Now, iterating through the top $K$ qubits and applying $U_k(\theta_k)$ on the register $\mathcal{B}$ controlled on qubit $k$ in register $\mathcal{A}$ results in:
\begin{align}
     \sum_{k=1}^K \sqrt{\alpha_k} \ket{\mathbf{e_k}}_{\mathcal{A}} \ket{\boldsymbol{x}}_{\mathcal{B}} \rightarrow & \sum_{k=1}^K \sqrt{\alpha_k} \ket{\mathbf{e}_k}_{\mathcal{A}} U_k(\boldsymbol{\theta})\ket{\boldsymbol{x}}_{\mathcal{B}}, \\
     \implies & \rho_{\mathcal{A}\mathcal{B}} = \sum_{k=1}^K\sum_{j=1}^K \sqrt{\alpha_k} \sqrt{\alpha_j}\ketbra{\mathbf{e}_k}{\mathbf{e}_j}_{\mathcal{A}} \left[U_k(\boldsymbol{\theta})\ketbra{\boldsymbol{x}}{\boldsymbol{x}}U^\dagger_k(\boldsymbol{\theta})\right]_{\mathcal{B}}
\end{align}
Finally, $\rho(\boldsymbol{\theta}, \boldsymbol{\alpha}, \boldsymbol{x}) = \tr_{\mathcal{A}}\left(\rho_{\mathcal{A}\mathcal{B}}\right)$ in \eqref{eqn:density_qnns_app_dropout_1} is prepared by tracing out register $\mathcal{A}$, leaving only the trace-full diagonal elements $\ketbra{\mathbf{e}_k}{\mathbf{e}_k}$ with trace $=1$. 

There are some final notes on this point:
\begin{enumerate}
    \item This technique of applying controlled unitaries a circuit is well-known as the \emph{linear combination of unitaries} (LCU) method, which is a primary method of performing quantum simulation on a quantum computer. The LCU method emulates the effect of a non-unitary matrix $A$ on a state which can be decomposed as a linear combination of unitary operations, $A = \sum_i \alpha_i U_i$. While this work was in preparation, we became aware of~\cite{heredge_non-unitary_2024} which proposes exactly the LCU method for quantum machine learning. However, this differs from the proposal in this work as we are interested in trading off efficiency and trainability for already defined models as discussed in the main text. 
    \item Dealing with the density state as here means we do not require post selection on the top register, $\mathcal{A}$. Post-selecting on a particular outcome, e.g. $\ket{0}^{\otimes n}$ adds an addition overhead to the overall model, but is necessary for correctly applying the desired matrix, $A$, to the input.
    \item If we need to reuse the ancillary qubits in register $\mathcal{A}$ for another purpose after the creation of the density state, we will need to uncompute the qubits with the operation $\mathsf{Load}^\dagger\left(\sqrt{\boldsymbol{\alpha}}\right)$.
\end{enumerate}

\subsection{Sampling from density \texorpdfstring{QNNs}{}}
\label{app:dropout_density_matrix_loader_generative}
If we have a case where there is no data to be encoded into the circuit, i.e. $V(\boldsymbol{x}) = \mathds{1}$ and the sub-unitaries $\mathcal{U} = \{U_k\}_{k=1}^K$ are all Hamming weight preserving, we can use a closer analogue to the matrix loader of~\cite{cherrat_quantum_2022} to prepare the state $\rho(\boldsymbol{\theta}, \boldsymbol{\alpha})$ (notice $\boldsymbol{x}$ independence) as in~\figref{fig:matrix_loader_dropout_app}b. Instead of directly controlling on the unitaries $U_k$, we instead interleave a CNOT gate between the unitary and its inverse. Since $U_k$ is Hamming weight preserving, it will be activated with probability $\alpha_k$ as before, but since the initial state is Hamming weight $0$, the inverses, $U^\dagger_k(\theta_k)$, will not apply for that particular $k$, only the CNOT which activates an initial unary state followed by the unitary $U_k(\theta_k)$, which preserves the Hamming weight $1$ state on the $\mathcal{B}$ register. If these Hamming weight preserving unitaries are vector loaders~\cite{cherrat_quantum_2022}, the output state $\rho(\boldsymbol{\theta}, \boldsymbol{\alpha})$ will correspond exactly the reduced state of \emph{some}, unknown, matrix loader state $\ket{\boldsymbol{x}^*}$ generated by the angles $\boldsymbol{\theta}$:
\begin{align*}
    \ket{\boldsymbol{x}^*} &= \sum_{i=1}^n\sum_{k=1}^d \sqrt{\alpha_k}\boldsymbol{x}^*_{i, k}\ket{\mathbf{e}_k}_{\mathcal{A}}\ket{\mathbf{e}_i}_{\mathcal{B}}\\
    \rho_{\mathcal{A}\mathcal{B}} &= \sum_{i=1}^n\sum_{i'=1}^n\sum_{k=1}^d \sum_{k'=1}^d \sqrt{\alpha_k\alpha_{k'}}\boldsymbol{x}^*_{i, k}\boldsymbol{x}^*_{i', k'}\ketbra{\mathbf{e}_k}{\mathbf{e}_{k'}}_{\mathcal{A}}\ketbra{\mathbf{e}_i}{\mathbf{e}_i}_{\mathcal{B}}\\
    \implies \rho(\boldsymbol{\theta}, \boldsymbol{\alpha}) &= \sum_{i=1}^n\sum_{i'=1}^n\sum_{k=1}^d \alpha_k\boldsymbol{x}^*_{i, k}\boldsymbol{x}^*_{i', k}\ketbra{\mathbf{e}_i}{\mathbf{e}_{i'}}_{\mathcal{B}}
\end{align*}

Here, we can view the density QNN as preparing a generative state (akin to a Born machine~\cite{cheng_information_2018, liu_differentiable_2018, benedetti_generative_2019, coyle_born_2020}) and sampling the state in the computational basis can correspond to sampling an index $i \in [n]$ with probability weighted by $\sum_k\alpha_k^2\left(\boldsymbol{x}^*_{i, k}\right)^2$, where $\boldsymbol{x}^*_{*, k}$
is a vector (suitably normalised) derived from the angles $\boldsymbol{\theta}$. One could also generate a more efficient circuit in the unary space which \emph{includes} an initial state preparation unitary using recent techniques from quantum fourier networks~\cite{jain_quantum_2024}.

Here, we note the connection to the recently proposed \emph{variational} measurement-based quantum computing (MBQC), which was applied to generative modelling~\cite{majumder_variational_2023}. There, a distribution over ``sub-unitaries'' appears naturally due the the nature of measurement-driven quantum computation. Specifically, each time a qubit is measured in the MBQC model, its output result is used to fork the next level of computation. In order to deterministically implement a single (yet arbitrary) unitary, MBQC corrects the `wrong' path in the fork by applying corrective rotations on subsequent qubits. Rather than being motivated by implementing a single (known) unitary via deterministic correction, Ref.~\cite{majumder_variational_2023} proposes to use this inherent MBQC measurement randomness for generative modelling purposes, as the effect of \emph{not} correcting outcomes results exactly in a mixed-unitary (density) channel as we have above. Similarly to our proposal, the authors demonstrated superior learning capabilities of the mixed channel over a single unitary ansatz. We hope that the parallel tracks traversed in our work, along with variational MBQC, and post-variational quantum machine learning can be unified to ultimately advance the field.

\section{Combating overfitting with density QNNs}
\label{app:overfitting_density_QNN}

In~\appref{app_sec:reuploading_density}, we found that using a data reuploading density QNN model with a Hamiltonian data encoding strategy results in the model output being a randomised linear combination of Fourier series in the data as below:
\begin{equation}\label{eqn:single_reup_layer_fourier_term_dropout_section}
g(\boldsymbol{\theta}, \boldsymbol{\alpha}, \boldsymbol{x}) = \Tr\left(\mathcal{O}\rho^R(\boldsymbol{\theta}, \boldsymbol{\alpha}, \boldsymbol{x})\right) = \sum_{k=1}^K \alpha_k f^k(\boldsymbol{\theta}, \boldsymbol{x}) = \sum_{k=1}^K \alpha_k \sum_{\boldsymbol{\omega}_k\in \Omega_k} c_{\boldsymbol{\omega}_k}(\boldsymbol{\theta}) e^{i \boldsymbol{\omega}_k^\top \boldsymbol{x}}
\end{equation}
One of the primary features of dropout classically is to prevent individual model complexity which in turn combats overfitting. While the density QNN may \emph{appear} as a dropout mechanism, in order to truly be compared to dropout, it must do what dropout does - i.e., regulate overfitting also.

We conclude the discussion on this topic by demonstrating that a density QNN \emph{can} indeed prevent overfitting with a simple and interpretable example. Specifically, we use only a single qubit, $n=1$ and compare a density QNN model, outputting functions $g(\boldsymbol{\theta}, \boldsymbol{\alpha}, \boldsymbol{x})$ to a pure state version, which would just learn a single term in the sum~\eqref{eqn:single_reup_layer_fourier_term_dropout_section}, $f^k(\boldsymbol{\theta}, \boldsymbol{x})$ (without lack of generality we drop the index $k$ hereafter). This will also tie into the mixture of experts (MoE) as discussed in~\appref{app_sec:moe_qnns}.

\begin{figure}
    \centering
    \includegraphics[width=0.95\textwidth]{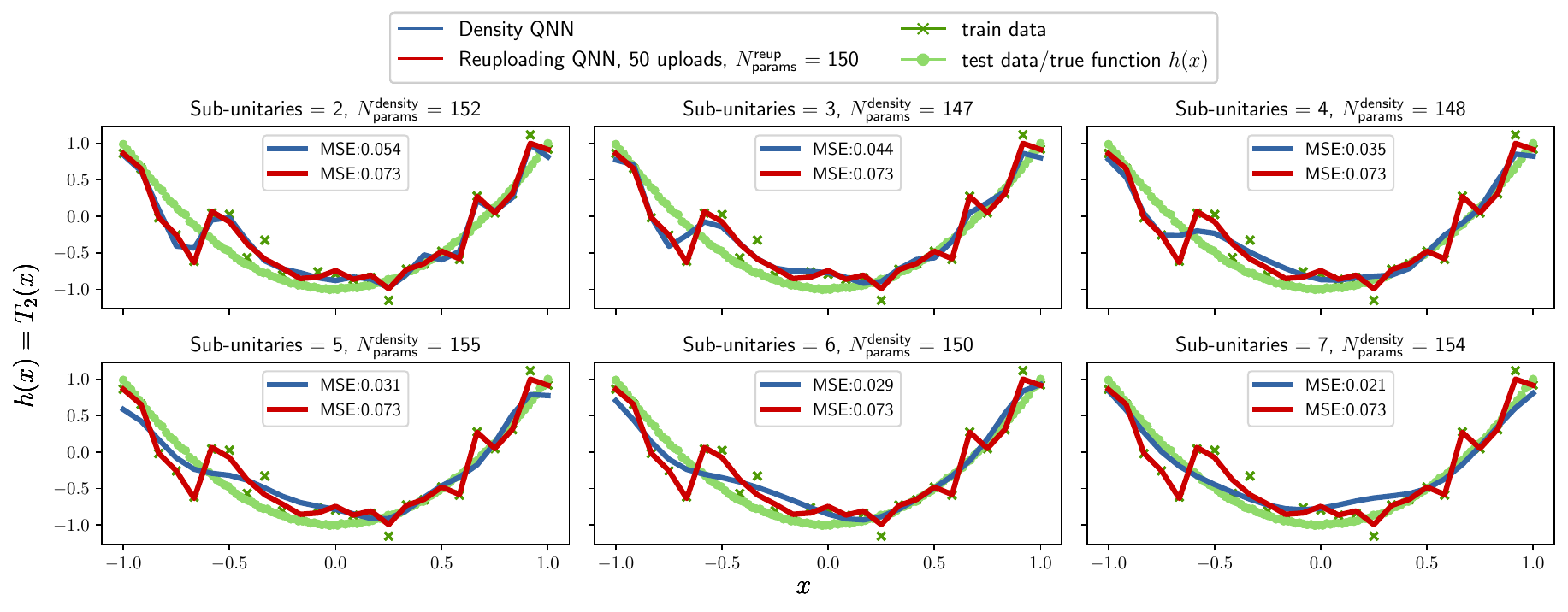}
  \caption{
  \textsf{\textbf{Regulating overfitting with a density QNN.} \\
  Each panel compares a single qubit density QNN (blue) with some number of sub-unitaries ($K \in \{2, 7\}$) to a vanilla (red) single qubit data reuploading model with $50$ layers or reuploads of the data. The underlying function to fit is the second Chebyshev polynomial of the first kind, $h(x) = T_2(x)$. The number of quantum circuit parameters are kept approximately the same in each case. For a fixed parameter budget, the density QNN is less prone to overfitting to training data than the vanilla counterpart, and generalisation improves (test MSE decreases) as the number of sub-unitaries increases.
  }
  }
    \label{subfig:chebyshev_2_num_experts}
\end{figure}

\begin{figure}
\centering
    \includegraphics[width=0.75\textwidth]{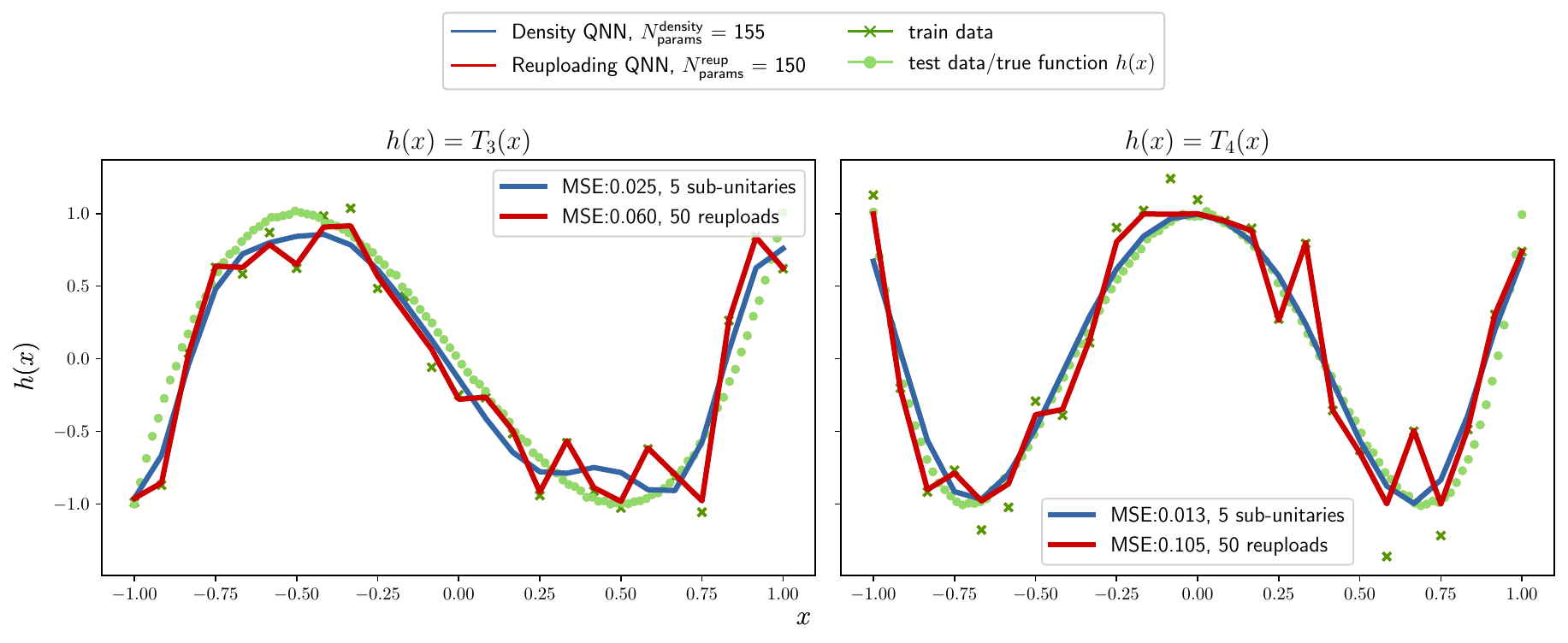}
  \caption{
  \textsf{\textbf{Regulating overfitting with a density QNN.} \\
   Higher degree Chebyshev polynomials, $h(x) = T_{\{3, 4\}}(x)$ learned with $5$ sub-unitaries in the density QNN (blue) versus the vanilla data reuploading model (red).
  }
  }
    \label{fig:chebyshev_3_4_num_experts}
\end{figure}

\subsection{Data and model training}
\label{app:overfitting_data}

For a toy example, we choose a regression problem, where both models, $g(\boldsymbol{\theta}, \boldsymbol{\alpha}, \boldsymbol{x}), f(\boldsymbol{\theta}, \boldsymbol{x})$ are tasked to predict the output of Chebyshev polynomials of the first kind, defined by:
\begin{equation}\label{eqn:chebyshev_polynomials_app}
    T_0(x) = 1, \qquad T_1(x) = x,  \qquad T_{n+1}(x) = 2xT_n(x) - T_{n-1}(x), 
\end{equation}
The data is generated to highlight a scenario where the original reuploading model is encouraged to overfit - we generate $25$ training points from $h(x) := T_n(x)$ and $110$ test points in the interval $[-1, 1]$. We then add noise to the training data from a zero mean normal distribution with standard deviation $0.2$. The metric we use is the mean squared error (MSE) between the test data, and the predictions from each model, $g(\boldsymbol{\theta}, \boldsymbol{\alpha}, \boldsymbol{x}), f(\boldsymbol{\theta}, \boldsymbol{x})$. This serves as a proxy for underlying fitting.
In all the below, we use a batch size of $5$ to train and the Adam optimiser with a learning rate of $0.001$.

\begin{figure}
    \centering
    \includegraphics[width=0.75\textwidth]{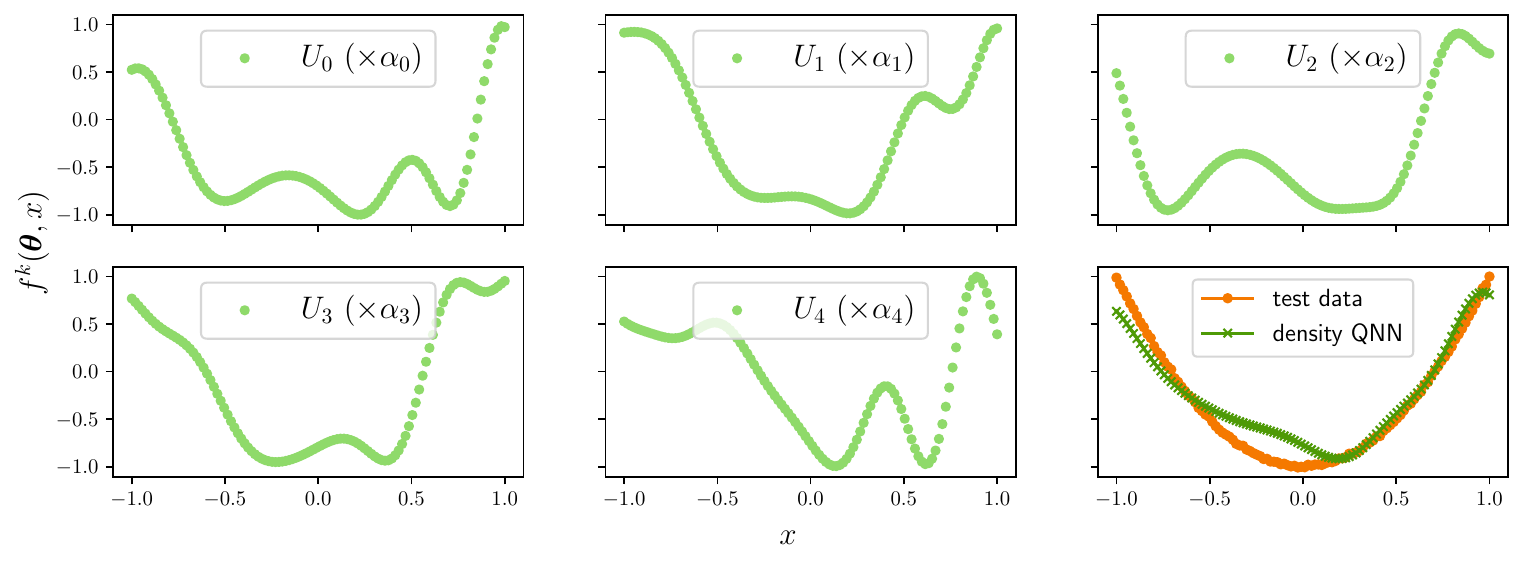}
      \caption{
  \textsf{\textbf{Partial Fourier series learned for each sub-unitary (illustrated in~\figref{fig:data_reuploading_overfitting}a).}\\
  Each panel shows the Fourier output of a single sub-unitary ($U_k(\boldsymbol{\theta})$) on the test datapoints from the second Chebyshev polynomial, $T_2(x)$. The final (bottom right) panel shows the complete model, $f(\boldsymbol{\theta}, \boldsymbol{\alpha}, x)$ with the output of each sub-unitary, $k$, weighted by the corresponding $\alpha_k$. We also plot the true test data.
  }
  }
\label{fig:density_partial_fourier}
\end{figure}

\subsection{Results}
\label{app:overfitting_results}

We begin with $h(x) = T_2(x)$. \figref{subfig:chebyshev_2_num_experts} shows the results. We compare a `vanilla' data reuploading model on a single qubit, i.e. the output function is the following:
\begin{align} \label{eqn:vanilla_reupload}
    \text{Vanilla} &\implies f(\boldsymbol{\theta}, x) = \bra{0}\left(\prod_{\ell=1}^L U_l(\boldsymbol{\theta}, x)\right)^{\dagger}Z\left(\prod_{j=1}^L U_j(\boldsymbol{\theta}, x)\right) \ket{0}, \\ 
    \label{eqn:density_qnn_upload}
    \text{Density QNN} &\implies f(\boldsymbol{\theta}, \boldsymbol{\alpha}, x) = \sum_{k=1}^K \alpha_k \bra{0}\left(\prod_{\ell=1}^{L'} U^k_l(\boldsymbol{\theta}, x)\right)^{\dagger}Z\left(\prod_{j=1}^{L'} U^k_j(\boldsymbol{\theta}, x)\right) \ket{0}\\
    U^{(k)}_l(\boldsymbol{\theta}, x)  &:= R^{\ell}_y(x)R^{\ell}_x(x)R^{\ell}_z(\phi)R^{\delta}_y(x)R^{\ell}_z(\omega), \boldsymbol{\theta}_{\ell} = \{\phi_{\ell}, \delta_{\ell}, \omega_{\ell}\}
\end{align}
Where each layer (reupload) has $3$ trainable parameters. Each panel in the figure contains a different number of sub-unitaries in the density QNN. Here we use $L=50$ reuploading layers so we have $150$ parameters in the vanilla model. As a result of the large depth, the model is more prone to overfitting because of the complexity of the resulting Fourier series. In the density versions we choose a number of reuploads, $L' < L$ such that the number of trainable parameters in the sub-unitaries ($K$) is approximately the same in each case ($\approx 150$). We see that as the number of sub-unitaries increases, the density QNN model generalises to the test data more accurately (MSE decreases as $K$ increases). This is simply because each component partial Fourier series of $f(\boldsymbol{\theta}, \boldsymbol{\alpha}, \boldsymbol{x})$ is less complex than that of the full vanilla model, and hence less prone to overfitting. This further indicates that the density QNN model may be a better choice than a vanilla pure state PQC model, given a fixed budget of parameter allocation.

Next, we test on higher degree Chebyshev polynomials, $T_3(x)$ and $T_4(x)$ in~\figref{fig:chebyshev_3_4_num_experts}. Again the density model is able to fit the test data well and again each combats overfitting by reallocating parameters from an overparameterised vanilla reuploading model. 

Finally, in~\figref{fig:density_partial_fourier} we plot the output of each sub-unitary within a density model, which are combined using the weighting parameters, $\boldsymbol{\alpha}$, in the final model.


\begin{thebibliography}{10}
\providecommand{\url}[1]{#1}
\csname url@samestyle\endcsname
\providecommand{\newblock}{\relax}
\providecommand{\bibinfo}[2]{#2}
\providecommand{\BIBentrySTDinterwordspacing}{\spaceskip=0pt\relax}
\providecommand{\BIBentryALTinterwordstretchfactor}{4}
\providecommand{\BIBentryALTinterwordspacing}{\spaceskip=\fontdimen2\font plus
\BIBentryALTinterwordstretchfactor\fontdimen3\font minus \fontdimen4\font\relax}
\providecommand{\BIBforeignlanguage}[2]{{%
\expandafter\ifx\csname l@#1\endcsname\relax
\typeout{** WARNING: IEEEtran.bst: No hyphenation pattern has been}%
\typeout{** loaded for the language `#1'. Using the pattern for}%
\typeout{** the default language instead.}%
\else
\language=\csname l@#1\endcsname
\fi
#2}}
\providecommand{\BIBdecl}{\relax}
\BIBdecl

\bibitem{johri_nearest_2021}
\BIBentryALTinterwordspacing
S.~Johri, S.~Debnath, A.~Mocherla, A.~Singk, A.~Prakash, J.~Kim, and I.~Kerenidis, ``\BIBforeignlanguage{en}{Nearest centroid classification on a trapped ion quantum computer},'' \emph{\BIBforeignlanguage{en}{npj Quantum Inf}}, vol.~7, no.~1, pp. 1--11, Aug. 2021. [Online]. Available: \url{https://www.nature.com/articles/s41534-021-00456-5}
\BIBentrySTDinterwordspacing

\bibitem{rosenblatt_perceptron_1958}
F.~Rosenblatt, ``The perceptron: {A} probabilistic model for information storage and organization in the brain,'' \emph{Psychological Review}, vol.~65, no.~6, pp. 386--408, 1958.

\bibitem{lecun_deep_2015}
\BIBentryALTinterwordspacing
Y.~LeCun, Y.~Bengio, and G.~Hinton, ``\BIBforeignlanguage{en}{Deep learning},'' \emph{\BIBforeignlanguage{en}{Nature}}, vol. 521, no. 7553, pp. 436--444, May 2015. [Online]. Available: \url{https://www.nature.com/articles/nature14539}
\BIBentrySTDinterwordspacing

\bibitem{bahdanau_neural_2016}
\BIBentryALTinterwordspacing
D.~Bahdanau, K.~Cho, and Y.~Bengio, ``Neural {Machine} {Translation} by {Jointly} {Learning} to {Align} and {Translate},'' May 2016. [Online]. Available: \url{http://arxiv.org/abs/1409.0473}
\BIBentrySTDinterwordspacing

\bibitem{vaswani_attention_2017}
\BIBentryALTinterwordspacing
A.~Vaswani, N.~Shazeer, N.~Parmar, J.~Uszkoreit, L.~Jones, A.~N. Gomez, L.~Kaiser, and I.~Polosukhin, ``Attention is {All} you {Need},'' in \emph{Advances in {Neural} {Information} {Processing} {Systems}}, I.~Guyon, U.~V. Luxburg, S.~Bengio, H.~Wallach, R.~Fergus, S.~Vishwanathan, and R.~Garnett, Eds., vol.~30.\hskip 1em plus 0.5em minus 0.4em\relax Curran Associates, Inc., 2017. [Online]. Available: \url{https://proceedings.neurips.cc/paper_files/paper/2017/file/3f5ee243547dee91fbd053c1c4a845aa-Paper.pdf}
\BIBentrySTDinterwordspacing

\bibitem{silver_mastering_2017}
\BIBentryALTinterwordspacing
D.~Silver, T.~Hubert, J.~Schrittwieser, I.~Antonoglou, M.~Lai, A.~Guez, M.~Lanctot, L.~Sifre, D.~Kumaran, T.~Graepel, T.~Lillicrap, K.~Simonyan, and D.~Hassabis, ``Mastering {Chess} and {Shogi} by {Self}-{Play} with a {General} {Reinforcement} {Learning} {Algorithm},'' Dec. 2017. [Online]. Available: \url{http://arxiv.org/abs/1712.01815}
\BIBentrySTDinterwordspacing

\bibitem{brown_language_2020}
T.~B. Brown, B.~Mann, N.~Ryder, M.~Subbiah, J.~Kaplan, P.~Dhariwal, A.~Neelakantan, P.~Shyam, G.~Sastry, A.~Askell, S.~Agarwal, A.~Herbert-Voss, G.~Krueger, T.~Henighan, R.~Child, A.~Ramesh, D.~M. Ziegler, J.~Wu, C.~Winter, C.~Hesse, M.~Chen, E.~Sigler, M.~Litwin, S.~Gray, B.~Chess, J.~Clark, C.~Berner, S.~McCandlish, A.~Radford, I.~Sutskever, and D.~Amodei, ``Language models are few-shot learners,'' in \emph{Proceedings of the 34th {International} {Conference} on {Neural} {Information} {Processing} {Systems}}, ser. {NIPS} '20.\hskip 1em plus 0.5em minus 0.4em\relax Red Hook, NY, USA: Curran Associates Inc., Dec. 2020, pp. 1877--1901.

\bibitem{ramesh_zero-shot_2021}
\BIBentryALTinterwordspacing
A.~Ramesh, M.~Pavlov, G.~Goh, S.~Gray, C.~Voss, A.~Radford, M.~Chen, and I.~Sutskever, ``\BIBforeignlanguage{en}{Zero-{Shot} {Text}-to-{Image} {Generation}},'' in \emph{\BIBforeignlanguage{en}{Proceedings of the 38th {International} {Conference} on {Machine} {Learning}}}.\hskip 1em plus 0.5em minus 0.4em\relax PMLR, Jul. 2021, pp. 8821--8831. [Online]. Available: \url{https://proceedings.mlr.press/v139/ramesh21a.html}
\BIBentrySTDinterwordspacing

\bibitem{rumelhart_learning_1986}
\BIBentryALTinterwordspacing
D.~E. Rumelhart, G.~E. Hinton, and R.~J. Williams, ``\BIBforeignlanguage{en}{Learning representations by back-propagating errors},'' \emph{\BIBforeignlanguage{en}{Nature}}, vol. 323, no. 6088, pp. 533--536, Oct. 1986. [Online]. Available: \url{https://www.nature.com/articles/323533a0}
\BIBentrySTDinterwordspacing

\bibitem{sivak_real-time_2023}
\BIBentryALTinterwordspacing
V.~V. Sivak, A.~Eickbusch, B.~Royer, S.~Singh, I.~Tsioutsios, S.~Ganjam, A.~Miano, B.~L. Brock, A.~Z. Ding, L.~Frunzio, S.~M. Girvin, R.~J. Schoelkopf, and M.~H. Devoret, ``\BIBforeignlanguage{en}{Real-time quantum error correction beyond break-even},'' \emph{\BIBforeignlanguage{en}{Nature}}, vol. 616, no. 7955, pp. 50--55, Apr. 2023. [Online]. Available: \url{https://www.nature.com/articles/s41586-023-05782-6}
\BIBentrySTDinterwordspacing

\bibitem{acharya_quantum_2024}
\BIBentryALTinterwordspacing
R.~Acharya, L.~Aghababaie-Beni, I.~Aleiner, T.~I. Andersen, M.~Ansmann, F.~Arute, K.~Arya, A.~Asfaw, N.~Astrakhantsev, J.~Atalaya, R.~Babbush, D.~Bacon, B.~Ballard, J.~C. Bardin, J.~Bausch, A.~Bengtsson, A.~Bilmes, S.~Blackwell, S.~Boixo, G.~Bortoli, A.~Bourassa, J.~Bovaird, L.~Brill, M.~Broughton, D.~A. Browne, B.~Buchea, B.~B. Buckley, D.~A. Buell, T.~Burger, B.~Burkett, N.~Bushnell, A.~Cabrera, J.~Campero, H.-S. Chang, Y.~Chen, Z.~Chen, B.~Chiaro, D.~Chik, C.~Chou, J.~Claes, A.~Y. Cleland, J.~Cogan, R.~Collins, P.~Conner, W.~Courtney, A.~L. Crook, B.~Curtin, S.~Das, A.~Davies, L.~D. Lorenzo, D.~M. Debroy, S.~Demura, M.~Devoret, A.~D. Paolo, P.~Donohoe, I.~Drozdov, A.~Dunsworth, C.~Earle, T.~Edlich, A.~Eickbusch, A.~M. Elbag, M.~Elzouka, C.~Erickson, L.~Faoro, E.~Farhi, V.~S. Ferreira, L.~F. Burgos, E.~Forati, A.~G. Fowler, B.~Foxen, S.~Ganjam, G.~Garcia, R.~Gasca, √.~Genois, W.~Giang, C.~Gidney, D.~Gilboa, R.~Gosula, A.~G. Dau, D.~Graumann, A.~Greene, J.~A. Gross, S.~Habegger, J.~Hall, M.~C. Hamilton,
  M.~Hansen, M.~P. Harrigan, S.~D. Harrington, F.~J.~H. Heras, S.~Heslin, P.~Heu, O.~Higgott, G.~Hill, J.~Hilton, G.~Holland, S.~Hong, H.-Y. Huang, A.~Huff, W.~J. Huggins, L.~B. Ioffe, S.~V. Isakov, J.~Iveland, E.~Jeffrey, Z.~Jiang, C.~Jones, S.~Jordan, C.~Joshi, P.~Juhas, D.~Kafri, H.~Kang, A.~H. Karamlou, K.~Kechedzhi, J.~Kelly, T.~Khaire, T.~Khattar, M.~Khezri, S.~Kim, P.~V. Klimov, A.~R. Klots, B.~Kobrin, P.~Kohli, A.~N. Korotkov, F.~Kostritsa, R.~Kothari, B.~Kozlovskii, J.~M. Kreikebaum, V.~D. Kurilovich, N.~Lacroix, D.~Landhuis, T.~Lange-Dei, B.~W. Langley, P.~Laptev, K.-M. Lau, L.~L. Guevel, J.~Ledford, K.~Lee, Y.~D. Lensky, S.~Leon, B.~J. Lester, W.~Y. Li, Y.~Li, A.~T. Lill, W.~Liu, W.~P. Livingston, A.~Locharla, E.~Lucero, D.~Lundahl, A.~Lunt, S.~Madhuk, F.~D. Malone, A.~Maloney, S.~Mandr√°, L.~S. Martin, S.~Martin, O.~Martin, C.~Maxfield, J.~R. McClean, M.~McEwen, S.~Meeks, A.~Megrant, X.~Mi, K.~C. Miao, A.~Mieszala, R.~Molavi, S.~Molina, S.~Montazeri, A.~Morvan, R.~Movassagh, W.~Mruczkiewicz,
  O.~Naaman, M.~Neeley, C.~Neill, A.~Nersisyan, H.~Neven, M.~Newman, J.~H. Ng, A.~Nguyen, M.~Nguyen, C.-H. Ni, T.~E. O'Brien, W.~D. Oliver, A.~Opremcak, K.~Ottosson, A.~Petukhov, A.~Pizzuto, J.~Platt, R.~Potter, O.~Pritchard, L.~P. Pryadko, C.~Quintana, G.~Ramachandran, M.~J. Reagor, D.~M. Rhodes, G.~Roberts, E.~Rosenberg, E.~Rosenfeld, P.~Roushan, N.~C. Rubin, N.~Saei, D.~Sank, K.~Sankaragomathi, K.~J. Satzinger, H.~F. Schurkus, C.~Schuster, A.~W. Senior, M.~J. Shearn, A.~Shorter, N.~Shutty, V.~Shvarts, S.~Singh, V.~Sivak, J.~Skruzny, S.~Small, V.~Smelyanskiy, W.~C. Smith, R.~D. Somma, S.~Springer, G.~Sterling, D.~Strain, J.~Suchard, A.~Szasz, A.~Sztein, D.~Thor, A.~Torres, M.~M. Torunbalci, A.~Vaishnav, J.~Vargas, S.~Vdovichev, G.~Vidal, B.~Villalonga, C.~V. Heidweiller, S.~Waltman, S.~X. Wang, B.~Ware, K.~Weber, T.~White, K.~Wong, B.~W.~K. Woo, C.~Xing, Z.~J. Yao, P.~Yeh, B.~Ying, J.~Yoo, N.~Yosri, G.~Young, A.~Zalcman, Y.~Zhang, N.~Zhu, and N.~Zobrist, ``Quantum error correction below the surface code
  threshold,'' Aug. 2024. [Online]. Available: \url{http://arxiv.org/abs/2408.13687}
\BIBentrySTDinterwordspacing

\bibitem{silva_demonstration_2024}
\BIBentryALTinterwordspacing
M.~P.~d. Silva, C.~Ryan-Anderson, J.~M. Bello-Rivas, A.~Chernoguzov, J.~M. Dreiling, C.~Foltz, F.~Frachon, J.~P. Gaebler, T.~M. Gatterman, L.~Grans-Samuelsson, D.~Hayes, N.~Hewitt, J.~Johansen, D.~Lucchetti, M.~Mills, S.~A. Moses, B.~Neyenhuis, A.~Paz, J.~Pino, P.~Siegfried, J.~Strabley, A.~Sundaram, D.~Tom, S.~J. Wernli, M.~Zanner, R.~P. Stutz, and K.~M. Svore, ``Demonstration of logical qubits and repeated error correction with better-than-physical error rates,'' Apr. 2024. [Online]. Available: \url{http://arxiv.org/abs/2404.02280}
\BIBentrySTDinterwordspacing

\bibitem{benedetti_parameterized_2019}
\BIBentryALTinterwordspacing
M.~Benedetti, E.~Lloyd, S.~Sack, and M.~Fiorentini, ``\BIBforeignlanguage{en}{Parameterized quantum circuits as machine learning models},'' \emph{\BIBforeignlanguage{en}{Quantum Sci. Technol.}}, vol.~4, no.~4, p. 043001, Nov. 2019. [Online]. Available: \url{https://dx.doi.org/10.1088/2058-9565/ab4eb5}
\BIBentrySTDinterwordspacing

\bibitem{bharti_noisy_2022}
\BIBentryALTinterwordspacing
K.~Bharti, A.~Cervera-Lierta, T.~H. Kyaw, T.~Haug, S.~Alperin-Lea, A.~Anand, M.~Degroote, H.~Heimonen, J.~S. Kottmann, T.~Menke, W.-K. Mok, S.~Sim, L.-C. Kwek, and A.~Aspuru-Guzik, ``Noisy intermediate-scale quantum algorithms,'' \emph{Rev. Mod. Phys.}, vol.~94, no.~1, p. 015004, Feb. 2022. [Online]. Available: \url{https://link.aps.org/doi/10.1103/RevModPhys.94.015004}
\BIBentrySTDinterwordspacing

\bibitem{cerezo_variational_2021}
\BIBentryALTinterwordspacing
M.~Cerezo, A.~Arrasmith, R.~Babbush, S.~C. Benjamin, S.~Endo, K.~Fujii, J.~R. McClean, K.~Mitarai, X.~Yuan, L.~Cincio, and P.~J. Coles, ``\BIBforeignlanguage{en}{Variational quantum algorithms},'' \emph{\BIBforeignlanguage{en}{Nat Rev Phys}}, vol.~3, no.~9, pp. 625--644, Sep. 2021. [Online]. Available: \url{https://www.nature.com/articles/s42254-021-00348-9}
\BIBentrySTDinterwordspacing

\bibitem{cerezo_challenges_2022}
\BIBentryALTinterwordspacing
M.~Cerezo, G.~Verdon, H.-Y. Huang, L.~Cincio, and P.~J. Coles, ``\BIBforeignlanguage{en}{Challenges and opportunities in quantum machine learning},'' \emph{\BIBforeignlanguage{en}{Nat Comput Sci}}, vol.~2, no.~9, pp. 567--576, Sep. 2022. [Online]. Available: \url{https://www.nature.com/articles/s43588-022-00311-3}
\BIBentrySTDinterwordspacing

\bibitem{abbas_quantum_2023}
\BIBentryALTinterwordspacing
A.~Abbas, R.~King, H.-Y. Huang, W.~J. Huggins, R.~Movassagh, D.~Gilboa, and J.~R. McClean, ``\BIBforeignlanguage{en}{On quantum backpropagation, information reuse, and cheating measurement collapse},'' \emph{\BIBforeignlanguage{en}{arXiv.org}}, May 2023. [Online]. Available: \url{https://arxiv.org/abs/2305.13362v1}
\BIBentrySTDinterwordspacing

\bibitem{mitarai_quantum_2018}
\BIBentryALTinterwordspacing
K.~Mitarai, M.~Negoro, M.~Kitagawa, and K.~Fujii, ``Quantum circuit learning,'' \emph{Phys. Rev. A}, vol.~98, no.~3, p. 032309, Sep. 2018. [Online]. Available: \url{https://link.aps.org/doi/10.1103/PhysRevA.98.032309}
\BIBentrySTDinterwordspacing

\bibitem{crooks_gradients_2019}
\BIBentryALTinterwordspacing
G.~E. Crooks, ``Gradients of parameterized quantum gates using the parameter-shift rule and gate decomposition,'' May 2019. [Online]. Available: \url{http://arxiv.org/abs/1905.13311}
\BIBentrySTDinterwordspacing

\bibitem{vidal_calculus_2018}
\BIBentryALTinterwordspacing
J.~G. Vidal and D.~O. Theis, ``Calculus on parameterized quantum circuits,'' Dec. 2018. [Online]. Available: \url{http://arxiv.org/abs/1812.06323}
\BIBentrySTDinterwordspacing

\bibitem{schuld_evaluating_2019}
\BIBentryALTinterwordspacing
M.~Schuld, V.~Bergholm, C.~Gogolin, J.~Izaac, and N.~Killoran, ``Evaluating analytic gradients on quantum hardware,'' \emph{Phys. Rev. A}, vol.~99, no.~3, p. 032331, Mar. 2019. [Online]. Available: \url{https://link.aps.org/doi/10.1103/PhysRevA.99.032331}
\BIBentrySTDinterwordspacing

\bibitem{sweke_stochastic_2020}
\BIBentryALTinterwordspacing
R.~Sweke, F.~Wilde, J.~Meyer, M.~Schuld, P.~K. Faehrmann, B.~Meynard-Piganeau, and J.~Eisert, ``\BIBforeignlanguage{en-GB}{Stochastic gradient descent for hybrid quantum-classical optimization},'' \emph{\BIBforeignlanguage{en-GB}{Quantum}}, vol.~4, p. 314, Aug. 2020. [Online]. Available: \url{https://quantum-journal.org/papers/q-2020-08-31-314/}
\BIBentrySTDinterwordspacing

\bibitem{kyriienko_generalized_2021}
\BIBentryALTinterwordspacing
O.~Kyriienko and V.~E. Elfving, ``Generalized quantum circuit differentiation rules,'' \emph{Phys. Rev. A}, vol. 104, no.~5, p. 052417, Nov. 2021. [Online]. Available: \url{https://link.aps.org/doi/10.1103/PhysRevA.104.052417}
\BIBentrySTDinterwordspacing

\bibitem{mcclean_barren_2018}
\BIBentryALTinterwordspacing
J.~R. McClean, S.~Boixo, V.~N. Smelyanskiy, R.~Babbush, and H.~Neven, ``\BIBforeignlanguage{en}{Barren plateaus in quantum neural network training landscapes},'' \emph{\BIBforeignlanguage{en}{Nat Commun}}, vol.~9, no.~1, p. 4812, Nov. 2018. [Online]. Available: \url{https://www.nature.com/articles/s41467-018-07090-4}
\BIBentrySTDinterwordspacing

\bibitem{landman_classically_2023}
\BIBentryALTinterwordspacing
J.~Landman, S.~Thabet, C.~Dalyac, H.~Mhiri, and E.~Kashefi, ``Classically {Approximating} {Variational} {Quantum} {Machine} {Learning} with {Random} {Fourier} {Features},'' in \emph{The {Eleventh} {International} {Conference} on {Learning} {Representations}}, 2023. [Online]. Available: \url{https://openreview.net/forum?id=ymFhZxw70uz}
\BIBentrySTDinterwordspacing

\bibitem{rudolph_classical_2023}
\BIBentryALTinterwordspacing
M.~S. Rudolph, E.~Fontana, Z.~Holmes, and L.~Cincio, ``Classical surrogate simulation of quantum systems with {LOWESA},'' Aug. 2023. [Online]. Available: \url{http://arxiv.org/abs/2308.09109}
\BIBentrySTDinterwordspacing

\bibitem{bermejo_quantum_2024}
\BIBentryALTinterwordspacing
P.~Bermejo, P.~Braccia, M.~S. Rudolph, Z.~Holmes, L.~Cincio, and M.~Cerezo, ``Quantum {Convolutional} {Neural} {Networks} are ({Effectively}) {Classically} {Simulable},'' Aug. 2024. [Online]. Available: \url{http://arxiv.org/abs/2408.12739}
\BIBentrySTDinterwordspacing

\bibitem{cerezo_does_2024}
\BIBentryALTinterwordspacing
M.~Cerezo, M.~Larocca, D.~Garc√≠a-Mart√≠n, N.~L. Diaz, P.~Braccia, E.~Fontana, M.~S. Rudolph, P.~Bermejo, A.~Ijaz, S.~Thanasilp, E.~R. Anschuetz, and Z.~Holmes, ``Does provable absence of barren plateaus imply classical simulability? {Or}, why we need to rethink variational quantum computing,'' Mar. 2024. [Online]. Available: \url{http://arxiv.org/abs/2312.09121}
\BIBentrySTDinterwordspacing

\bibitem{holmes_connecting_2022}
\BIBentryALTinterwordspacing
Z.~Holmes, K.~Sharma, M.~Cerezo, and P.~J. Coles, ``Connecting {Ansatz} {Expressibility} to {Gradient} {Magnitudes} and {Barren} {Plateaus},'' \emph{PRX Quantum}, vol.~3, no.~1, p. 010313, Jan. 2022. [Online]. Available: \url{https://link.aps.org/doi/10.1103/PRXQuantum.3.010313}
\BIBentrySTDinterwordspacing

\bibitem{schuld_supervised_2021}
\BIBentryALTinterwordspacing
M.~Schuld, ``Supervised quantum machine learning models are kernel methods,'' Apr. 2021. [Online]. Available: \url{http://arxiv.org/abs/2101.11020}
\BIBentrySTDinterwordspacing

\bibitem{jerbi_quantum_2023}
\BIBentryALTinterwordspacing
S.~Jerbi, L.~J. Fiderer, H.~Poulsen~Nautrup, J.~M. K√ºbler, H.~J. Briegel, and V.~Dunjko, ``\BIBforeignlanguage{en}{Quantum machine learning beyond kernel methods},'' \emph{\BIBforeignlanguage{en}{Nat Commun}}, vol.~14, no.~1, p. 517, Jan. 2023. [Online]. Available: \url{https://www.nature.com/articles/s41467-023-36159-y}
\BIBentrySTDinterwordspacing

\bibitem{jerbi_shadows_2024}
\BIBentryALTinterwordspacing
S.~Jerbi, C.~Gyurik, S.~C. Marshall, R.~Molteni, and V.~Dunjko, ``\BIBforeignlanguage{en}{Shadows of quantum machine learning},'' \emph{\BIBforeignlanguage{en}{Nat Commun}}, vol.~15, no.~1, p. 5676, Jul. 2024. [Online]. Available: \url{https://www.nature.com/articles/s41467-024-49877-8}
\BIBentrySTDinterwordspacing

\bibitem{huang_predicting_2020}
\BIBentryALTinterwordspacing
H.-Y. Huang, R.~Kueng, and J.~Preskill, ``\BIBforeignlanguage{en}{Predicting many properties of a quantum system from very few measurements},'' \emph{\BIBforeignlanguage{en}{Nat. Phys.}}, vol.~16, no.~10, pp. 1050--1057, Oct. 2020. [Online]. Available: \url{https://www.nature.com/articles/s41567-020-0932-7}
\BIBentrySTDinterwordspacing

\bibitem{huang_post-variational_2023}
\BIBentryALTinterwordspacing
P.-W. Huang and P.~Rebentrost, ``Post-variational quantum neural networks,'' Jul. 2023. [Online]. Available: \url{http://arxiv.org/abs/2307.10560}
\BIBentrySTDinterwordspacing

\bibitem{huang_near-term_2021}
\BIBentryALTinterwordspacing
H.-Y. Huang, K.~Bharti, and P.~Rebentrost, ``\BIBforeignlanguage{en}{Near-term quantum algorithms for linear systems of equations with regression loss functions},'' \emph{\BIBforeignlanguage{en}{New J. Phys.}}, vol.~23, no.~11, p. 113021, Nov. 2021. [Online]. Available: \url{https://dx.doi.org/10.1088/1367-2630/ac325f}
\BIBentrySTDinterwordspacing

\bibitem{heredge_non-unitary_2024}
\BIBentryALTinterwordspacing
J.~Heredge, M.~West, L.~Hollenberg, and M.~Sevior, ``Non-{Unitary} {Quantum} {Machine} {Learning},'' May 2024. [Online]. Available: \url{http://arxiv.org/abs/2405.17388}
\BIBentrySTDinterwordspacing

\bibitem{brassard_quantum_2002}
\BIBentryALTinterwordspacing
G.~Brassard, P.~Hoyer, M.~Mosca, and A.~Tapp, ``Quantum {Amplitude} {Amplification} and {Estimation},'' 2002. [Online]. Available: \url{http://arxiv.org/abs/quant-ph/0005055}
\BIBentrySTDinterwordspacing

\bibitem{huggins_nearly_2022}
\BIBentryALTinterwordspacing
W.~J. Huggins, K.~Wan, J.~McClean, T.~E. O'Brien, N.~Wiebe, and R.~Babbush, ``Nearly {Optimal} {Quantum} {Algorithm} for {Estimating} {Multiple} {Expectation} {Values},'' \emph{Phys. Rev. Lett.}, vol. 129, no.~24, p. 240501, Dec. 2022. [Online]. Available: \url{https://link.aps.org/doi/10.1103/PhysRevLett.129.240501}
\BIBentrySTDinterwordspacing

\bibitem{bowles_backpropagation_2023}
\BIBentryALTinterwordspacing
J.~Bowles, D.~Wierichs, and C.-Y. Park, ``Backpropagation scaling in parameterised quantum circuits,'' Jun. 2023. [Online]. Available: \url{http://arxiv.org/abs/2306.14962}
\BIBentrySTDinterwordspacing

\bibitem{kandala_hardware-efficient_2017}
\BIBentryALTinterwordspacing
A.~Kandala, A.~Mezzacapo, K.~Temme, M.~Takita, M.~Brink, J.~M. Chow, and J.~M. Gambetta, ``\BIBforeignlanguage{en}{Hardware-efficient variational quantum eigensolver for small molecules and quantum magnets},'' \emph{\BIBforeignlanguage{en}{Nature}}, vol. 549, no. 7671, pp. 242--246, Sep. 2017. [Online]. Available: \url{https://www.nature.com/articles/nature23879}
\BIBentrySTDinterwordspacing

\bibitem{landman_quantum_2022}
\BIBentryALTinterwordspacing
J.~Landman, N.~Mathur, Y.~Y. Li, M.~Strahm, S.~Kazdaghli, A.~Prakash, and I.~Kerenidis, ``\BIBforeignlanguage{en-GB}{Quantum {Methods} for {Neural} {Networks} and {Application} to {Medical} {Image} {Classification}},'' \emph{\BIBforeignlanguage{en-GB}{Quantum}}, vol.~6, p. 881, Dec. 2022. [Online]. Available: \url{https://quantum-journal.org/papers/q-2022-12-22-881/}
\BIBentrySTDinterwordspacing

\bibitem{cherrat_quantum_2022}
\BIBentryALTinterwordspacing
E.~A. Cherrat, I.~Kerenidis, N.~Mathur, J.~Landman, M.~Strahm, and Y.~Y. Li, ``Quantum {Vision} {Transformers},'' Sep. 2022. [Online]. Available: \url{http://arxiv.org/abs/2209.08167}
\BIBentrySTDinterwordspacing

\bibitem{hamze_parallelized_2021}
\BIBentryALTinterwordspacing
F.~Hamze, ``Parallelized {Computation} and {Backpropagation} {Under} {Angle}-{Parametrized} {Orthogonal} {Matrices},'' May 2021. [Online]. Available: \url{http://arxiv.org/abs/2106.00003}
\BIBentrySTDinterwordspacing

\bibitem{hastings_turning_2017}
M.~B. Hastings, ``Turning gate synthesis errors into incoherent errors,'' \emph{Quantum Info. Comput.}, vol.~17, no. 5-6, pp. 488--494, Mar. 2017.

\bibitem{campbell_shorter_2017}
\BIBentryALTinterwordspacing
E.~Campbell, ``Shorter gate sequences for quantum computing by mixing unitaries,'' \emph{Phys. Rev. A}, vol.~95, no.~4, p. 042306, Apr. 2017. [Online]. Available: \url{https://link.aps.org/doi/10.1103/PhysRevA.95.042306}
\BIBentrySTDinterwordspacing

\bibitem{khatri_quantum-assisted_2019}
\BIBentryALTinterwordspacing
S.~Khatri, R.~LaRose, A.~Poremba, L.~Cincio, A.~T. Sornborger, and P.~J. Coles, ``\BIBforeignlanguage{en-GB}{Quantum-assisted quantum compiling},'' \emph{\BIBforeignlanguage{en-GB}{Quantum}}, vol.~3, p. 140, May 2019. [Online]. Available: \url{https://quantum-journal.org/papers/q-2019-05-13-140/}
\BIBentrySTDinterwordspacing

\bibitem{sharma_noise_2020}
\BIBentryALTinterwordspacing
K.~Sharma, S.~Khatri, M.~Cerezo, and P.~J. Coles, ``\BIBforeignlanguage{en}{Noise resilience of variational quantum compiling},'' \emph{\BIBforeignlanguage{en}{New J. Phys.}}, vol.~22, no.~4, p. 043006, Apr. 2020. [Online]. Available: \url{https://dx.doi.org/10.1088/1367-2630/ab784c}
\BIBentrySTDinterwordspacing

\bibitem{srivastava_dropout_2014}
N.~Srivastava, G.~Hinton, A.~Krizhevsky, I.~Sutskever, and R.~Salakhutdinov, ``Dropout: a simple way to prevent neural networks from overfitting,'' \emph{J. Mach. Learn. Res.}, vol.~15, no.~1, pp. 1929--1958, Jan. 2014.

\bibitem{baldi_understanding_2013}
\BIBentryALTinterwordspacing
P.~Baldi and P.~J. Sadowski, ``Understanding {Dropout},'' in \emph{Advances in {Neural} {Information} {Processing} {Systems}}, C.~J. Burges, L.~Bottou, M.~Welling, Z.~Ghahramani, and K.~Q. Weinberger, Eds., vol.~26.\hskip 1em plus 0.5em minus 0.4em\relax Curran Associates, Inc., 2013. [Online]. Available: \url{https://proceedings.neurips.cc/paper_files/paper/2013/file/71f6278d140af599e06ad9bf1ba03cb0-Paper.pdf}
\BIBentrySTDinterwordspacing

\bibitem{nguyen_theory_2022}
\BIBentryALTinterwordspacing
Q.~T. Nguyen, L.~Schatzki, P.~Braccia, M.~Ragone, P.~J. Coles, F.~Sauvage, M.~Larocca, and M.~Cerezo, ``Theory for {Equivariant} {Quantum} {Neural} {Networks},'' Oct. 2022. [Online]. Available: \url{http://arxiv.org/abs/2210.08566}
\BIBentrySTDinterwordspacing

\bibitem{cong_quantum_2019}
\BIBentryALTinterwordspacing
I.~Cong, S.~Choi, and M.~D. Lukin, ``\BIBforeignlanguage{en}{Quantum convolutional neural networks},'' \emph{\BIBforeignlanguage{en}{Nat. Phys.}}, vol.~15, no.~12, pp. 1273--1278, Dec. 2019. [Online]. Available: \url{https://www.nature.com/articles/s41567-019-0648-8}
\BIBentrySTDinterwordspacing

\bibitem{jacobs_adaptive_1991}
\BIBentryALTinterwordspacing
R.~A. Jacobs, M.~I. Jordan, S.~J. Nowlan, and G.~E. Hinton, ``Adaptive {Mixtures} of {Local} {Experts},'' \emph{Neural Computation}, vol.~3, no.~1, pp. 79--87, Mar. 1991. [Online]. Available: \url{https://ieeexplore.ieee.org/document/6797059}
\BIBentrySTDinterwordspacing

\bibitem{jordan_hierarchical_1993}
M.~Jordan and R.~Jacobs, ``Hierarchical mixtures of experts and the {EM} algorithm,'' in \emph{Proceedings of 1993 {International} {Conference} on {Neural} {Networks} ({IJCNN}-93-{Nagoya}, {Japan})}, vol.~2, 1993, pp. 1339--1344 vol.2.

\bibitem{jiang_mixtral_2024}
\BIBentryALTinterwordspacing
A.~Q. Jiang, A.~Sablayrolles, A.~Roux, A.~Mensch, B.~Savary, C.~Bamford, D.~S. Chaplot, D.~d.~l. Casas, E.~B. Hanna, F.~Bressand, G.~Lengyel, G.~Bour, G.~Lample, L.~R. Lavaud, L.~Saulnier, M.-A. Lachaux, P.~Stock, S.~Subramanian, S.~Yang, S.~Antoniak, T.~L. Scao, T.~Gervet, T.~Lavril, T.~Wang, T.~Lacroix, and W.~E. Sayed, ``Mixtral of {Experts},'' Jan. 2024. [Online]. Available: \url{http://arxiv.org/abs/2401.04088}
\BIBentrySTDinterwordspacing

\bibitem{eigen_learning_2014}
\BIBentryALTinterwordspacing
D.~Eigen, M.~Ranzato, and I.~Sutskever, ``Learning {Factored} {Representations} in a {Deep} {Mixture} of {Experts},'' Mar. 2014. [Online]. Available: \url{http://arxiv.org/abs/1312.4314}
\BIBentrySTDinterwordspacing

\bibitem{shazeer_outrageously_2017}
\BIBentryALTinterwordspacing
N.~Shazeer, A.~Mirhoseini, K.~Maziarz, A.~Davis, Q.~Le, G.~Hinton, and J.~Dean, ``Outrageously {Large} {Neural} {Networks}: {The} {Sparsely}-{Gated} {Mixture}-of-{Experts} {Layer},'' Jan. 2017. [Online]. Available: \url{http://arxiv.org/abs/1701.06538}
\BIBentrySTDinterwordspacing

\bibitem{cao_support_2003}
\BIBentryALTinterwordspacing
L.~Cao, ``Support vector machines experts for time series forecasting,'' \emph{Neurocomputing}, vol.~51, pp. 321--339, Apr. 2003. [Online]. Available: \url{https://www.sciencedirect.com/science/article/pii/S0925231202005775}
\BIBentrySTDinterwordspacing

\bibitem{lima_hybridizing_2007}
\BIBentryALTinterwordspacing
C.~A.~M. Lima, A.~L.~V. Coelho, and F.~J. Von~Zuben, ``Hybridizing mixtures of experts with support vector machines: {Investigation} into nonlinear dynamic systems identification,'' \emph{Information Sciences}, vol. 177, no.~10, pp. 2049--2074, May 2007. [Online]. Available: \url{https://www.sciencedirect.com/science/article/pii/S0020025507000382}
\BIBentrySTDinterwordspacing

\bibitem{perez-salinas_data_2020}
\BIBentryALTinterwordspacing
A.~P√©rez-Salinas, A.~Cervera-Lierta, E.~Gil-Fuster, and J.~I. Latorre, ``\BIBforeignlanguage{en-GB}{Data re-uploading for a universal quantum classifier},'' \emph{\BIBforeignlanguage{en-GB}{Quantum}}, vol.~4, p. 226, Feb. 2020. [Online]. Available: \url{https://quantum-journal.org/papers/q-2020-02-06-226/}
\BIBentrySTDinterwordspacing

\bibitem{larocca_theory_2023}
\BIBentryALTinterwordspacing
M.~Larocca, N.~Ju, D.~Garc√≠a-Mart√≠n, P.~J. Coles, and M.~Cerezo, ``\BIBforeignlanguage{en}{Theory of overparametrization in quantum neural networks},'' \emph{\BIBforeignlanguage{en}{Nat Comput Sci}}, vol.~3, no.~6, pp. 542--551, Jun. 2023. [Online]. Available: \url{https://www.nature.com/articles/s43588-023-00467-6}
\BIBentrySTDinterwordspacing

\bibitem{peters_generalization_2023}
\BIBentryALTinterwordspacing
E.~Peters and M.~Schuld, ``\BIBforeignlanguage{en-GB}{Generalization despite overfitting in quantum machine learning models},'' \emph{\BIBforeignlanguage{en-GB}{Quantum}}, vol.~7, p. 1210, Dec. 2023. [Online]. Available: \url{https://quantum-journal.org/papers/q-2023-12-20-1210/}
\BIBentrySTDinterwordspacing

\bibitem{lee-thorp_sparse_2022}
\BIBentryALTinterwordspacing
J.~Lee-Thorp and J.~Ainslie, ``Sparse {Mixers}: {Combining} {MoE} and {Mixing} to build a more efficient {BERT},'' in \emph{Findings of the {Association} for {Computational} {Linguistics}: {EMNLP} 2022}, Y.~Goldberg, Z.~Kozareva, and Y.~Zhang, Eds.\hskip 1em plus 0.5em minus 0.4em\relax Abu Dhabi, United Arab Emirates: Association for Computational Linguistics, Dec. 2022, pp. 58--75. [Online]. Available: \url{https://aclanthology.org/2022.findings-emnlp.5/}
\BIBentrySTDinterwordspacing

\bibitem{riquelme_scaling_2021}
\BIBentryALTinterwordspacing
C.~Riquelme, J.~Puigcerver, B.~Mustafa, M.~Neumann, R.~Jenatton, A.~Susano~Pinto, D.~Keysers, and N.~Houlsby, ``Scaling {Vision} with {Sparse} {Mixture} of {Experts},'' in \emph{Advances in {Neural} {Information} {Processing} {Systems}}, M.~Ranzato, A.~Beygelzimer, Y.~Dauphin, P.~S. Liang, and J.~W. Vaughan, Eds., vol.~34.\hskip 1em plus 0.5em minus 0.4em\relax Curran Associates, Inc., 2021, pp. 8583--8595. [Online]. Available: \url{https://proceedings.neurips.cc/paper_files/paper/2021/file/48237d9f2dea8c74c2a72126cf63d933-Paper.pdf}
\BIBentrySTDinterwordspacing

\bibitem{chi_representation_2022}
\BIBentryALTinterwordspacing
Z.~Chi, L.~Dong, S.~Huang, D.~Dai, S.~Ma, B.~Patra, S.~Singhal, P.~Bajaj, X.~Song, X.-L. Mao, H.~Huang, and F.~Wei, ``On the {Representation} {Collapse} of {Sparse} {Mixture} of {Experts},'' in \emph{Advances in {Neural} {Information} {Processing} {Systems}}, A.~H. Oh, A.~Agarwal, D.~Belgrave, and K.~Cho, Eds., 2022. [Online]. Available: \url{https://openreview.net/forum?id=mWaYC6CZf5}
\BIBentrySTDinterwordspacing

\bibitem{larocca_diagnosing_2022}
\BIBentryALTinterwordspacing
M.~Larocca, P.~Czarnik, K.~Sharma, G.~Muraleedharan, P.~J. Coles, and M.~Cerezo, ``\BIBforeignlanguage{en-GB}{Diagnosing {Barren} {Plateaus} with {Tools} from {Quantum} {Optimal} {Control}},'' \emph{\BIBforeignlanguage{en-GB}{Quantum}}, vol.~6, p. 824, Sep. 2022. [Online]. Available: \url{https://quantum-journal.org/papers/q-2022-09-29-824/}
\BIBentrySTDinterwordspacing

\bibitem{cherrat_quantum_2023}
\BIBentryALTinterwordspacing
E.~A. Cherrat, S.~Raj, I.~Kerenidis, A.~Shekhar, B.~Wood, J.~Dee, S.~Chakrabarti, R.~Chen, D.~Herman, S.~Hu, P.~Minssen, R.~Shaydulin, Y.~Sun, R.~Yalovetzky, and M.~Pistoia, ``\BIBforeignlanguage{en-GB}{Quantum {Deep} {Hedging}},'' \emph{\BIBforeignlanguage{en-GB}{Quantum}}, vol.~7, p. 1191, Nov. 2023. [Online]. Available: \url{https://quantum-journal.org/papers/q-2023-11-29-1191/}
\BIBentrySTDinterwordspacing

\bibitem{monbroussou_trainability_2023}
\BIBentryALTinterwordspacing
L.~Monbroussou, J.~Landman, A.~B. Grilo, R.~Kukla, and E.~Kashefi, ``Trainability and {Expressivity} of {Hamming}-{Weight} {Preserving} {Quantum} {Circuits} for {Machine} {Learning},'' Sep. 2023. [Online]. Available: \url{http://arxiv.org/abs/2309.15547}
\BIBentrySTDinterwordspacing

\bibitem{kiani_projunn_2022}
\BIBentryALTinterwordspacing
B.~Kiani, R.~Balestriero, Y.~LeCun, and S.~Lloyd, ``{projUNN}: efficient method for training deep networks with unitary matrices,'' Oct. 2022. [Online]. Available: \url{http://arxiv.org/abs/2203.05483}
\BIBentrySTDinterwordspacing

\bibitem{schuld_effect_2021}
\BIBentryALTinterwordspacing
M.~Schuld, R.~Sweke, and J.~J. Meyer, ``Effect of data encoding on the expressive power of variational quantum-machine-learning models,'' \emph{Phys. Rev. A}, vol. 103, no.~3, p. 032430, Mar. 2021. [Online]. Available: \url{https://link.aps.org/doi/10.1103/PhysRevA.103.032430}
\BIBentrySTDinterwordspacing

\bibitem{mhiri_constrained_2024}
\BIBentryALTinterwordspacing
H.~Mhiri, L.~Monbroussou, M.~Herrero-Gonzalez, S.~Thabet, E.~Kashefi, and J.~Landman, ``Constrained and {Vanishing} {Expressivity} of {Quantum} {Fourier} {Models},'' Mar. 2024. [Online]. Available: \url{http://arxiv.org/abs/2403.09417}
\BIBentrySTDinterwordspacing

\bibitem{bowles_josephbowlesbackprop_scaling_2023}
\BIBentryALTinterwordspacing
J.~Bowles, ``josephbowles/backprop\_scaling,'' Jun. 2023. [Online]. Available: \url{https://github.com/josephbowles/backprop_scaling}
\BIBentrySTDinterwordspacing

\bibitem{bergholm_pennylane_2022}
\BIBentryALTinterwordspacing
V.~Bergholm, J.~Izaac, M.~Schuld, C.~Gogolin, S.~Ahmed, V.~Ajith, M.~S. Alam, G.~Alonso-Linaje, B.~AkashNarayanan, A.~Asadi, J.~M. Arrazola, U.~Azad, S.~Banning, C.~Blank, T.~R. Bromley, B.~A. Cordier, J.~Ceroni, A.~Delgado, O.~Di~Matteo, A.~Dusko, T.~Garg, D.~Guala, A.~Hayes, R.~Hill, A.~Ijaz, T.~Isacsson, D.~Ittah, S.~Jahangiri, P.~Jain, E.~Jiang, A.~Khandelwal, K.~Kottmann, R.~A. Lang, C.~Lee, T.~Loke, A.~Lowe, K.~McKiernan, J.~J. Meyer, J.~A. Monta√±ez-Barrera, R.~Moyard, Z.~Niu, L.~J. O'Riordan, S.~Oud, A.~Panigrahi, C.-Y. Park, D.~Polatajko, N.~Quesada, C.~Roberts, N.~S√°, I.~Schoch, B.~Shi, S.~Shu, S.~Sim, A.~Singh, I.~Strandberg, J.~Soni, A.~Sz√°va, S.~Thabet, R.~A. Vargas-Hern√°ndez, T.~Vincent, N.~Vitucci, M.~Weber, D.~Wierichs, R.~Wiersema, M.~Willmann, V.~Wong, S.~Zhang, and N.~Killoran, ``{PennyLane}: {Automatic} differentiation of hybrid quantum-classical computations,'' Jul. 2022. [Online]. Available: \url{http://arxiv.org/abs/1811.04968}
\BIBentrySTDinterwordspacing

\bibitem{bradbury_jax_2018}
\BIBentryALTinterwordspacing
J.~Bradbury, R.~Frostig, P.~Hawkins, M.~J. Johnson, C.~Leary, D.~Maclaurin, G.~Necula, A.~Paszke, J.~VanderPlas, S.~Wanderman-Milne, and Q.~Zhang, ``{JAX}: composable transformations of {Python}+{NumPy} programs,'' 2018. [Online]. Available: \url{http://github.com/google/jax}
\BIBentrySTDinterwordspacing

\bibitem{caro_generalization_2022}
\BIBentryALTinterwordspacing
M.~C. Caro, H.-Y. Huang, M.~Cerezo, K.~Sharma, A.~Sornborger, L.~Cincio, and P.~J. Coles, ``\BIBforeignlanguage{en}{Generalization in quantum machine learning from few training data},'' \emph{\BIBforeignlanguage{en}{Nat Commun}}, vol.~13, no.~1, p. 4919, Aug. 2022. [Online]. Available: \url{https://www.nature.com/articles/s41467-022-32550-3}
\BIBentrySTDinterwordspacing

\bibitem{akiba_optuna_2019}
T.~Akiba, S.~Sano, T.~Yanase, T.~Ohta, and M.~Koyama, ``Optuna: {A} {Next}-{Generation} {Hyperparameter} {Optimization} {Framework},'' in \emph{The 25th {ACM} {SIGKDD} {International} {Conference} on {Knowledge} {Discovery} \& {Data} {Mining}}, 2019, pp. 2623--2631.

\bibitem{anselmetti_local_2021}
\BIBentryALTinterwordspacing
G.-L.~R. Anselmetti, D.~Wierichs, C.~Gogolin, and R.~M. Parrish, ``\BIBforeignlanguage{en}{Local, expressive, quantum-number-preserving {VQE} ans√§tze for fermionic systems},'' \emph{\BIBforeignlanguage{en}{New J. Phys.}}, vol.~23, no.~11, p. 113010, Nov. 2021. [Online]. Available: \url{https://dx.doi.org/10.1088/1367-2630/ac2cb3}
\BIBentrySTDinterwordspacing

\bibitem{kerenidis_quantum_2022}
\BIBentryALTinterwordspacing
I.~Kerenidis and A.~Prakash, ``Quantum machine learning with subspace states,'' Feb. 2022. [Online]. Available: \url{http://arxiv.org/abs/2202.00054}
\BIBentrySTDinterwordspacing

\bibitem{kazdaghli_improved_2023}
\BIBentryALTinterwordspacing
S.~Kazdaghli, I.~Kerenidis, J.~Kieckbusch, and P.~Teare, ``Improved clinical data imputation via classical and quantum determinantal point processes,'' Dec. 2023. [Online]. Available: \url{http://arxiv.org/abs/2303.17893}
\BIBentrySTDinterwordspacing

\bibitem{thakkar_improved_2024}
\BIBentryALTinterwordspacing
S.~Thakkar, S.~Kazdaghli, N.~Mathur, I.~Kerenidis, A.~J. Ferreira-Martins, and S.~Brito, ``Improved {Financial} {Forecasting} via {Quantum} {Machine} {Learning},'' Apr. 2024. [Online]. Available: \url{http://arxiv.org/abs/2306.12965}
\BIBentrySTDinterwordspacing

\bibitem{kerenidis_quantum_2020}
\BIBentryALTinterwordspacing
I.~Kerenidis, J.~Landman, and A.~Prakash, ``Quantum {Algorithms} for {Deep} {Convolutional} {Neural} {Networks},'' in \emph{International {Conference} on {Learning} {Representations}}, 2020. [Online]. Available: \url{http://arxiv.org/abs/1911.01117}
\BIBentrySTDinterwordspacing

\bibitem{collobert_parallel_2001}
\BIBentryALTinterwordspacing
R.~Collobert, S.~Bengio, and Y.~Bengio, ``A {Parallel} {Mixture} of {SVMs} for {Very} {Large} {Scale} {Problems},'' in \emph{Advances in {Neural} {Information} {Processing} {Systems}}, T.~Dietterich, S.~Becker, and Z.~Ghahramani, Eds., vol.~14.\hskip 1em plus 0.5em minus 0.4em\relax MIT Press, 2001. [Online]. Available: \url{https://proceedings.neurips.cc/paper_files/paper/2001/file/36ac8e558ac7690b6f44e2cb5ef93322-Paper.pdf}
\BIBentrySTDinterwordspacing

\bibitem{theis_generative_2015}
\BIBentryALTinterwordspacing
L.~Theis and M.~Bethge, ``Generative {Image} {Modeling} {Using} {Spatial} {LSTMs},'' Sep. 2015. [Online]. Available: \url{http://arxiv.org/abs/1506.03478}
\BIBentrySTDinterwordspacing

\bibitem{deisenroth_distributed_2015}
\BIBentryALTinterwordspacing
M.~Deisenroth and J.~W. Ng, ``Distributed {Gaussian} {Processes},'' in \emph{Proceedings of the 32nd {International} {Conference} on {Machine} {Learning}}, ser. Proceedings of {Machine} {Learning} {Research}, F.~Bach and D.~Blei, Eds., vol.~37.\hskip 1em plus 0.5em minus 0.4em\relax Lille, France: PMLR, Jul. 2015, pp. 1481--1490. [Online]. Available: \url{https://proceedings.mlr.press/v37/deisenroth15.html}
\BIBentrySTDinterwordspacing

\bibitem{verdon_learning_2019}
\BIBentryALTinterwordspacing
G.~Verdon, M.~Broughton, J.~R. McClean, K.~J. Sung, R.~Babbush, Z.~Jiang, H.~Neven, and M.~Mohseni, ``Learning to learn with quantum neural networks via classical neural networks,'' Jul. 2019. [Online]. Available: \url{http://arxiv.org/abs/1907.05415}
\BIBentrySTDinterwordspacing

\bibitem{wilson_optimizing_2021}
\BIBentryALTinterwordspacing
M.~Wilson, R.~Stromswold, F.~Wudarski, S.~Hadfield, N.~M. Tubman, and E.~G. Rieffel, ``\BIBforeignlanguage{en}{Optimizing quantum heuristics with meta-learning},'' \emph{\BIBforeignlanguage{en}{Quantum Mach. Intell.}}, vol.~3, no.~1, p.~13, Apr. 2021. [Online]. Available: \url{https://doi.org/10.1007/s42484-020-00022-w}
\BIBentrySTDinterwordspacing

\bibitem{scala_general_2023}
\BIBentryALTinterwordspacing
F.~Scala, A.~Ceschini, M.~Panella, and D.~Gerace, ``A {General} {Approach} to {Dropout} in {Quantum} {Neural} {Networks},'' \emph{Adv Quantum Tech}, p. 2300220, Dec. 2023. [Online]. Available: \url{http://arxiv.org/abs/2310.04120}
\BIBentrySTDinterwordspacing

\bibitem{kobayashi_overfitting_2022}
\BIBentryALTinterwordspacing
M.~Kobayashi, K.~Nakaji, and N.~Yamamoto, ``\BIBforeignlanguage{en}{Overfitting in quantum machine learning and entangling dropout},'' \emph{\BIBforeignlanguage{en}{Quantum Mach. Intell.}}, vol.~4, no.~2, p.~30, Nov. 2022. [Online]. Available: \url{https://doi.org/10.1007/s42484-022-00087-9}
\BIBentrySTDinterwordspacing

\bibitem{cheng_information_2018}
\BIBentryALTinterwordspacing
S.~Cheng, J.~Chen, and L.~Wang, ``\BIBforeignlanguage{en}{Information {Perspective} to {Probabilistic} {Modeling}: {Boltzmann} {Machines} versus {Born} {Machines}},'' \emph{\BIBforeignlanguage{en}{Entropy}}, vol.~20, no.~8, p. 583, Aug. 2018. [Online]. Available: \url{https://www.mdpi.com/1099-4300/20/8/583}
\BIBentrySTDinterwordspacing

\bibitem{liu_differentiable_2018}
\BIBentryALTinterwordspacing
J.-G. Liu and L.~Wang, ``Differentiable learning of quantum circuit {Born} machines,'' \emph{Phys. Rev. A}, vol.~98, no.~6, p. 062324, Dec. 2018. [Online]. Available: \url{https://link.aps.org/doi/10.1103/PhysRevA.98.062324}
\BIBentrySTDinterwordspacing

\bibitem{benedetti_generative_2019}
\BIBentryALTinterwordspacing
M.~Benedetti, D.~Garcia-Pintos, O.~Perdomo, V.~Leyton-Ortega, Y.~Nam, and A.~Perdomo-Ortiz, ``\BIBforeignlanguage{en}{A generative modeling approach for benchmarking and training shallow quantum circuits},'' \emph{\BIBforeignlanguage{en}{npj Quantum Inf}}, vol.~5, no.~1, pp. 1--9, May 2019. [Online]. Available: \url{https://www.nature.com/articles/s41534-019-0157-8}
\BIBentrySTDinterwordspacing

\bibitem{coyle_born_2020}
\BIBentryALTinterwordspacing
B.~Coyle, D.~Mills, V.~Danos, and E.~Kashefi, ``\BIBforeignlanguage{en}{The {Born} supremacy: quantum advantage and training of an {Ising} {Born} machine},'' \emph{\BIBforeignlanguage{en}{npj Quantum Inf}}, vol.~6, no.~1, pp. 1--11, Jul. 2020. [Online]. Available: \url{https://www.nature.com/articles/s41534-020-00288-9}
\BIBentrySTDinterwordspacing

\bibitem{jain_quantum_2024}
\BIBentryALTinterwordspacing
N.~Jain, J.~Landman, N.~Mathur, and I.~Kerenidis, ``\BIBforeignlanguage{en}{Quantum {Fourier} networks for solving parametric {PDEs}},'' \emph{\BIBforeignlanguage{en}{Quantum Sci. Technol.}}, vol.~9, no.~3, p. 035026, May 2024. [Online]. Available: \url{https://dx.doi.org/10.1088/2058-9565/ad42ce}
\BIBentrySTDinterwordspacing

\bibitem{majumder_variational_2023}
\BIBentryALTinterwordspacing
A.~Majumder, M.~Krumm, T.~Radkohl, H.~P. Nautrup, S.~Jerbi, and H.~J. Briegel, ``Variational measurement-based quantum computation for generative modeling,'' Oct. 2023. [Online]. Available: \url{http://arxiv.org/abs/2310.13524}
\BIBentrySTDinterwordspacing

\end{thebibliography}
\end{document}